%% file: main.tex
\documentclass[acmsmall]{acmart}
\usepackage[english]{babel}
\usepackage[utf8]{inputenc}
\usepackage{amsmath,bm,amsfonts,amsthm}
\usepackage{physics}
\usepackage{algorithmic}
\usepackage{graphicx}
\usepackage{textcomp}
\usepackage{mathtools}
\usepackage{enumitem}
\usepackage{xspace}
\usepackage{pifont}
\usepackage{bbm}
\usepackage{dsfont}
\usepackage{graphicx}
\usepackage{grffile}
\usepackage{pdflscape}
\usepackage{breqn}
\usepackage{amsthm}
\usepackage{color}
\usepackage{hyperref}
\usepackage{subcaption}
\usepackage{tikz}
\usetikzlibrary{calc,arrows.meta}
\newtheorem{theorem}{Theorem}
\newtheorem{corollary}{Corollary}

\newtheorem{definition}{Definition}
\newtheorem{remark}{Remark}

\newcommand{\eg}{\textit{e.g.}}
\newcommand{\ie}{\textit{i.e.}}

\newcommand*\chancery{\fontfamily{pzc}\selectfont}
\allowdisplaybreaks
\graphicspath{{./figs/}}

\newcommand{\gv}[1]{\textcolor{magenta}{#1}}
\newcommand{\sg}[1]{\textcolor{orange}{#1}}

\AtBeginDocument{%
  }

\begin{document}

%\title{Performance of on-demand scheduling for a quantum network hub}
\title{An on-demand resource allocation algorithm for a quantum network hub and its performance analysis}
\renewcommand{\shorttitle}{An on-demand resource allocation algorithm for a quantum network hub and its performance analysis}


\author{Scarlett Gauthier}
% \if{false}
\email{s.s.gauthier@tudelft.nl}
\affiliation{%
  \institution{EEMCS and QuTech, Delft University of Technology}
  % \institution{Quantum Computer Science, Electrical Engineering, Mathematics and Computer Science, Delft University of Technology}
  % \streetaddress{Mekelweg 5}
  \city{Delft}
  \country{The Netherlands}
  % \postcode{2628 CD}
}

\author{Thirupathaiah Vasantam}
% \authornotemark[2]
\email{thirupathaiah.vasantam@durham.ac.uk}
\affiliation{%
  \institution{Department of Computer Science, Durham University}
  % \streetaddress{The Palatine Centre, Stockton Rd, Durham, United Kingdom}
  \city{Durham}
  \country{United Kingdom}}
  % \postcode{DH1 3LE}

\author{Gayane Vardoyan}
% \authornotemark[3]
\email{g.s.vardoyan@tudelft.nl}
\affiliation{%
    \institution{EEMCS and QuTech, Delft University of Technology}
    \city{Delft}
    \country{The Netherlands}
}
\additionalaffiliation{%
  \institution{University of Massachusetts, Amherst}
  \city{Amherst}
  \country{United States}
}
% \fi

\renewcommand{\shortauthors}{Scarlett Gauthier, Thirupathaiah Vasantam, Gayane Vardoyan}

\begin{abstract}
 %A quantum network must allocate shared resources effectively to support many sets of user-controlled quantum nodes executing quantum network applications.    
 To effectively support the execution of quantum network applications for multiple sets of user-controlled quantum nodes, a quantum network must efficiently allocate shared resources.
 We study traffic models for a type of quantum network hub called an Entanglement Generation Switch (EGS), a device that allocates resources to enable entanglement generation between nodes in response to user-generated demand. We propose an on-demand resource allocation algorithm, where a demand is either blocked if no resources are available or else results in immediate resource allocation. We model the EGS as an Erlang loss system, with demands corresponding to sessions whose arrival is modelled as a Poisson process. To reflect the operation of a practical quantum switch, our model captures scenarios where a resource is allocated for batches of entanglement generation attempts, possibly interleaved with calibration periods for the quantum network nodes. 
 Calibration periods are necessary to correct against drifts or jumps in the physical parameters of a quantum node that occur on a timescale that is long compared to the duration of an attempt. We then derive a formula for the demand blocking probability under three different traffic scenarios using analytical methods from applied probability and queueing theory.  We prove an insensitivity theorem which guarantees that the probability a demand is blocked only depends upon the mean duration of each entanglement generation attempt and calibration period, and is not sensitive to the underlying distributions of attempt and calibration period duration. We provide numerical results to support our analysis. 
 Our numerical results suggest that in homogeneous traffic scenarios, enhancing the communication multiplexing ability of a limited quantum node by increasing the number of communication qubits from one to two has a large impact on the blocking probability. However, further increases in the number of communication qubits have a limited effect. Our work is the first analysis of traffic characteristics at an EGS system and provides a valuable analytic tool for devising performance driven resource allocation algorithms.
  
\end{abstract}

%    <concept>
%        <concept_id>10003033.10003079.10011672</concept_id>
%        <concept_desc>Networks~Network performance analysis</concept_desc>
%        <concept_significance>500</concept_significance>
%        </concept>
%    <concept>
%        <concept_id>10010147.10010341</concept_id>
%        <concept_desc>Computing methodologies~Modeling and simulation</concept_desc>
%        <concept_significance>500</concept_significance>
%        </concept>
%    <concept>
%        <concept_id>10010583.10010786.10010813</concept_id>
%        <concept_desc>Hardware~Quantum technologies</concept_desc>
%        <concept_significance>500</concept_significance>
%        </concept>
%  </ccs2012>

% \ccsdesc[500]{Networks~Network performance analysis}
% \ccsdesc[500]{Computing methodologies~Modeling and simulation}
% \ccsdesc[500]{Hardware~Quantum technologies}
% \keywords{quantum networking, entanglement, queueing theory}


\maketitle

\section{Introduction}
\begin{figure}[t]
    \centering
    \begin{minipage}[c]{0.44\textwidth}
    \includegraphics[width=\textwidth]
    {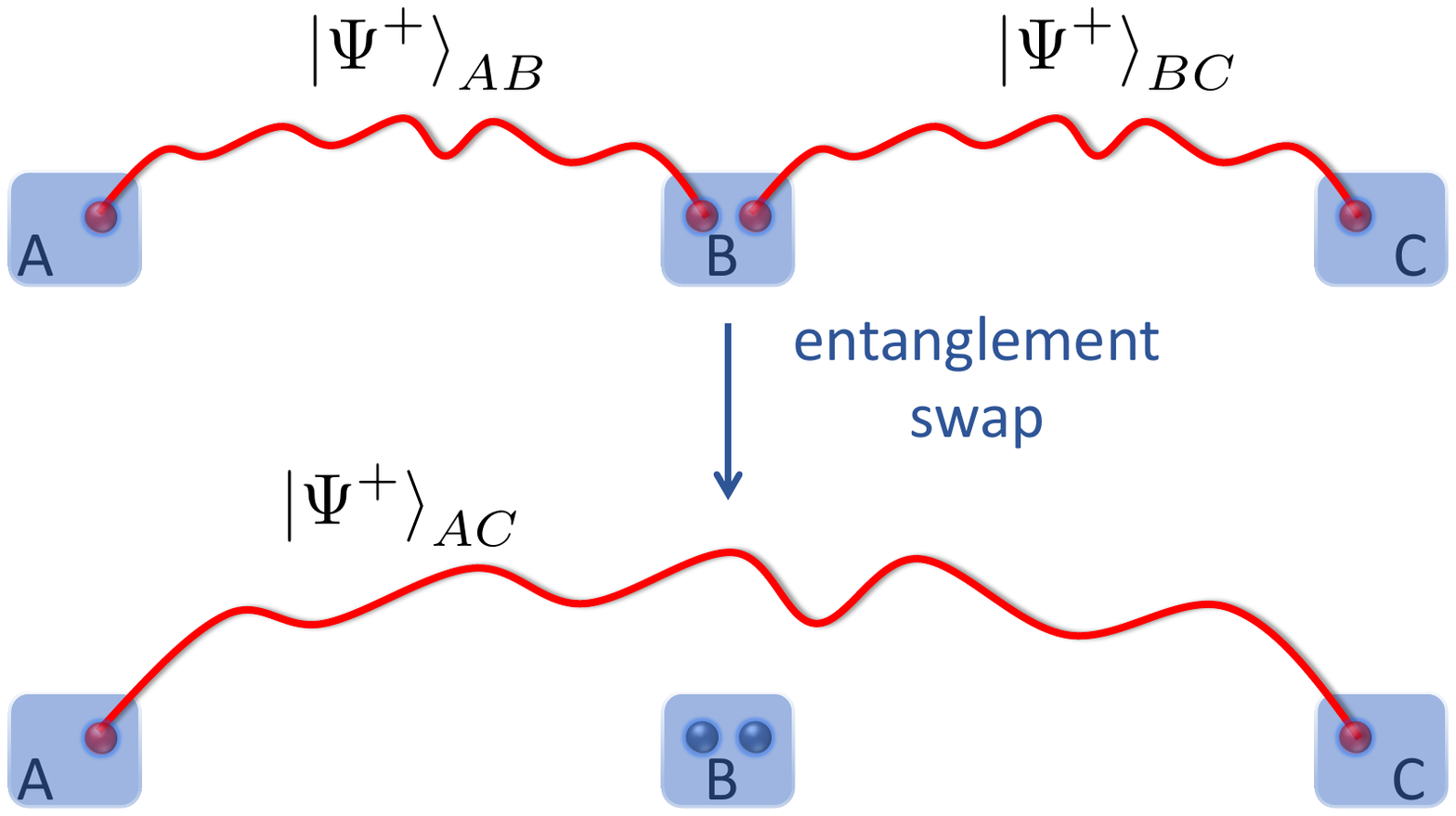}
    \end{minipage}
    \begin{minipage}[c]{0.55\textwidth}
    \caption{A simple quantum network with end nodes $A$ and $C$ wishing to share entanglement, and an intermediate node $B$ assisting them with the task. Initially, two entangled links -- $\ket{\Psi^+}_{AB}$ between $A$ and $B$ and $\ket{\Psi^+}_{BC}$ between $B$ and $C$ -- are established. $B$ then performs a swapping operation to directly entangle $A$ and $C$'s qubits. Depending on the distance between $A$ and $C$, direct generation of entanglement (without an intermediate node), may not be feasible.}
    \label{fig:swap}
    \end{minipage}
    \vspace{-10mm}
\end{figure}
%\gv{
%\begin{itemize}
    %\item General discussion of quantum networking -- this is mostly a classical audience. 
    %\item Introduce the EGS switch at a high level.
    %\item Why study the EGS switch?
    %\item What makes this switch study different from others? Focus on the physical relevance of our model, both in the hardware design we study, as well as the theoretical session-based analysis.
    %\item Quantum switching background and previous analytical studies -- after some thought, I feel it would be better to put this in the background section before or after the Erlang discussion
    %\item Bulletpoint contributions, etc.
    %\item Be sure to specify what makes our work different: first work looking at specific traffic patterns and service modes for an all-optical switch. Our model can also be applied to a variety of hardware settings (with/without memories at nodes, etc.) -- useful for NISQ.
%\end{itemize}}
Quantum networks enable a variety of distributed applications that are not realizable via classical means alone. Among these are quantum key distribution (QKD) \cite{bb84, e91}, blind quantum computation (BQC) \cite{bqc1, bqc2, vbqc}, and several entanglement-based quantum sensing techniques \cite{qVLBI, clockSynchPaper,giovannetti2001quantum,guo2020distributed}. A quantum network consists of end nodes equipped with quantum hardware, as well as intermediate nodes -- quantum repeaters or switches -- whose main function is to enable the end nodes to carry out quantum communication tasks. In first- and second-generation quantum networks \cite{munro2015inside}, these intermediate nodes use methods such as entanglement swapping as shown in Figure~\ref{fig:swap} to provide entanglement as a resource to user-run applications -- either to be consumed directly, or used in teleportation-based transport of quantum information \cite{bennett1993teleporting,bouwmeester1997experimental}. When multiple applications have simultaneous demand for shared and limited resources (e.g., quantum memories, links, measurement modules), contention may arise, and the network must enact a resource allocation scheme. 

\begin{figure}[t]
    \centering
     \begin{minipage}[c]{0.44\textwidth}
    \includegraphics[width=\textwidth]{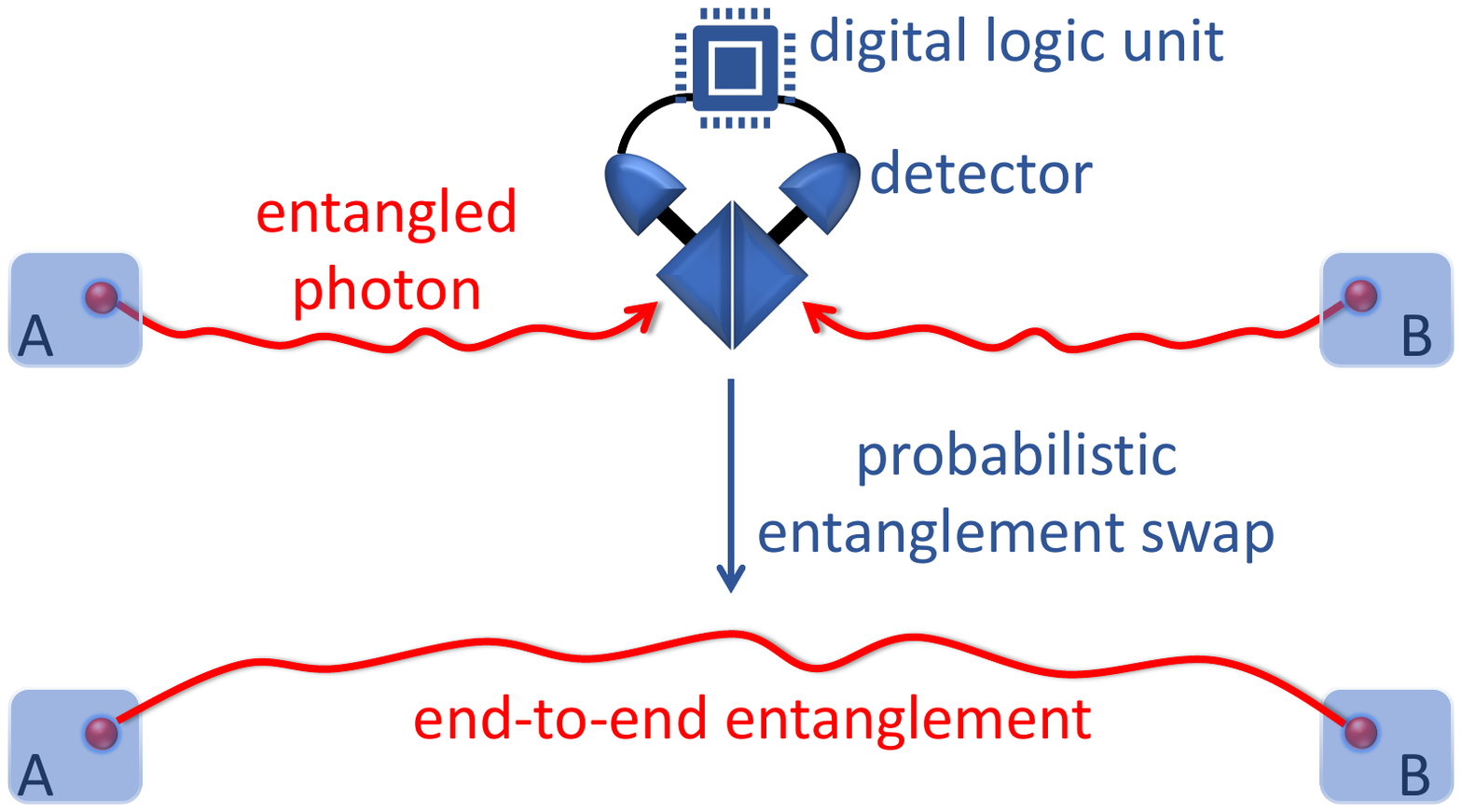}
    \end{minipage}
    \begin{minipage}[c]{0.55\textwidth}
    \caption{Entanglement generation with a Bell state analyzer. Entangled photons travel towards the BSA, which carries out a probabilistic entanglement swap. At the BSA, photons pass through a beam splitter whose output ports are connected to a pair of detectors. The digital logic unit reads out measurement results and determines if a swap was successful. Results are classically communicated to nodes $A$ and $B$, whose qubits become entangled upon a successful event.}
    \label{fig:optical_swap_BSA}
    \end{minipage}
    \vspace{-8mm}
\end{figure}
In this work, we study a type of quantum network hub previously referred to as an Entanglement Generation Switch (EGS) \cite{RCP_EGS}. An EGS is a type of quantum switch with control over a pool of resources which, when allocated to a set of nodes, enable a task such as entanglement generation. Unlike its memory-equipped counterpart (sometimes referred to as an entanglement distribution switch, or EDS), the EGS is relatively easy to fabricate since it has no memories: it possesses only a number of resources such as Bell state analyzers (BSAs) \cite{braunstein1995measurement,michler1996interferometric,walther2005experimental}, which serve as a means of performing probabilistic optical entanglement swapping on incoming photons (each entangled with a qubit at an end node), and upon a successful swap generating end-to-end entanglement. This method of generating entanglement is illustrated in Figure~\ref{fig:optical_swap_BSA}, and the process is explained in more detail in Section~\ref{sec:system_descr}. In principle, an EGS can serve any number of nodes with a single shared BSA, but more BSAs can ameliorate contention for this resource; see Figure~\ref{fig:star_topology} for an example of an EGS with three BSAs to service four nodes. Some EDS proposals on the other hand additionally place BSAs in the middle of each physical link that connects the device to other nodes in the network.
%(\textcolor{red}{I think we are talking about EDS not EGS here. This whole paragraph is not clear enough somehow as EDS and EGS statements are mixed. We need to first add a sentence that one BSA is needed to create switch-node entanglements across each link.' Explain why one extra BSA is placed at the centre in figure 2 eventhough there are four users. Also elaborate more on functionality/use of BSAs.}) 
While these BSAs assist with entanglement generation at the link level, the resulting architecture is fairly demanding in the number and type of hardware components: $K$ links translate into $K$ dedicated BSAs and at least an equal number of quantum memories at the EDS. In contrast, the EGS architecture places all BSAs at the central hub along with a switching fabric for reconfiguration so as to serve any set of end nodes, resource limitations permitting. 

\begin{figure}[t!]
    \centering
    \begin{minipage}{0.48\textwidth}
    \centering
    \includegraphics[width=\textwidth,trim={0 0 3cm 0},clip]{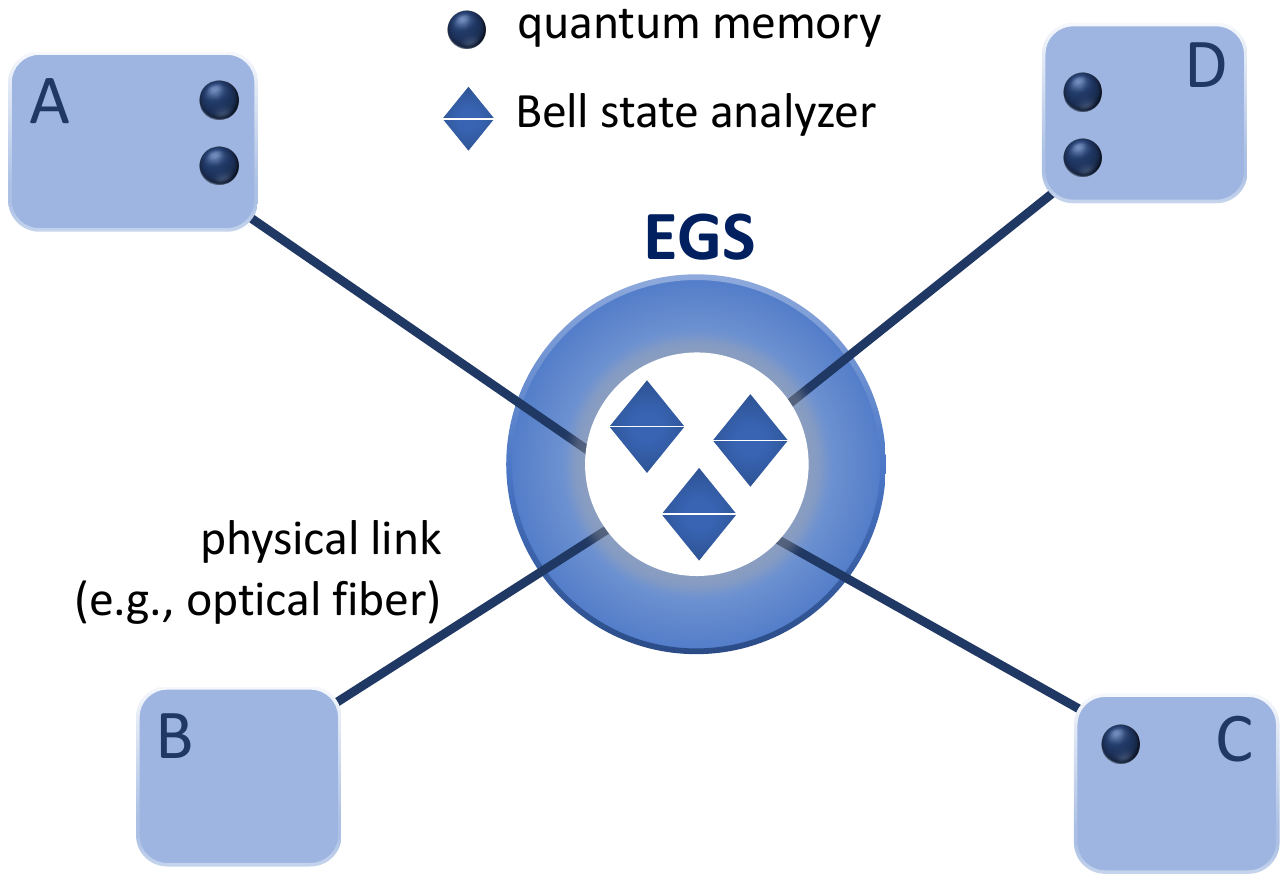}
    \vspace{-5mm}
    \caption{End nodes being serviced by an EGS with a pool of three BSAs.}
    \label{fig:star_topology}
    \end{minipage}\hfill
    \begin{minipage}{0.49\textwidth}
    \centering
    \includegraphics[width=\textwidth]{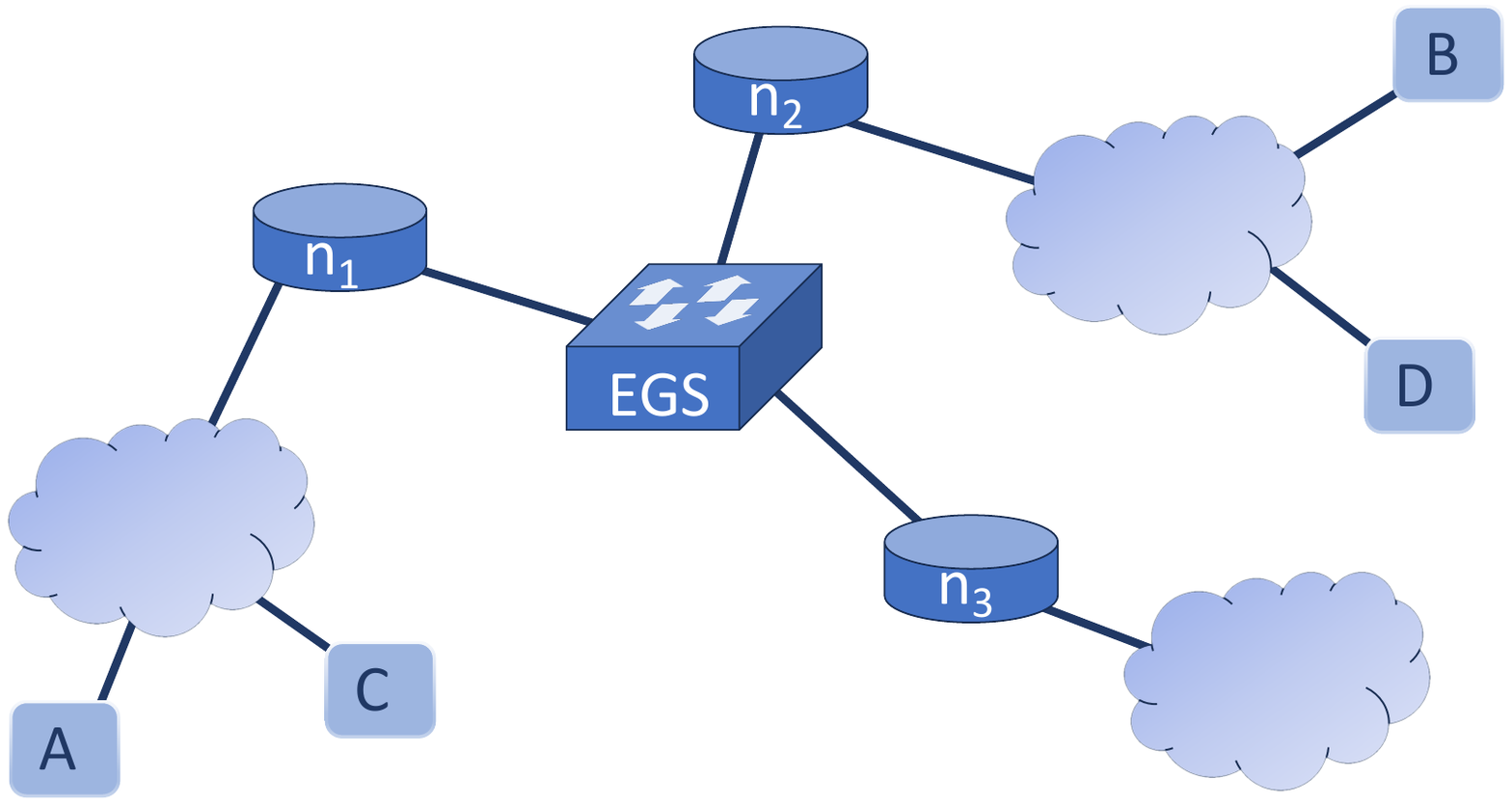}
    \vspace{-5mm}
    \caption{Intermediate quantum network nodes leveraging the EGS to interconnect their respective quantum local area networks.}
    \label{fig:larger_net}
    \end{minipage}
    \vspace{-5mm}
\end{figure}
With these properties, the EGS is poised to be an excellent candidate for a scalable and straightforwardly implementable quantum network component, especially in the Noisy Intermediate Scale Quantum (NISQ) era \cite{NISQ}. While the EGS can be used to directly connect end nodes, as shown in Figure~\ref{fig:star_topology}, it can also provide entanglement to other intermediate quantum network nodes, e.g., quantum repeaters/switches equipped with quantum memories of sufficiently long coherence time, each servicing a quantum local area network, as shown in Figure~\ref{fig:larger_net}. The versatility of the EGS warrants investigation into its practical operation within a quantum network; we provide a detailed explanation of this in Section~\ref{sec:system_descr}. We then model the hub as an Erlang loss system, with the EGS acting as a server and the nodes attached to it generating entanglement requests. We assume that these entanglement requests arrive as sessions according to a Poisson process, where each session consists of multiple entanglement generation attempts.
We then analyze three operational modes of the device, while heeding physical capabilities and limitations of both the EGS and the nodes connected to it. Namely, inspired by realistic expectations of hardware characteristics in the near term, we equip the system with two important and pragmatic features: $(i)$ batching of entanglement generation attempts due to generally high rates of failure; and $(ii)$ provisioning for calibration periods necessitated by the quantum communication qubits of nodes that are served by the EGS. In the first operational mode scenario we model strict resource reservation, where an accepted demand is scheduled for several batches of entanglement generation interleaved with calibration periods and does not release the resource until all attempts are complete or entanglement generation succeeds, whatever happens first. The second scenario we study is a variation of the first in that a successful attempt does not trigger release of the EGS resource; all attempts are carried out, thus possibly resulting in multiple successfully-established entangled states. In the final scenario, nodes relinquish resources during calibration periods, and must re-obtain them afterwards -- we opt for jump-over blocking to describe the retrial behavior of this system. In all scenarios, an incoming request is blocked (dropped) if upon arrival it sees no free resource at the switch. While studying the EGS in the context of these different operation modes, we make the following contributions:
\begin{itemize}[noitemsep, topsep=0pt]
    \item[--] We provide a comprehensive description of EGS operational details, with system specifications rooted in practical considerations of the underlying physical architecture;
    \item[--] We model the EGS as an Erlang loss system with sessions of demands arriving as a Poisson process, which we subsequently analyze to obtain $(i)$ the stationary distribution of the number of active requests being served at the switch; $(ii)$ request blocking probabilities; and $(iii)$ an insensitivity result that highlights the broad applicability of our model to practical systems;
    \item[--] We develop an extensive simulation framework capable of enacting sequences of events that model the operation of a real EGS -- one that operates in discrete time -- in a variety of configurations. Simulation code will be made available to enable future studies of the EGS.
    \if{false} %commented out by Gayane since a lot of the content below is not contributions
    We observe agreement between simulated blocking probabilities and numerical evaluation of our analysis and we study the effect of constraints such as the number of EGS resources and the number of communication qubits available to nodes. We have also observed that numerical results that we obtain for a discrete-time switch model matches closely with that of our continuous-time switch model with demand sessions arriving as a Poisson process. In practice, a quantum switch might operate in a discrete-time setting in some situations. \textcolor{blue}{thiru: I have added a sentence for comparing results of discrete and continuous time systems. Please modify it if needed}.
    \fi
\end{itemize}
%We observe agreement between simulated blocking probabilities and numerical evaluation of our analysis, and we study the effect of constraints such as the number of EGS resources and communication qubits available to nodes. 
The numerical results obtained from our continuous-time switch model closely match the more realistic discrete-time simulation of the switch.
The physical relevance of our model, both in the hardware design we consider, as well as system control protocols we propose, set this work apart from much of the previous literature, wherein hardware limitations are frequently understated. The wide scope of our framework moreover enables one to model arbitrary traffic patterns and a wide variety of hardware settings, including ones where nodes have multiple communication qubits.

The rest of this manuscript is organized as follows: in Section~\ref{sec:background}, we provide the relevant queueing-theoretic and quantum switching background. In Section~\ref{sec:system_descr}, we outline the system description, including physically-motivated operation settings. In Section~\ref{sec:model_assumptions}, we introduce the model of the EGS and state our assumptions. Section~\ref{sec:analysis} presents the analysis, while Section~\ref{sec:NumEval} provides a numerical evaluation of the system. We make concluding remarks in Section~\ref{sec:conclusion}.

\section{Background}\label{sec:background}
% \gv{Subsections for quantum and classical. For classical, introduce Erlang model and Bonald's work. Define insensitivity, give formula, refer to books.}
% \gv{What is entanglement, how it can be used, different types of applications, and which ones we focus on here (e.g., prepare-and-measure?). Quantum switches are a core component -- really worth studying thoroughly and understanding performance within different operational paradigms. Also the simplest case of a local area network.}

An EGS serves a role analogous to a central hub in a system of telephone lines, managing the connection of pairs of quantum nodes to resources, much like a telephone exchange directs and facilitates communication between sets of callers. Traditionally, telephone systems are studied using the Erlang loss model, wherein calls arrive according to a Poisson process to a server with a total of $C$ telephone lines.
%In this model, a server manages $C$ telephone lines while calls arrive according to a Poisson process to use one of the available telephone lines. 
An incoming call will be blocked if all $C$ lines are occupied upon arrival. The blocking probability of a call is computed using the well-known Erlang formula \cite{erlang1917solution}. It has been shown that this model exhibits insensitivity to the type of service time distribution of calls, as the blocking formula depends only on the \textit{average} service time of calls 
%but not on the type of service time distributions 
\cite{sevast1957ergodic}. This result follows from the argument that the underlying Markov process describing the system is a partially reversible one, which is a necessary and sufficient condition to have insensitivity \cite{Bonald_insensitivity}. 
A common method to prove insensitivity is to first study the given queueing model assuming service time distributions are Coxian (these are dense in the class of distributions with nonnegative support \cite{sarfozo_book}), and then show that the stationary distribution or blocking formula is the same as in the case of an exponential service time distribution with the same mean. Then, by using the continuity of queueing models to service time distributions \cite{Whit}, the insensitivity property also holds for generally distributed service times, which can be approximated by Coxian distributions to an arbitrary degree of accuracy.
The insensitivity property is a useful tool to dimension a practical system with a general service time distribution by studying the same system with the simpler case of an exponential service time distribution with the same mean.

In \cite{bonald2006erlang}, Bonald studied the scenario where requests are generated as sessions that arrive according to a Poisson process, with each session containing several calls. It was shown that even in this case, the Erlang model is insensitive to service time distributions. In our model, calls also arrive as sessions, with each session consisting of several attempts for entanglement generation to describe practical quantum systems where entanglement requests arrive in batches from an application. We also assume sessions arrive according to a Poisson process, which will be a valid assumption when a large number of users or applications trigger entanglement requests. The analysis of this paper is based on the analysis of \cite{bonald2006erlang}, albeit our aim here is to analyze a quantum system which has distinctive features when compared to classical systems. Our analysis is significantly different from  \cite{bonald2006erlang} due to the presence of new parameters and characteristics of quantum systems, and subsequently analytical expressions for blocking probabilities for our model are different from those presented in \cite{bonald2006erlang}. We show that a quantum switch that can be modelled as an Erlang model also exhibits insensitivity to service time distributions under certain necessary conditions. 

The EGS architecture was initially introduced in \cite{RCP_EGS}, where the authors highlighted the scalable properties of this type of hub. The authors then proposed and studied a Rate Control Protocol whose aim is to modulate user demand rates to the switch based on the EGS's capacity to serve users, as well as on overall traffic trends. The focus of this work is mainly on fair resource allocation, achieved through a network utility maximization-based \cite{kelly1997charging} framework. In contrast, the protocols proposed in our work use request blocking instead of rate control as means of resource management. Furthermore, our work aims to accurately represent the EGS in a discrete setting, with concrete descriptions of request structure and procedures for request handling.

Memory-equipped quantum switches (EDSs) have been extensively investigated from queueing-theoretic and request scheduling perspectives, see e.g., \cite{vardoyan2021stochastic,vardoyan2020exact,panigrahy2023capacity,Wenhan}. 
In \cite{vardoyan2020exact} the authors modeled an idealized EDS as a discrete-time Markov chain and computed its capacity under a simple swapping protocol. While this study shares similarities with ours in terms of the focus on discrete-time system evolution, important differences exist.
As briefly outlined in the previous section, EDS-based models often assume the presence of quantum memories both at the end nodes and at the switch  (\cite{vardoyan2020exact}, for instance assumes infinitely-sized buffers). Each node may even have dedicated infrastructure connecting it to the hub, including a pre-allocated quantum switch buffer and midpoint stations located at each link to herald entanglement. This enables switch-node entanglement to be attempted independently by each link, and successfully generated states to be stored until a swapping opportunity arises.  
With such an abundance of resources, contention is minimal, and resource allocation is not a primary focus. In another line of work on EDS-based models, contention among entanglement requests of different flows for using switch-node entanglement was studied in \cite{thiru_switch,Wenhan}. In \cite{thiru_switch}, a throughput optimal scheduling policy was designed for a quantum switch model when switch-node entanglement has a short lifetime of one time slot. They proved that the queues of entanglement requests are stable under their max-weight scheduling policy for all feasible entanglement request rates. A similar model with the scenario where switch-node entanglement has infinite lifetime was studied in \cite{Wenhan}.

In contrast to the EDS models, the EGS lacks memories, necessitating resource solicitation by nodes, followed by entanglement generation attempts executed in a synchronized manner to ensure nearly simultaneous photon arrival at the hub.
Furthermore, our EGS protocols involve batched attempts interleaved with periods of EGS inactivity, effectively constituting extended "sessions" of engagement with EGS resources. The system studied in our work thus exhibits both architectural and algorithmic differences to the EDS, requiring novel and tailored analytical methodologies.

\section{System Description}\label{sec:system_descr}

An EGS consists of three main components:
%\begin{itemize}
%    \item[-] 
(1)
    A pool of allocatable \textit{resources} such as BSAs;
%    \item[-] 
(2)
    A \textit{switch} capable of allocating any resource to any pair of nodes connected to the EGS;
%    \item[-] 
(3)
    and a \textit{processor} capable of scheduling the allocation of resources to pairs of nodes, controlling the operation of the switch, and sending and receiving classical messages.
%\end{itemize}
An EGS with control of a pool of three BSA-type resources is illustrated in Figure \ref{fig:star_topology}. Nodes are connected to an EGS by physical links, such as optical fiber connection. To gain access to an EGS resource, pairs of nodes send a message to the EGS called a \textit{request}.

\begin{definition}{Request.} Nodes $n_i$ and $n_j$ who wish to communicate issue a \emph{request} to the EGS.
A request is a demand for the generation of one or more EPR pairs between $n_i$ and $n_j$.
The exact composition of a request depends both on the nodes' physical capabilities as well as on application-focused goals.
\end{definition}

%A distinct pair of nodes $(n_i, n_j)$ may be referred to as a \textit{flow}, $f_{i, j} = (n_i, n_j)$. For every flow one node is assumed to be 

For every distinct pair of nodes $(n_i, n_j)$, one node is designated the \textit{initiator} and the other the \textit{secondary node}. Requests are communicated to the EGS by the initiator node.
A node consists of the following set of components:
%\begin{itemize}
    %\item[-] 
    $(i)$
    one or more \textit{communication qubits}, each capable of preparing a quantum state and emitting one or more photons\footnote{In some simple realizations of a node, the communication qubits may be replaced by all photonic state preparation devices. See \cite{MDIQKDDeployed} for an example.};
    %\item[-] 
    $(ii)$
    devices needed to manipulate the state(s) of the communication qubit(s) -- examples include lasers, waveform generators and microwave sources;
    %\item[-] 
    $(iii)$ devices needed to measure a communication qubit;
    %\item[-] 
    $(iv)$
    possibly one or more quantum memories to which the quantum state of a communication qubit may be swapped, capable of storing the state for a finite period of time;
    %\item[-] 
    $(v)$
    and a classical processor capable of controlling the states prepared in communication qubits, triggering swaps to memory, triggering measurement of a communication qubit, and sending and receiving classical messages. 
%\end{itemize}

%Figure \ref{fig:star_topology} illustrates an EGS connected to four limited quantum nodes in a star topology. 
In Figure \ref{fig:star_topology}, the nodes connected to the EGS are assumed to be limited in that they either have few or no quantum memories, and may have a restricted number of communication qubits.
% For example, only one of the quantum memories in each node may be able to function as a communication qubit.
Figure~\ref{fig:larger_net} illustrates a scenario where the EGS may be connected to more powerful nodes, potentially with access to many communication qubits and quantum memories per node. As in Figure \ref{fig:larger_net}, nodes can function as relays or intermediates interfacing with several quantum local area networks.

\input{operationModes}

\section{Model and Assumptions}
\label{sec:model_assumptions}
\if{false}
\gv{
\begin{itemize}
    \item  Subsection: introduce setup: EGS connected to users (or nodes), star topology (more abstractly than above)
    \item Assumptions: on links, users, memory capabilities at end nodes, etc. (again, in abstract terms, but relating to discussion from previous section)
    \item Subsection: discuss different modes of operation (e.g., single vs multiple successes, retrial behavior or lack thereof, relinquishing resources during idle periods or not, etc.). Throughout, relate to real applications as much as possible.
\end{itemize}
}
\fi
\input{model_assumptions}

\section{Analysis}\label{sec:analysis}
In this section, we analyze the three EGS service models: single and multiple EPR pair generation with strict resource reservation, and multiple EPR pair generation with resource relinquishment during idle periods. For each scenario, we derive the stationary distribution of observing the system in a given state. The main quantity of interest is the probability that an arriving request is blocked -- an event that occurs when said request sees all $C$ EGS resources engaged. In strict resource reservation mode, blocking can only occur at the beginning of a session. On the other hand, if sessions contain idle periods as in the resource relinquishing mode, then blocking may also occur throughout the session, after departures from idle periods. We show that in all cases, the blocking probability depends only on the traffic intensities of the flows -- i.e., the insensitivity property.
%\gv{\begin{itemize}
%    \item Analysis of aforementioned modes of operation. Keep as general as possible.
%    \item Insensitivity, blocking probability, expected service time (at call level, not at session level), etc.
%\end{itemize}}

\input{strict_reservation_analysis_acm}

\section{Numerical Evaluation}\label{sec:NumEval}
\begin{figure}
    \centering
    \begin{minipage}{0.65\textwidth}
    \includegraphics[width=\textwidth, trim={0 0 0 0}]
    {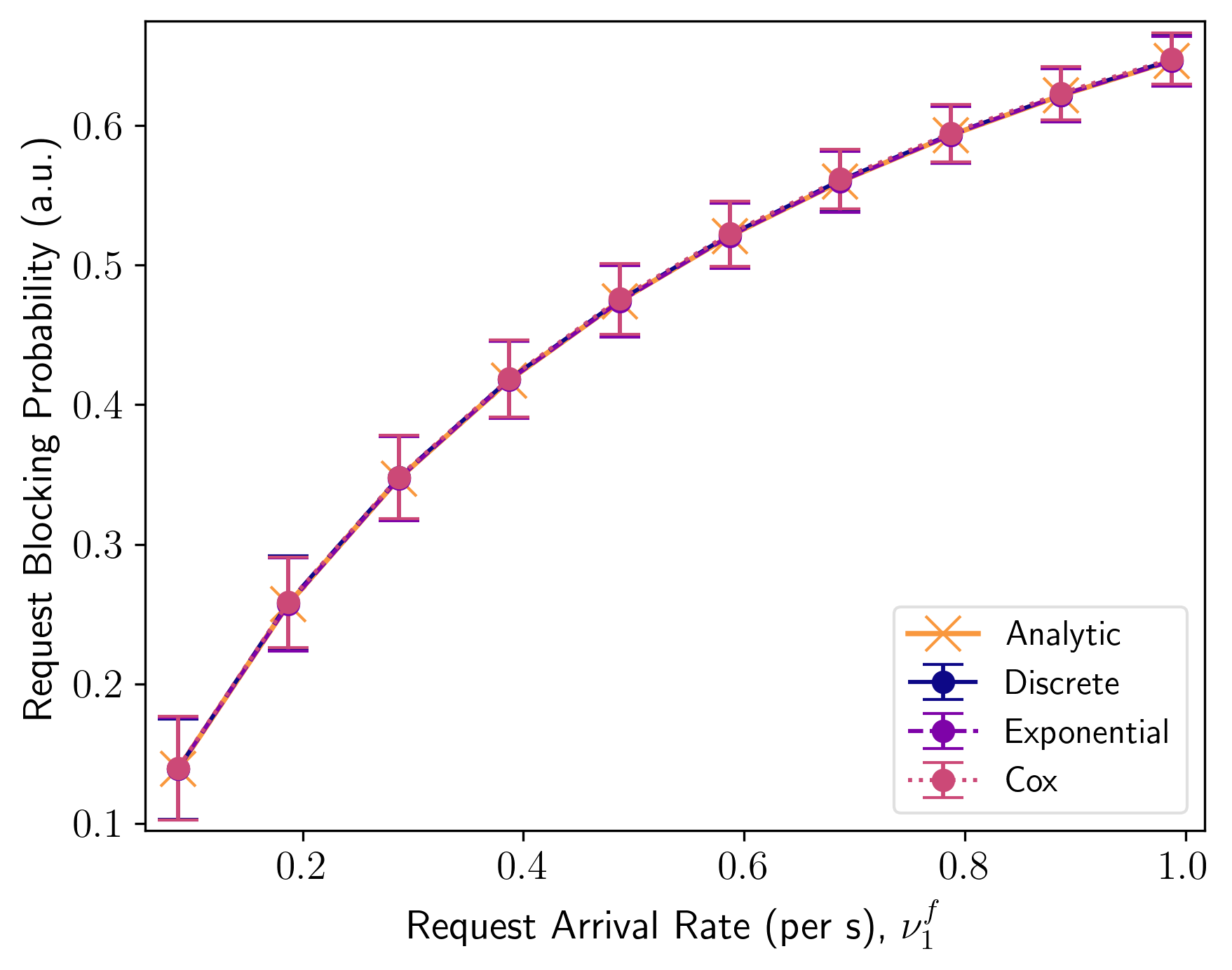}
    \end{minipage}\hfill
    \begin{minipage}{0.32\textwidth}
    \vspace{-2.5cm}
    \caption{Comparison of the average blocking probability per flow according to (\ref{eq:prob_C}) with simulations for an EGS with one resource, connected to eight nodes, and serving $\binom{8}{2} = 28$ flows, one for each possible node pairing. Every node is restricted to a single communication qubit. Session traffic is homogeneous. The absolute relative errors are $\delta_{\text{discrete}} = 0.004, \delta_{\text{exponential}} = 0.001, \delta_{\text{Cox}} = 0.003$.}
    \label{fig:validationStrict}
    \end{minipage}
\end{figure}

\begin{figure}
    \centering
    \begin{minipage}{0.48\textwidth}
    \includegraphics[width=\textwidth]{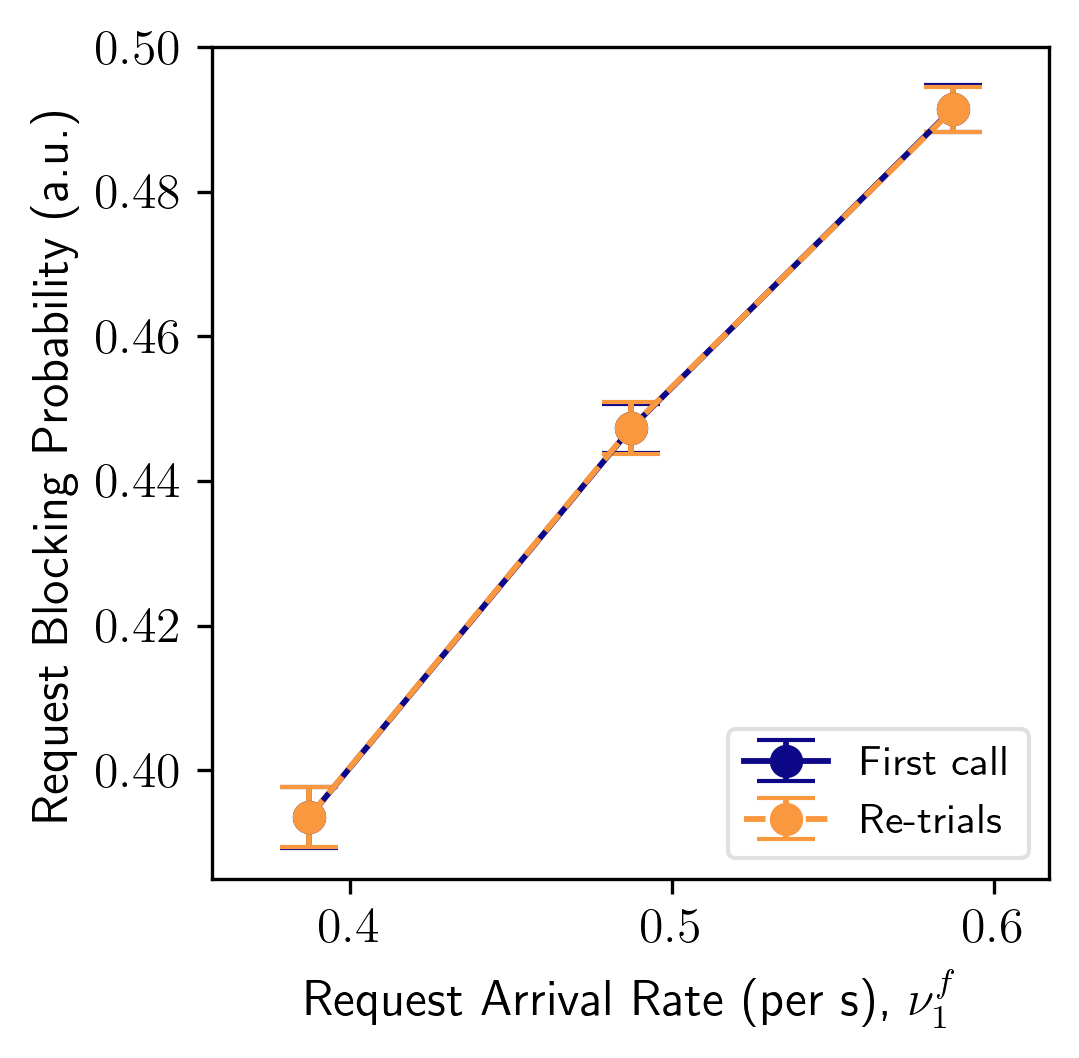}
    \end{minipage}\hfill
    \begin{minipage}{0.48\textwidth}
    \vspace{-3cm}
        \caption{Average blocking probability for the first call of a session in the jump-over service model compared to that of calls which follow idle periods (re-trials). Data is obtained from discrete simulations of an EGS with one resource, serving eight nodes with $\binom{8}{2} = 28$ flows. Every node is restricted to a single communication qubit. Session traffic is homogeneous. The maximum relative difference between the average blocking probabilities is 0.00025.}
        \label{fig:jumpOverReTrials}
    \end{minipage}
\end{figure}

\begin{figure*}
    \centering
    \includegraphics[width=\textwidth, trim={0 0 0 0}]{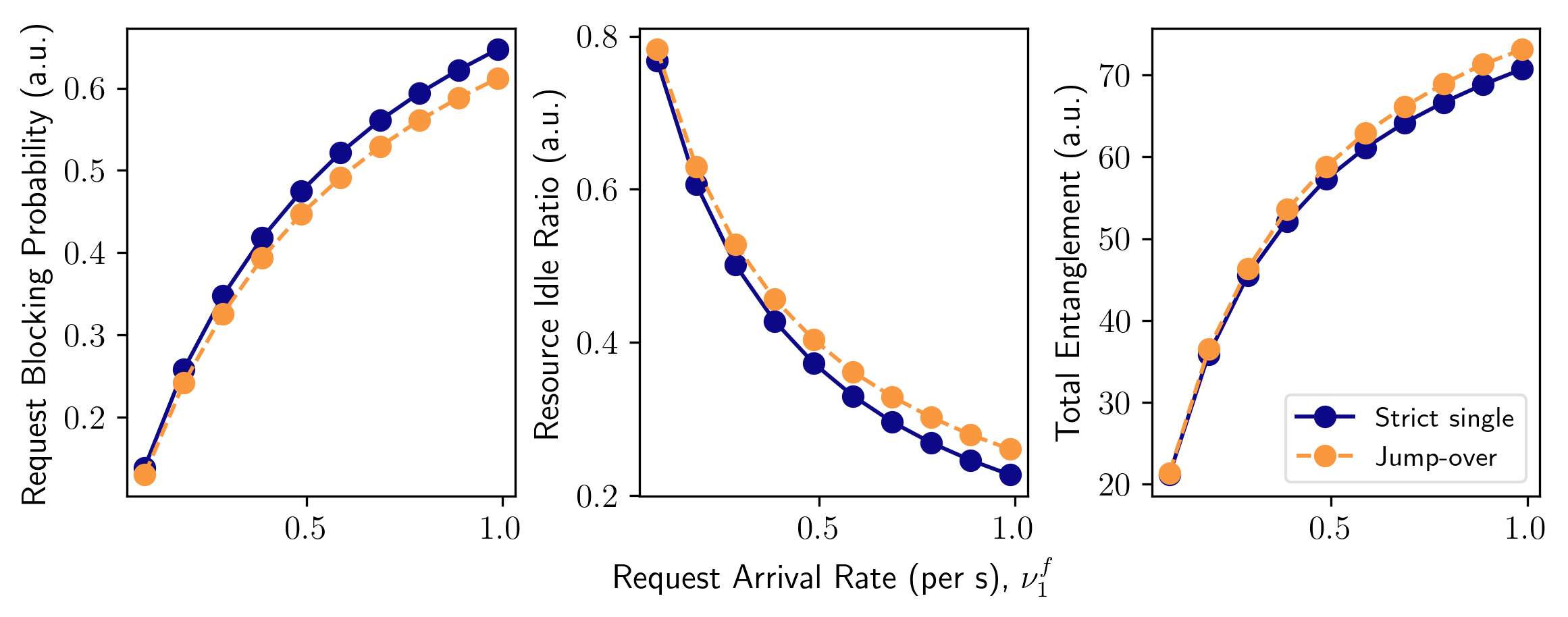}
     \caption{Comparison of the strict and jump-over service models, for an EGS with one resource, serving eight nodes via $\binom{8}{2} = 28$ flows. Every node is restricted to a single communication qubit. Session traffic is homogeneous. Left: blocking probability per request. Middle: proportion of time that the EGS resource is idle, compared to total simulation time. Right: total amount of entanglement generated across all flows, during the time simulated. Data is obtained using discrete simulations.}
    \label{fig:comparePerfMetrics}
\end{figure*}

In real-world implementations of entanglement generation, every individual attempt and calibration period has a finite duration. In a demonstration of deterministic HE delivery carried out between two NV nodes \cite{SingleClickDiamond}, for instance, the authors describe entanglement generation attempts as taking a fixed amount of time, $\Delta t_{\text{attempt}}$. On the other hand, the calibration periods take a variable amount of time, of which the mean duration $\mu_{\text{calib}}$ is known. To model such an experiment one could sample the duration of each calibration period from an exponential distribution with mean $\mu_{\text{calib}}$ and fix the duration of attempts to $\Delta t_{\text{attempt}}$.

\input{figs/SampleParameters}

To validate our analysis, we simulate three models of an entanglement generation experiment. These are referred to as \textit{discrete}, \textit{Cox}, and \textit{exponential} and are differentiated by how the duration of entanglement generation attempts and calibration periods are determined. To ensure the simulations are compatible with our analysis two key assumptions are made. First, in all simulation modes and in numeric evaluation of (\ref{eq:prob_C}) we fix the mean duration of every attempt (calibration period/idle period) to some value $T_{\text{attempt}}$ ($T_{\text{calib}}/T_{\text{idle}}$). In discrete simulations the duration of each attempt (calibration period/idle period) is set exactly to these values. Settings (number of phases, duration of phases, transition probabilities between phases)
for the Cox distribution, are chosen to ensure (\ref{eq:mu_exp_cox}) is satisfied. 
Our insensitivity result suggests these settings may be chosen arbitrarily, as long as they obey (\ref{eq:mu_exp_cox}). In Appendix \ref{app:SimulationImplementation} we include a sample of parameters used for the Cox distribution governing attempts. 
The second assumption is that the probability any attempt results in successful entanglement generation is a fixed value, $p_{\text{gen}}$. For justification, see Section \ref{subsubsec:pgen}. Simulation parameters are detailed in Table \ref{tab:sampleParams}. Further details on simulation implementation are included in Appendix \ref{app:SimulationImplementation}.

To quantify agreement between numeric evaluation of (\ref{eq:prob_C}) and simulated results we define error parameters $\delta_{\text{sim. type}}$ based on the maximum absolute relative difference between the points of the analytic and simulated data sets,
\begin{equation}
\label{eq:absRelDiff}
    \delta_{\text{sim. type}} : = \frac{\abs{\max_x (y_{\text{Analytic}}[x] - y_{\text{Sim.}}[x])}}{y_{\text{Analytic}}[\text{arg}\max_x(y_{\text{Analytic}}[x]  - y_{\text{Sim.}}[x])]},
\end{equation}
where $y_{\text{Analytic}}$ denotes an analytic data set, $y_{\text{Sim.}}$ denotes a simulated data set, and square brackets denote indexing the data sets. The error parameter reports the difference between the analytic and simulated data point for which the difference is maximum, relative to the analytic value at that point. 

Each simulation is run for a duration corresponding to 1e7 iterations of the discrete simulation, equivalent to a simulation of 1150.73 seconds of \textit{simulated real-time}. See Appendix \ref{app:SimulationImplementation} for details of the expected number of requests placed over the duration of a simulation.
For any run  of a simulation, the average blocking probability is calculated by averaging the blocking probabilities recorded by each flow. Every data point from a simulation is the result of averaging over 200 independent runs of the simulation. Error bars for request blocking probabilities correspond to the average (over the runs) standard deviation in blocking probabilities between flows.

Traffic in a deployed network will in general be non-homogeneous. The causes of non-homogeneity may stem from differences in the physical parameters of quantum nodes and their connections to an EGS or they may result from differences in request parameters submitted by flows, enabling them to target different applications. In general, our analysis applies to non-homogeneous traffic. However, the simplest physical implementation of the EGS is one where all nodes are connected to the hub by links of equal length and all flows submit identical requests for sessions with equal rates. This is a homogeneous traffic scenario. In the remainder of this section we first present results from each of the three service models developed in Section \ref{sec:model_assumptions} in a homogeneous traffic scenario. Then, we demonstrate the application of our analysis to a non-homogeneous traffic scenario in the single entanglement generation with strict resource reservation service model.

\subsection{Homogenous Traffic}
\label{subsec:homoTraffic}

\begin{figure}
    \centering
    \begin{minipage}{0.64\textwidth}
    \centering
    \includegraphics[width=\textwidth]{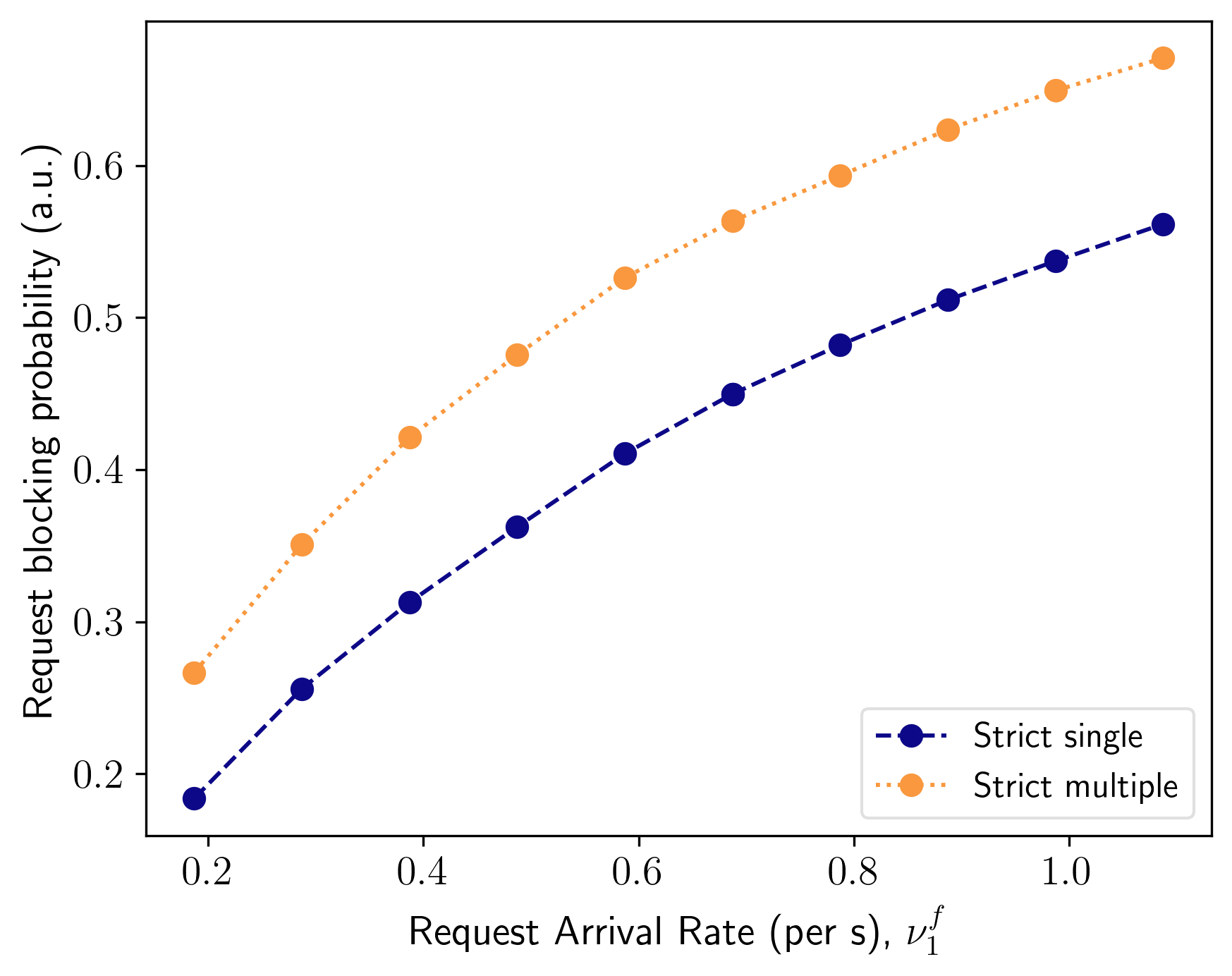}
    \end{minipage}\hfill
    \begin{minipage}{0.35\textwidth}
    \centering
    \vspace{-1cm}
    \caption{Comparison of the request blocking probabilities of the strict single and strict multiple service modes when there is a high probability ($p_{\text{gen}}=0.001$) that an attempt to generate entanglement succeeds. Data is obtained from discrete simulations of an EGS with one resource, serving eight nodes with $\binom{8}{2} = 28$ flows. Every node is restricted to a single communication qubit. Session traffic is homogeneous. The absolute relative errors are $(\delta_{\text{discrete}}, \delta_{\text{exponential}}, \delta_{\text{Cox}}) = (0.014, 0.018, 0.003)$ for the strict single service mode and (0.032, 0.015, 0.002) for the strict multiple service mode.}
    \label{fig:highPGenStrict}
    \end{minipage}
\end{figure}

In a deployed quantum network, a network operator may be interested in selecting a service model for an EGS based on its performance across a range of metrics. We compare the performance of the three service models of Section \ref{sec:analysis}. 
% single EPR pair generation with strict resource reservation, multiple EPR pair generation with strict resource reservation, and  multiple EPR pair generation with resource relinquishment. 
In what follows, these service models are referred to simply as \textit{strict single}, \textit{strict multiple} and \textit{jump-over}, respectively. 
%From the perspective of a network operator, interesting performance metrics for an EGS service model may include the blocking probability, 
Besides request blocking probabilities, we also study resource utilization and the total entanglement generated in a fixed amount of time. 
%The blocking probability provides information about the load on the network. 
The former provides information on how efficiently network resources are used, 
%(i.e. how much of the network capacity is used), 
and the latter gives a measure of the productivity derived from the allocation of network resources. 
To study the resource utilization in simulation, we define the \textit{resource idle ratio} as the proportion of time that one or more EGS resources is idle, relative to the entire simulated time. To study the total entanglement generated we track and sum the successful generation of entangled pairs by any flow over the duration of the entire simulated time. 

Before comparing the service models, we validate our analysis of the blocking probability for each of the three service models. Figure \ref{fig:validationStrict} compares numeric evaluation of (\ref{eq:prob_C}) with results from discrete, exponential, and Cox simulations of an EGS operated in the strict single service model, with control of one resource $(C = 1)$, connected to eight nodes, each with a single communication qubit. We observe close agreement between the analytic and simulated results, all of which overlap well within one standard deviation for every data point. This is expected due to the insensitivity result (Theorem \ref{thm:insensitivity}) and supports our analysis. The tightness of the overlap between each simulated data set and the analytic results is captured by the absolute relative errors, defined by (\ref{eq:absRelDiff}).  These errors are $<1\%$ for each simulation type. 
We validate the other service models in the same way and obtain error parameters of $(\delta_{\text{discrete}}, \delta_{\text{exponential}}, \delta_{\text{Cox}}) = ( 0.061, 0.015, 0.006$) for the multiple EPR pair service mode and ($0.001, 0.003, 0.004$) for the jump-over service mode. To further validate our analysis of the jump-over service model, Figure \ref{fig:jumpOverReTrials} compares the average probability that the first call of a session is blocked with the probability that any call which follows an idle period (re-trial) is blocked. The comparison is made using results from discrete simulations, however similar results are obtained from exponential and Cox simulations. As predicted by our analysis (Appendix \ref{app:jump_over}), the blocking probability is equal for any arbitrary call and does not depend on whether a call is the first one of a session or follows an idle period (re-trials).  

The performance of the strict single and jump-over service models is contrasted in Figure \ref{fig:comparePerfMetrics}. Load on the EGS network results from requests for resource access by the flows. Increased load directly leads to increased blocking probability and decreased idle time ratio in each service mode. When there is a relatively low-load on the EGS, the difference between each performance metric of the two service modes is marginal. 
%However, as the load increases, the differences between service modes increase. 
When there is high load on the EGS, the jump-over service mode results in lower blocking probabilities, indicating better handling of the high load. The lower blocking probability of the jump-over model is reflected in the increased idle time ratio. Interestingly, although the EGS resources are idle for a greater proportion of the time in the jump-over mode, a greater total amount of entanglement is produced. This indicates that for this EGS configuration, the jump-over service mode makes more efficient and productive use of the EGS resources. 
For the physically motivated simulation parameters used, the performance of the strict multiple service mode is not significantly different from that of the strict single service model, hence it is omitted from Figure \ref{fig:comparePerfMetrics}. With these parameters, the expectation value of successful entanglements per session is 0.01. 

 The differences between the strict single and strict multiple service modes are more evident if  the probability that an attempt to generate entanglement succeeds is significantly increased. In Figure~\ref{fig:highPGenStrict} we compare the blocking probabilities resulting from discrete simulation of these service modes when $p_{\text{gen}}= 0.001$, an increase of two orders of magnitude relative to our baseline simulation parameters. The expectation value of successful entanglements per session is one EPR pair. In this setting, the strict single service mode realizes a much lower blocking probability than the strict multiple service mode. This difference can be understood as resulting from a difference in the mean service times of the sessions. The recorded mean service time of sessions in the strict multiple service mode is $125$ ms. In the strict single service model however, sessions may end earlier if an attempt at entanglement generation succeeds, resulting in a lower observed mean service time of $78$ ms for these sessions. Since sessions are on average shorter in the strict single service mode, for the same request rate, resources are free more often, leading to a lower blocking probability.

%The blocking probability per flow, (\ref{eq:condn_blocking_prob}), depends upon the number of communication qubits available to each node and the number of resources at the EGS. 
To investigate the impact of the restrictions on communication qubits, we numerically evaluated the blocking probability for an EGS controlled by the strict single service mode with two or three resources (Figure \ref{fig:commQHeatmaps}) as the restriction varies from one to ten communication qubits per node. In each case, for any fixed request rate we observe that an increase in the number of communication qubits per node from one to two results in a large increase in the blocking probability, but further increases in the number of communication qubits have little impact. Numeric evaluation for strict service mode scenarios where an EGS with one, two, or three resources is connected to 20 nodes and serves $\binom{20}{2} = 190$ flows (see Appendix \ref{appendix:extendedData}) confirm that the same effect holds for an EGS serving a large number of flows. Recall that a flow does not submit a request if one or more of the nodes of the flow does not have a communication qubit available. Hence for any fixed request rate, a lower blocking probability is expected if nodes do not have enough communication qubits available, due to the lowered effective request rate. 
As a result of these properties and the observed trends, we conclude that when each node is restricted to one or two communication qubits, resource unavailability at the EGS has less of an impact on the blocking probability than communication qubit unavailability. In contrast, when nodes are equipped with three or more communication qubits we may conclude that communication qubit unavailability does not impact the blocking probability. Thus a limitation to any finite number of communication qubits per node $\geq 3$ is seen not to limit the blocking probability at an EGS controlled by the strict service model and serving homogeneous traffic. 

% We conclude that when an EGS controlled by the strict single service mode serves homogeneous traffic, an increase in the number of communication qubits per node from one to two has a large impact on the blocking probability, but further increases have a limited impact. \gv{Not sure I understand this result.}

\begin{figure}
    \centering
    \begin{minipage}{0.74\textwidth}
    \begin{subfigure}[t]{0.48\textwidth}
        \centering
        \includegraphics[height=2.0in]{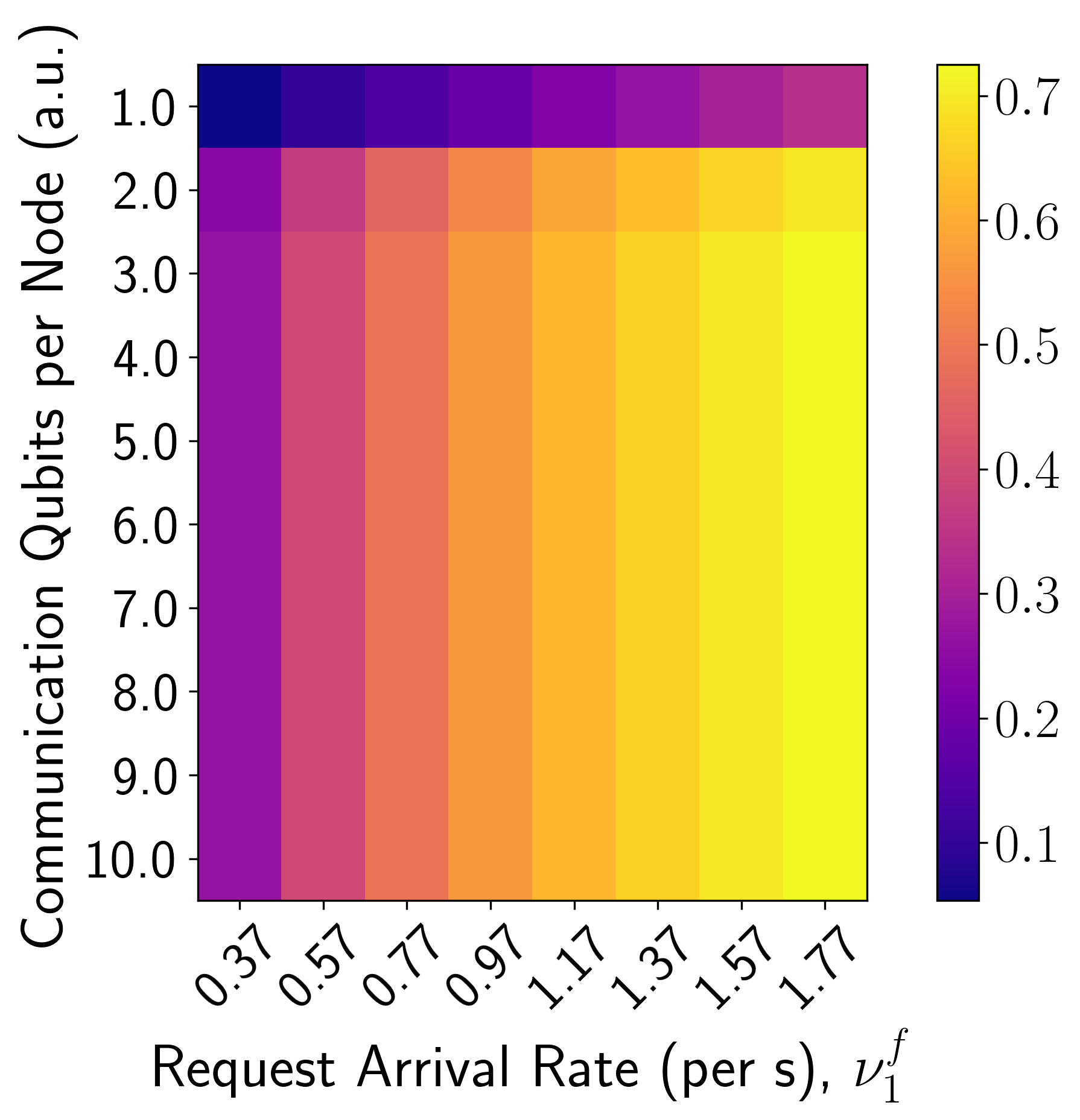}
        \caption{two EGS Resources ($C=2$)}
        \end{subfigure}
        ~
        \begin{subfigure}[t]{0.48\textwidth}
        \centering
        \includegraphics[height=2.0in]{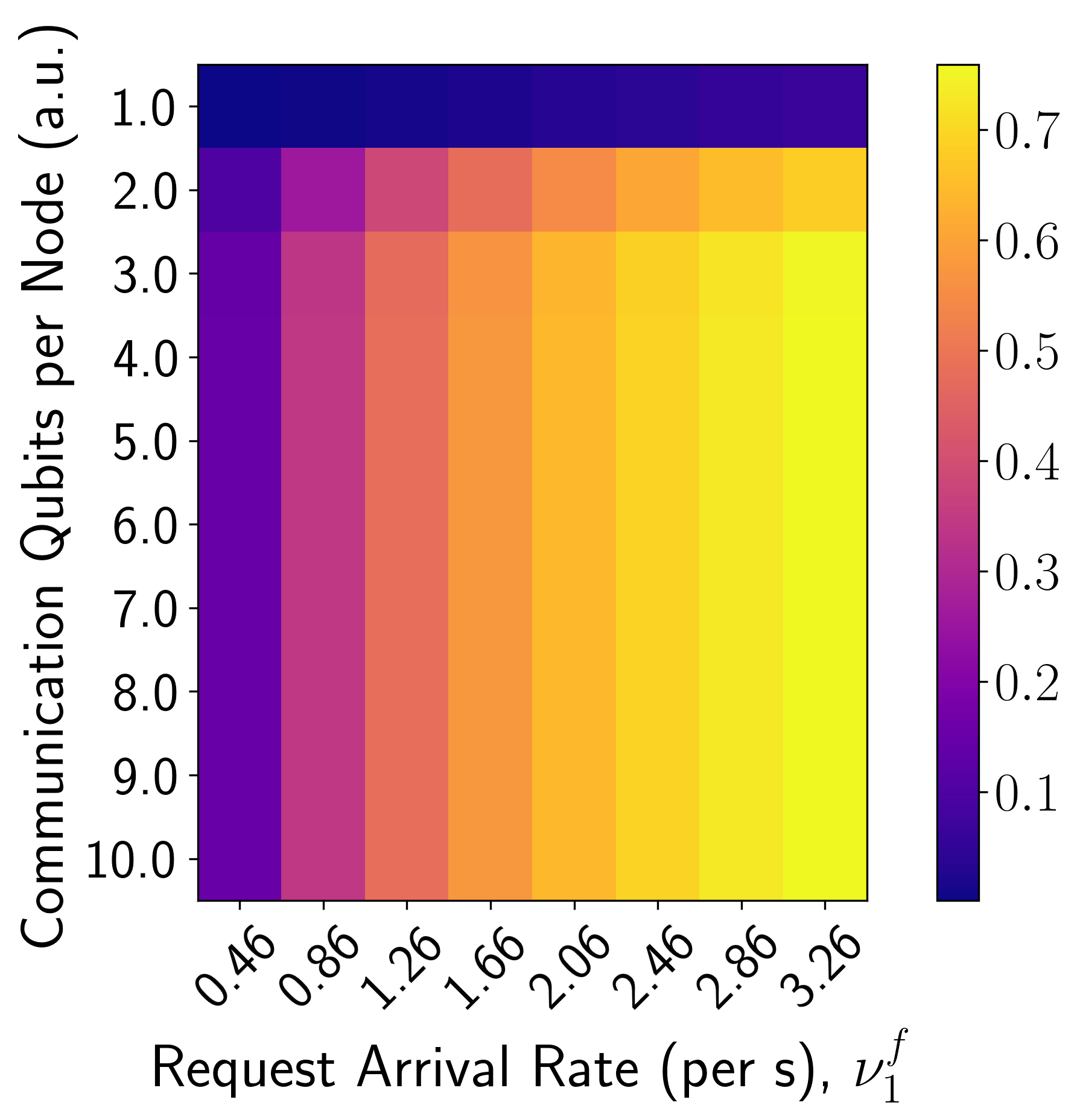}
        \caption{three EGS Resources ($C=3$)}
        \end{subfigure}
    \end{minipage}\hfill
    \begin{minipage}{0.24\textwidth}
    \vspace{-2mm}
    \caption{Heatmaps of the average blocking probability per flow when the number of communication qubits per node and the request arrival rates are varied. Data results from numeric evaluation of (\ref{eq:prob_C}) for an EGS with eight nodes and serving $\binom{8}{2}=28$ flows, one for each possible node pairing. Session traffic is homogeneous.}
    \label{fig:commQHeatmaps}
    \end{minipage}
    % \vspace{-5mm}
\end{figure}

\subsection{Non-homogeneous traffic}\label{subsec:nonHomoTraffic}

\begin{table}
    \centering
    \begin{tabular}{c|c}
        \textbf{Description} &  \textbf{Value} \\
        \hline
        Link lengths, set 1: users $1 - \lceil K/2 \rceil $ & 10 km  \\
        Link lengths, set 2: users $ (\lceil K/2 \rceil  + 1) -  K $ & 20 km  \\
        Single-mode optical fiber attenuation coefficient, $\sigma_{\text{att}}$ & 0.2 \ dB/km \\
        One-way communication time, $S_2$ (\text{RTT}/2) & 100.07 $\mu$s \\
        % Buffer for jitter in heralding flag arrival, set 2 &  20.008   $\mu$s \\
        Resource allocation duration, $S_1$ & 125 ms \\
        Resource allocation duration, $S_2$ & 240 ms \\
        Probability photon travelling 1 km optical fiber arrives, $p_a = {10^{-(\sigma_{\text{att}} / 10)}}$ & 0.95499 (a.u.) \\
        Probability of single attempt success, set 2, $p_{\text{gen}} \cdot p^{20}_{a} / p^{10}_{a}$ & 6.31e-6 (a.u.)\\
        Duration of a single attempt, $S_2$ & 230.146 $\mu$s\\ 
    \end{tabular}
    \caption{Additional physical parameter settings for a non-homogeneous traffic scenario where half of the nodes have doubled link lengths to the EGS.
    These parameters supplement those of Table \ref{tab:sampleParams}.}
\label{tab:nonUniformParams}
\vspace{-5mm}
\end{table}

 Non-homogeneity resulting from network topology will be relevant in any deployed network, as it is unrealistic to expect that nodes may be located in a perfect disk with fixed link lengths to an EGS. We examine a situation where non-homogeneous traffic in the strict single service model results from a network topology where half of the nodes are connected to the EGS by links of length $10$ km and the other half of the nodes, are connected by links of double the length, i.e. $20$ km. Flows are partitioned into two sets, labelled as \textit{set 1} ($S_1$) and \textit{set 2} ($S_2$). For every flow in $S_2$, at least one of the nodes is connected to the EGS by a $20$ km link. For every flow in $S_1$, both nodes are connected to the EGS by $10$ km links. 
 The length of a link between a node and the EGS affects the RTT and the probability that a photon sent over the link arrives at the EGS. Table \ref{tab:nonUniformParams} recounts the physical parameters for each set of nodes. See Appendix \ref{app:SimulationImplementation} for further discussion of the simulation implementation details. Note that within each set of nodes, traffic is homogeneous. 

% Theorems \ref{thm:blocking_flow} and \ref{thm:insensitivity} apply to the general case of non-homogeneous traffic. For validation, 

 Figure \ref{fig:nonHomogenous} demonstrates that an arbitrary flow receives different service depending on whether it is an element of set $S_1$ or $S_2$. Requests from flows in $S1$ are blocked with a higher probability and have shorter mean service times than flows in $S_2$. Figure \ref{subfig:nonUniformAvBP} demonstrates that the blocking probability of any arbitrary flow $f$ from the set of total flows $\mathcal{F}$, which is equal to the average request blocking probability as given by (\ref{eq:pi_bar_avg}), is greater than for flows in $S_2$ but less than for flows in $S_1$.
 The gap between the blocking probability for flows in $S_1$ and $S_2$ is due to the restriction to a single communication qubit per node and the difference in service times of sessions for flows in each set. Although flows from $S_1$ and $S_2$ have equal request rates, due to the longer mean service times of flows in $S_2$, one or more nodes from flows in $S_2$ will not have a free communication qubit more often than for flows in $S_1$. As a result, the effective request rate and hence the blocking probability is lower for flows in $S_2$.
 Figure \ref{subfig:nonUniformBP} validates that the analytic blocking probability calculated for any flow in $S_1$ ($S_2$) agrees with $<4 \%$ ($<6 \%$) absolute relative error with the average blocking probability over flows in set $S_1$ ($S_2$) observed in each of discrete, exponential and Cox simulations. 
 These results validate (\ref{eq:prob_C}) in the presence of non-homogeneity and highlight the complex interplay between the number of communication qubits at a node, mean service times of a session, and blocking probability of an arriving session.

\begin{figure}[H]
    \centering
    \begin{subfigure}[t]{0.25\textwidth}
    \centering
    \includegraphics[height=1.9in, trim={0.4cm 0 0.4cm 0}]{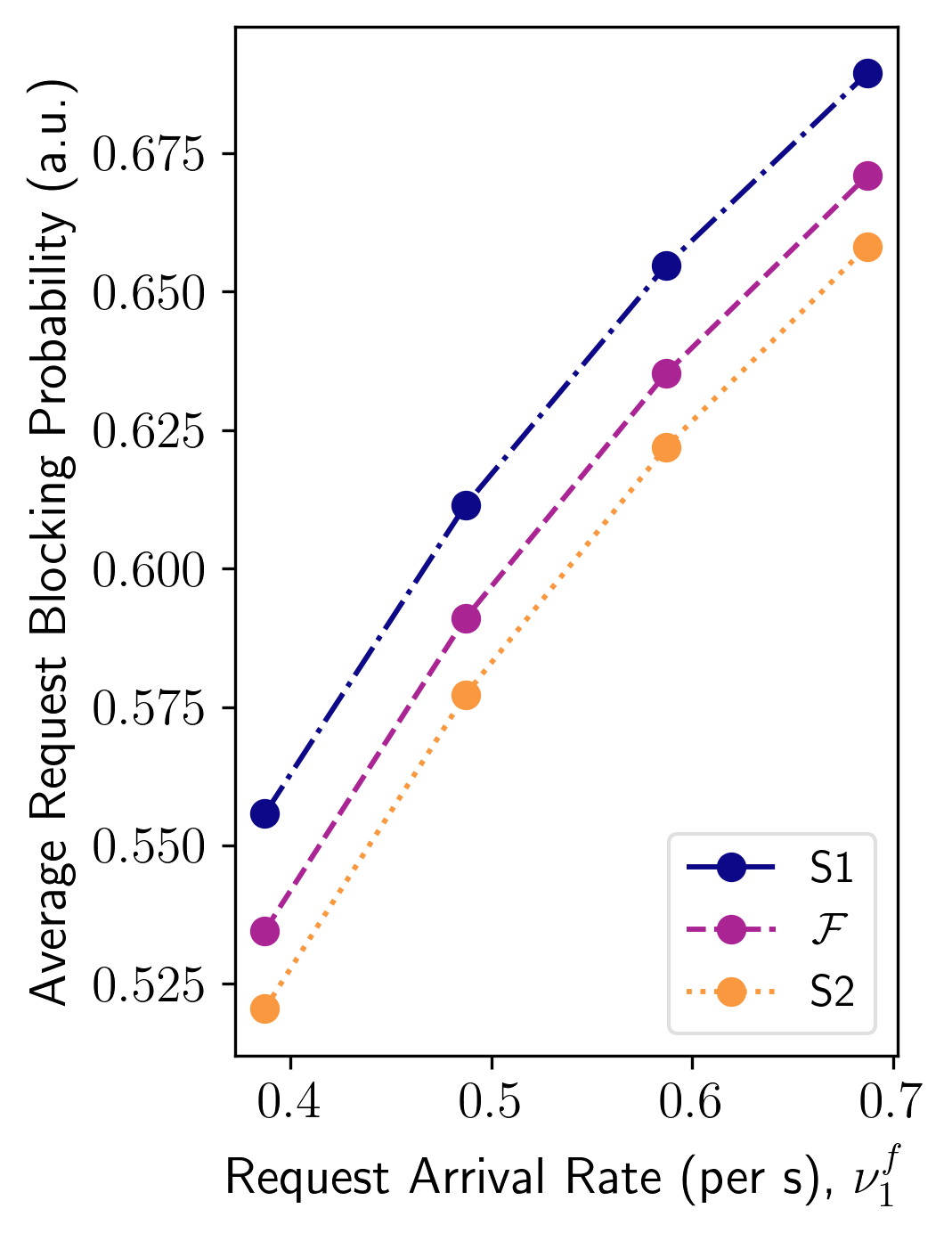}
    \caption{}
    \label{subfig:nonUniformAvBP}
    \end{subfigure}
    ~
    \begin{subfigure}[t]{0.42\textwidth}
        \centering
        \includegraphics[height=1.9in, trim={0.4cm 0 0.4cm 0}]{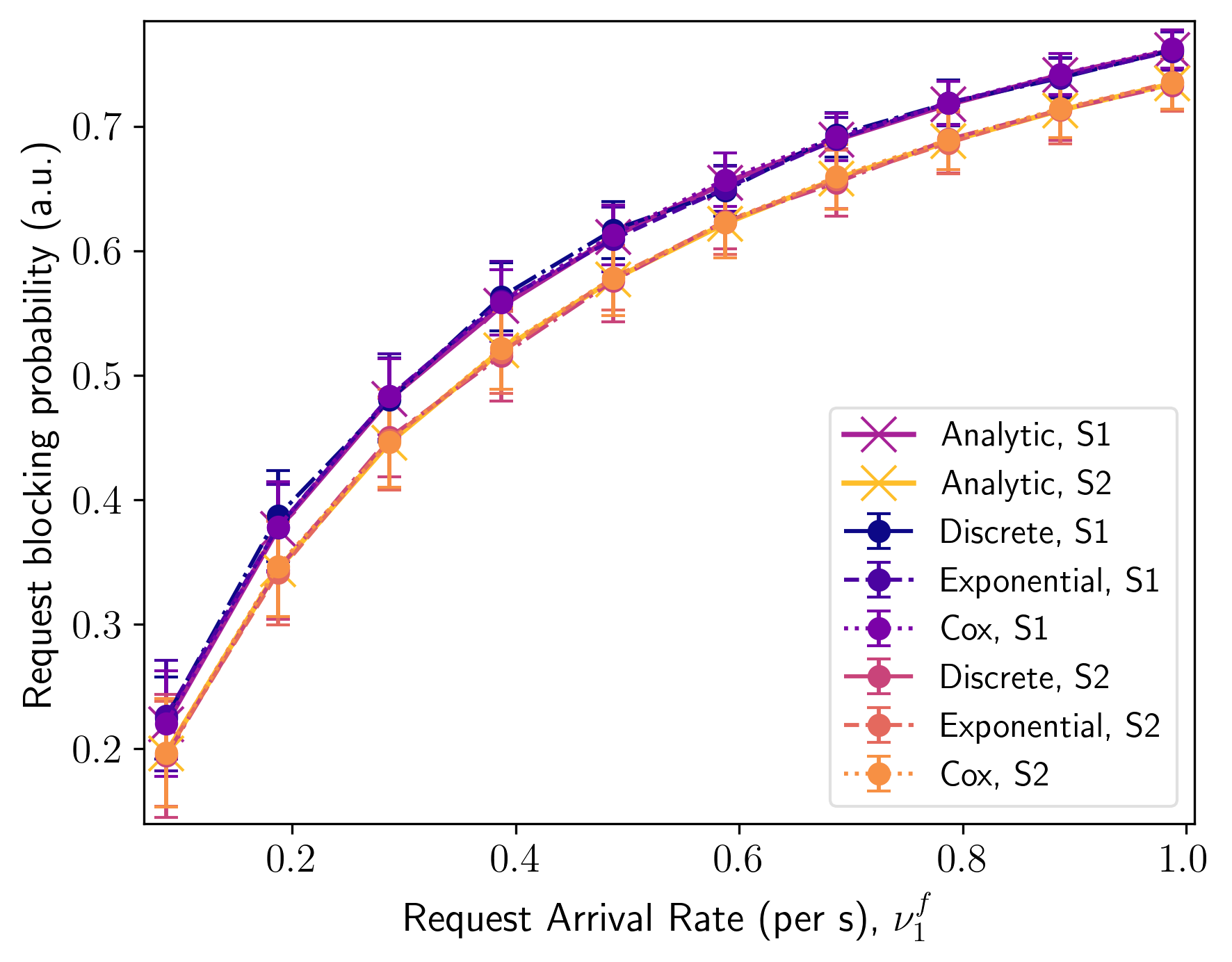}
        \caption{}
        \label{subfig:nonUniformBP}
        \end{subfigure}
        ~
    \begin{subfigure}[t]{0.25\textwidth}
        \centering
        \includegraphics[height=1.9in, trim={0.4cm 0 0.4cm 0}]{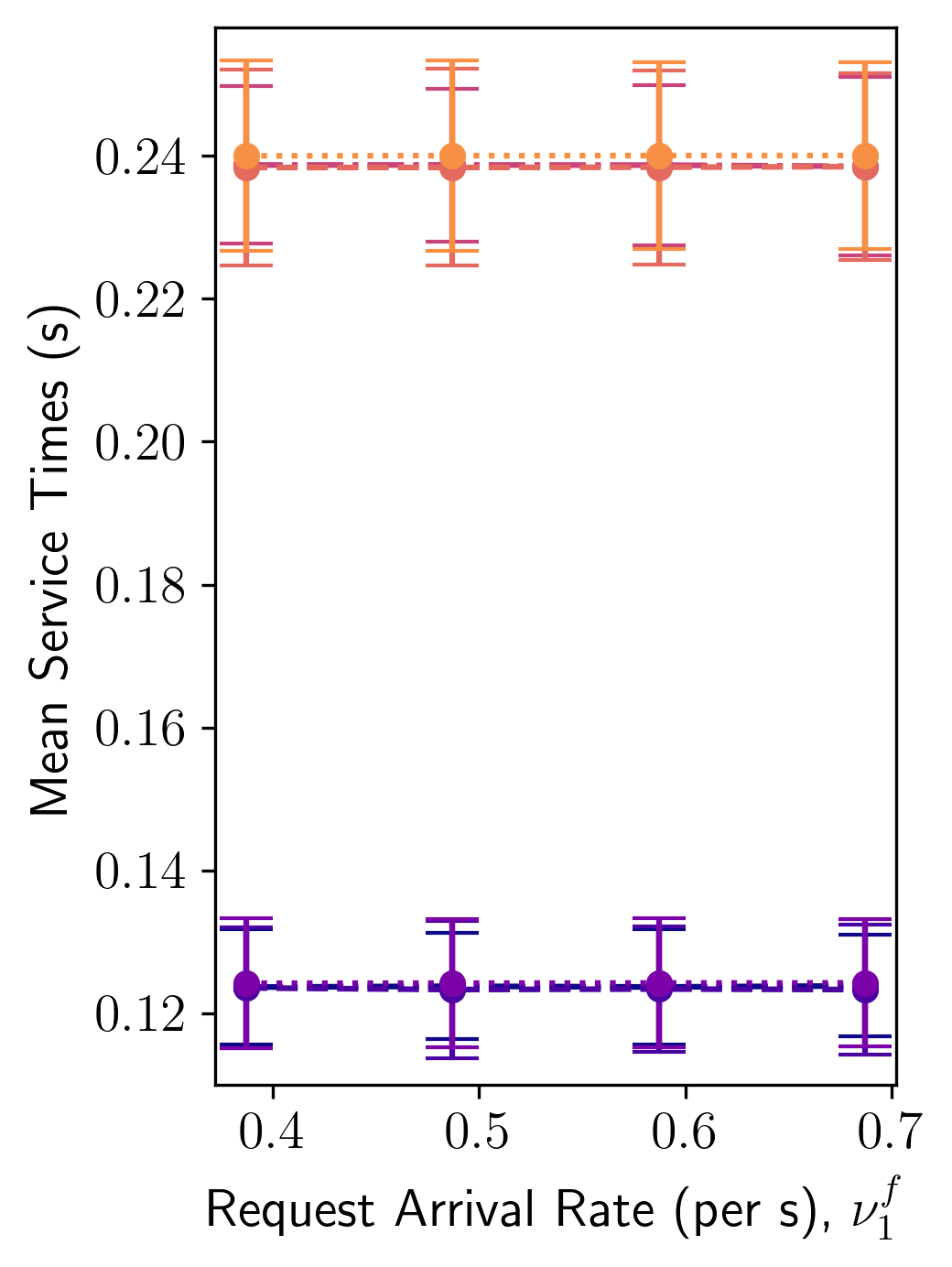}
        \caption{}
        \label{subfig:nonUniformMST}
        \end{subfigure}
    \label{fig:validateStrictNonUniform}
 \caption{Non-homogeneous traffic, strict single service model: An EGS with one resource is connected to eight nodes with one communication qubit each, of which four nodes are connected to the EGS by 10 km links and four by 20 km links. The set $\mathcal{F}$ is partitioned into sets $S_1$ and $S_2$. For every flow in $S_2$, at least one of the nodes is connected to the EGS by a $20$ km link. For every flow in $S_1$, both nodes are connected to the EGS by $10$ km links. Traffic within each set is homogeneous. 
 (a) Analytic blocking probability for flows within sets $S_1$ and $S_2$ compared to the average blocking probability for any flow $f \in \mathcal{F}$, which results from evaluation of (\ref{eq:pi_bar_avg}).
 (b) Average request blocking probabilities within sets $S_1$ and $S_2$ from discrete, exponential, and cox simulations compared with numeric evaluation of (\ref{eq:prob_C}). The absolute relative errors are $(\delta^{S_1}_{\text{discrete}}, \delta^{S_1}_{\text{exponential}}, \delta^{S_1}_{\text{cox}}) = (0.027, 0.032, 0.0032)$ and $(\delta^{S_2}_{\text{discrete}}$, $\delta^{S_2}_{\text{exponential}}, \delta^{S_2}_{\text{cox}}) = ( 0.054, 0.0088, 0.0064)$. 
 (c) Comparison of mean service times for sessions in $S_1$ with those in $S_2$. Legend is the same as in (b).}
 \label{fig:nonHomogenous}
\end{figure}

\section{Conclusion and Future Work}\label{sec:conclusion}
We have proposed an on-demand resource allocation algorithm for an EGS and developed its performance analysis in a variety of traffic scenarios and operation modes. We modeled the system as an Erlang loss one, and discovered an insensitivity property for the request blocking probability. Numerical results validate our analysis. Our work can lead to several new research directions; we provide a non-exhaustive list here. 
The analytic and simulation frameworks we provide are valuable tools for the development of load-balancing control algorithms for an EGS, which could run at a higher level in the control stack to ensure stable quality of service can be delivered to flows. 
%Our results also provide tools for analyzing how resources should be provisioned in a network in order to enable service of a certain level. 
An important highlight of our model is that it flexibly incorporates restrictions that are very present in NISQ era quantum devices, hence being relevant for the development of a real near-term network. This feature of the model can be used as a tool to investigate efficient resource provisioning schemes -- not only for a single EGS serving a number of nodes in a star topology, but also for a more complex network made up of heterogeneous devices.
In future work, one could allow a flow to have control over multiple communication qubits. This generalization, along with some of the modeling challenges that might be encountered, were discussed briefly in Section~\ref{sec:model_assumptions}.
Another natural extension of our model is to generalize from bipartite to multipartite entangled states. This can be accomplished by expanding the definition of a "flow" to include $k\geq 2$ nodes. Last, one could study  various back-off mechanisms in lieu of blocking, and determine suitable hardware regimes for each of the schemes.

% \if{false}
\begin{acks}
GV acknowledges support from NWO QSC grant BGR2 17.269.
SG acknowledges the support of the European Union's Horizon Europe research and innovation program under grant agreement No. 101102140. SG thanks her PhD supervisor Stephanie Wehner for her support in conducting this research.
\end{acks}
% \fi
\bibliographystyle{ACM-Reference-Format}
\bibliography{refs}

\appendix
\section{Probability of Reaching a Queue in Single EPR Pair Generation Mode}
\label{app:prob_reaching_queue}
In this appendix, we derive 
%various probabilities of reaching queues within the queueing networks introduced throughout the work. First, we begin with 
an expression for
$\tilde{p}_i^f$ as seen in (\ref{eq:lambda_i_cox}), the probability of reaching the $i$th phase of a type $f$ session, starting from its first phase, while in strict resource reservation, single EPR pair generation service mode. Suppose that the $i$th phase is located within the $j$th period of the session. To have reached this phase, the session must not have ended during any of the previous attempt phases (recall that in this service mode, sessions retain EGS resources during calibration periods). In the following, let $p^k_{i,j}$ denote the probability of transitioning from the $i$th phase of period $k$ to the $j$th phase of the same period; we omit $f$ below for cleanliness. The probability of leaving during an attempt period $k$ due to success in generating entanglement, conditioned on the event that the session has entered period $k$, is given by
\begin{align}
    P_k &\equiv\mathrm{P}\left(\text{leave during attempt period}~ k | \text{session has entered period k}\right)\\
    &= p^k_1+p^k_{1,2} p^k_2 + \dots + (p^k_{1,2}\cdots p_{N^k-1,N^k}^kp^k_N)\\
    &= \sum\limits_{m=1}^{N^k}p_m^k\prod\limits_{l=1}^{m-1}p^k_{l,l+1},
    \label{eq:prob_leave_period_cox}
\end{align}
%\gv{Probability of leaving after first phase (success prob) should be equal to the probability of leaving after any phase of the same period. $p_i^k = p_j^k$. This should also be equal to the BSM success prob. of exponential distribution.}
%\gv{Let $q$ = BSM success in exponential case. Let $p_i^k = q$. $P_k$ = BSM success in Coxian case.}
where $p^k_l$ is the probability of leaving the queueing network after phase $l$ of period $k$ due to a successful BSM at the EGS, $p^k_{l,l+1}$ is the probability of transitioning to the next phase of the period, and $N^k$ is the number of phases in the period (above, we let $\prod\limits_{m=a}^b x_m=1$ if $b<a$). Then, the probability of the session reaching the $i$th phase of the $j$th period is given by
\begin{align}
    \mathrm{P}\left(\text{session reaches } i\text{th phase of period }j,\text{ starting from the 1st phase of the session}\right)\nonumber\\ 
    = \prod\limits_{l=1}^{i-1}p^j_{l,l+1}\prod\limits_{k=1}^{j-1}\left(1-P_k\right),
    \label{eq:prob_sesh_reaches_ith_cox}
\end{align}
where for any $k$ that corresponds to a calibration period $P_k=0$.
\section{Proof of Insensitivity Theorem for Single EPR Pair Generation with Strict Resource Reservation}
\label{app:insensitivity}
In this appendix, we present a proof of Theorem~\ref{thm:insensitivity}.
\input{insensitivity_appendix}
\section{Analysis of System with Idle Periods and Jump-Over Blocking}
\label{app:jump_over}
\input{jump_over_appendix}
\input{AppendixCorrelatedInformation}
\input{AppendixSimulationImplementation}
\input{AppendixValidation}
\end{document}

%% file: operationModes.tex
\subsection{Physical Operation Settings}
Bipartite Heralded Entanglement (HE) generation and generation of Correlated Information (CI) are two ways in which a pair of nodes can interact via the EGS. In Appendix \ref{appendix:OpAppendix}, we describe a protocol for the latter, while here we
 describe in detail a protocol for bipartite HE generation \cite{theoryHeraldCCGZ, theoryHeraldDLCZ} that is based on a single-click scheme \cite{SingleClickDiamond}.
 This method of producing entanglement has been successfully demonstrated in several experimental platforms, including Color Centers \cite{DoubleClickDiamond, SingleClickDiamond}, Ion Traps \cite{HerEntTrappedIons1, HerEntTrappedIons2}, Atomic Ensembles \cite{HerEntOriginalAE, HerEntSecondAE} and Neutral Atoms \cite{NeutralAtomsHeralded}. 
 %We describe a single-click scheme for the HE generation protocol.
 Applications of HE generation include BQC, teleportation and clock-synchronization \cite{clockSynchPaper}, and an application of CI generation is Measurement Device Independent QKD (MDI-QKD) \cite{originalMDIQKD1, originalMDIQKD2}. Each of these tasks can be enabled by an EGS where the shareable resource is a BSA. The EGS can be equipped with multiple BSAs (Figure \ref{fig:star_topology}), or more generally with a different type of resource, in order to support other interaction protocols between quantum network nodes. Each of the protocols we describe is compatible with a BSA that consists of a 50/50 beam splitter with two input channels; each of the two output ports of the 50/50 beam splitter is connected to a photon detector; the outputs of the photon detectors are connected to a digital logic unit such as a Field Programmable Gate Array (FPGA) which processes the measurement outcomes and can communicate a success/failure flag back to the nodes of a flow. Such a BSA is depicted in Figure~\ref{fig:optical_swap_BSA}. In Table~\ref{tab:sampleParams} (Section~\ref{sec:NumEval}) we provide an inventory of physical and protocol parameters motivated by implementations of single click bipartite HE generation based on experimental demonstrations with an Nitrogen Vacancy (NV) center in diamond \cite{SingleClickDiamond, ThreeNodeQN}.

\subsubsection{Heralded entanglement generation}
The goal of node pair $(n_i, n_j)$ running the bipartite HE generation protocol is to entangle a communication qubit of node $n_i$ with a communication qubit of node $n_j$. The term "bipartite" thus refers to the property of the protocol that the resulting entanglement involves two qubits. We will sometimes refer to such states as Einstein–Podolsky–Rosen (EPR) pairs or Bell states. The protocol is called \textit{heralded} because for every attempt to generate entanglement that the nodes make, a success or failure flag is generated and converted into a message that is sent to the nodes, thereby indicating if the attempt succeeded or failed. Triggering a subsequent attempt after a success would destroy the entanglement that was created. To prevent wasting entanglement, the protocol includes a wait time for the heralding flag to arrive before triggering subsequent attempts. In certain EGS operation modes an attempt may not be triggered if a success flag is received (see Section~\ref{subsec:single_ent_gen}). In a setting where node pairs use the entanglement generated between them to perform some application, this protocol is beneficial because it allows the nodes to condition commencement of the application on the successful generation of entanglement. 

At a high level, a single-click heralded entanglement generation protocol consists of four stages. First, each node performs a sequence of calibration operations and prepares a communication qubit in a known state. Second, each node locally triggers the generation of entanglement between the state of their communication qubit and the presence/absence of a travelling photon. Third, the presence/absence encoded photons are sent to a BSA, at which a Bell-State Measurement (BSM) (entanglement swap) is attempted between the encoded photons. Fourth, if the BSM succeeds the communication qubits of the two quantum processing nodes will have become entangled and a success flag is sent to the nodes. If the measurement is not successful a failure flag is communicated to the nodes. The second, third and fourth stages occurring sequentially constitute a single HE generation attempt. For the example physical platform of the NV center in diamond, the calibration operations correspond to a Charge and Resonance (CR) check \cite{ThreeNodeQN}.

Attempts may be repeated in batches which must be occasionally interleaved with repetition of the first step -- calibration of the communication qubit. The main limitation on the batched attempt repetition rate is the Round Trip Time (RTT) of communication associated with the third and fourth stages of an attempt. The need to wait for arrival of the heralding flag especially limits the rate. 

A communication sequence between the node pair $(n_i, n_j)$, which requests and is allocated use of an EGS resource to perform HE generation, is included in Appendix \ref{appendix:OpAppendix} in Figure \ref{fig:HeraldedCommSeq}.

\subsubsection{Probability of success, $p_{\text{gen}}$}
\label{subsubsec:pgen}
We model experimental implementations of HE generation where the state of the communication qubit is reset (stage one) at the start of each attempt and every attempt in a batch corresponds to an identical experimental sequence. Furthermore we assume that the characteristics of devices used in triggering entanglement generation attempts, such as laser pump power and frequency, remain constant. Therefore, the probability of entanglement generation may only change over attempts if there are physical parameters that drift or jump over a batch of attempts. For any system where attempts have a fixed mean duration that is short in comparison to the parameter drift/jump timescales, such effects may be accounted for by assuming that the probability of successful entanglement generation is a function of the $j$th attempt in a batch of attempts, $p_{\text{gen}}(j)$. 

For an implementation in the NV colour center in diamond, one may assume that the outcomes of sequential attempts in a batch are identically and independently distributed (IID), with a fixed probability of success $p_{\text{gen}}$ \cite{SingleClickDiamond}.
This assumption is valid as long as calibration periods are performed frequently enough between batches of attempts to prevent slow effects -- such as the spectral diffusion which affects solid state quantum emitters -- from corrupting the state of the communication qubit. 
The assumption that the outcomes of sequential attempts are IID with a fixed probability of success $p_{\text{gen}}$ also applies to other experimental platforms, such as Trapped Ions \cite{HerEntTrappedIons1, HerEntTrappedIons2}, where the mean attempt duration is significantly shorter than sources of parameter drift. The assumption that $p_{\text{gen}}$ is constant and is independent of attempt duration distributions but depends only on the mean attempt duration is the necessary condition that we use in proving the insensitivity result discussed in the introduction section (this is Theorem~\ref{thm:insensitivity} in Section~\ref{sec:analysis}).
\subsubsection{Single vs multiple entanglement generation}
\label{subsec:single_ent_gen}
For a limited quantum node, such as a node with one communication qubit and possibly a memory, it may be most practical to engage in single entanglement generation. That is, if an attempt to generate entanglement succeeds, no further attempts in a batch will be executed. Physically, successful entanglement generation renders the communication qubit of the device unavailable for further attempts until that entanglement can be used or transferred to memory. Transfers to memory are not instantaneous and have a finite time cost, thus communication qubits can not be freed instantly even in a system with memory. Moreover, if a communication qubit is coupled to a memory, attempts to generate entanglement while a state is stored in memory may damage the stored state due to induced decoherence \cite{ThreeNodeQN}. This effect results from a persistent non-zero coupling to the memory. Single entanglement generation may be the preferential operation mode of limited quantum devices to account for these effects \cite{ThreeNodeQN}. In contrast, a quantum node with multiple communication qubits may leverage them to generate multiple entangled states, possibly by multiplexing photon emission from the node.

%% file: model_assumptions.tex
In this section, we lay the groundwork for the analysis of a star-topology system with the EGS at its center, as shown in Figure~\ref{fig:star_topology}. To this end, we present a number of abstractly-defined terms which bridge the gap between the physical and queueing-theoretic models of the EGS. We then discuss the various EGS operation modes, and for each introduce the corresponding session structure. Modeling assumptions and useful notation are introduced throughout. We begin by defining a number of terms that are relevant to all EGS operation modes.

\begin{definition}{Service model.}
    The \emph{service model} describes how the EGS handles requests, including:
    %\begin{itemize}
        %\item[-] 
        $(1)$
        \emph{resource reservation}, specifying the amount and duration of resource allocation to a pair of communicating nodes;
        %\item[-] 
        $(2)$
        \emph{retrial behavior}, specifying actions taken upon blocked service events; and
        %\item[-] 
        $(3)$
        \emph{termination behavior}, specifying events that trigger the EGS to end service to a pair of nodes. 
    %\end{itemize}  
\end{definition}
Service models, described briefly in the introduction, will be explained in more detail later in this section; we first provide a description of the underlying components of service.

\begin{definition}{Call.}
A \emph{call} is the basic service component for two nodes communicating via the EGS. A call involves the active use of EGS components for a period of time, such as the utilization of a resource to attempt entanglement generation between two nodes.
\end{definition}
In the context of entanglement generation, we use the terms ``call'' and ``attempt'' interchangeably.
For the purpose of this work, we establish two additional service component types which are not calls in that EGS resources are not in active use for their duration. For the first of these, we follow the convention set by \cite{bonald2006erlang} to define \emph{idle periods}.

\begin{definition}{Idle period.}
%An \emph{idle period} between nodes $n_i$ and $n_j$ requires neither EGS, nor node hardware resources. 
When nodes $n_i$ and $n_j$ enter an \emph{idle period},
%Moreover, when an idle period begins, the nodes 
they
relinquish all EGS resources, so that a subsequent call would require a new service reservation.
\end{definition}
In this work, we will assume that a node's communication qubit is unavailable during an idle period, so that it cannot be used to initiate a new service request for the node until the current one completes service. We leave relaxations of this assumption for future work.

In our EGS model, there is also a possibility that of two communicating nodes $n_i$ and $n_j$, one or more necessitates a calibration period after a number of active (call) periods. The duration of such calibration periods typically has a finite mean, albeit it can be randomly distributed. Depending on the service model, the nodes may decide to not relinquish resources they have already reserved at the EGS.
\begin{definition}{Calibration period.}
    A \emph{calibration period} between nodes $n_i$ and $n_j$ requires node hardware resources from one or both nodes, but no active utilization of EGS resources. Depending on the service model, however, the nodes may continue to hold onto EGS resources for the duration of a calibration period, precluding other nodes from accessing them.
\end{definition}
We assume in this work that a calibration period engages all qubits of the corresponding  communication request, as opposed to, e.g., all qubits of a node. While the latter scenario may also be of interest, it poses a challenge for analysis since service requests can no longer be treated independently from each other.
We also remark that if the nodes relinquish EGS resources at the beginning of a calibration period, then from the perspective of the EGS (and from a modeling perspective) the period is an idle one. We distinguish between calibration and idle period types because calibration periods are always physically motivated: they necessarily engage quantum hardware at nodes.
%, which may impose restrictions on the number of new service requests originating from the nodes. 
In contrast, idle periods need not stem from quantum hardware restrictions at nodes (even if we assume that communication qubits are unavailable for new service request creation during idle periods). Examples of these are an entanglement generation attempt followed by a \emph{classical} processing period at the nodes, or a link-layer protocol that imposes a back-off timer between successive entanglement generation attempts.

A request issued from nodes $n_i$, $n_j$ to the EGS, if accepted, triggers the creation of a \emph{session}.
\begin{definition}{Session.}\label{def:session} A \emph{session} between nodes $n_i$ and $n_j$, denoted by the tuple $(n_i,n_j,t)$ is a sequence of calls generated by the two nodes, intended to be serviced by the EGS. 
This sequence may be interleaved with idle and/or calibration periods. Tuple element $t$ specifies the session type.
%, whereas element $s$ denotes the number of active sessions of type $t$ between $n_i$ and $n_j$.
\end{definition}
By convention, sessions begin and end with calls, and not idle or calibration periods.
The definition above makes a reference to a \emph{session type}: as in \cite{bonald2006erlang}, we permit the existence of differently-structured sessions within one system. Physically, these may correspond to different applications, entanglement distribution algorithms, or even application \emph{instantiations}. An example of the latter is a pair of nodes $(n_1,n_2)$ that assist in carrying out two QKD instantiations across the network shown in Figure \ref{fig:larger_net}: one for the user pair $(A,B)$ for a requested key size $M_1$, and the other for the user pair $(C,D)$ for a requested key size $M_2> M_1$. The session type corresponding to $C$ and $D$'s request would have more calls (i.e., entanglement generation attempts), and possibly more calibration or idle periods, depending on physical system restrictions or the algorithmic design of the EGS control protocol, respectively. Another example is the use of session types to accommodate application-dependent fidelity requirements: an application with a high minimum fidelity threshold can for instance request a session with more attempts in the hopes of producing enough states from which to distill higher-fidelity entanglement \cite{bennett1996purification,deutsch1996quantum}.
We use the shorthand notation $f^{t}_{i,j}$
 to identify a session of type $t$ between nodes $n_i$ and $n_j$. The set of all possible flows is denoted as 
 \begin{align}
\mathcal{F}\equiv \{f^t_{i,j}:i,j\in\{1,\dots,K\},t\in\{1,\dots,T\}\},  
 \end{align}
 where $K$ is the number of nodes connected to the EGS and $T$ is the total number of possible session types. The cardinality of $\mathcal{F}$ is given by $F=\binom{K}{2}\times T$ since each node-pair could run any of the $T$ session types.\footnote{It is possible that certain node-pairs do not have the physical capability to carry out sessions of a given type $t$. The state space can be easily amended to reflect such restrictions; we make the assumption that all node-pairs are capable of all session types merely to simplify notation.} 
 We note that sessions need not be unique: $n_i$ and $n_j$ can have multiple concurrent sessions of type $t$, as long as resources (communication qubits and EGS resources) are available. Similar to the work of Bonald in \cite{bonald2006erlang}, we assume that sessions are independent, their arrivals are Poisson, and that session components (calls, idle and calibration periods) are decomposed into a finite and random number of exponential phases. 
 %We also adopt the convention of \cite{bonald2006erlang} that all sessions begin and end with calls (and not idle or calibration periods).

We note that Def. \ref{def:session} is not all-encompassing. As an example, a more general definition would allow a session to use multiple communication qubits from both nodes, enabling them to achieve a higher application rate (e.g., QKD secret key rate). Introducing this generalization, however, prompts several questions on the service model of the EGS. Consider a session that uses $c>1$ communication qubits from participating nodes $n_i$ and $n_j$. In one possible service model, the EGS accepts the session only when it can allocate exactly $c$ resources -- one for each $n_i/n_j$ qubit pair -- at the time of the request. In another service model, multiplexing involves all $c$ pairs of communication qubits to share access to a single EGS resource via carefully timed quantum signals and classical bookkeeping messages. Since analysis depends on the choice of service model, as well as on physical system operation details, we leave such generalization for a future study and focus here on the case where each active session -- one that has been accepted for service by the EGS -- uses a single qubit from each participating node. 

Throughout this work, we often use the term ``session" and the  (queueing theory inspired) term ``flow" interchangeably. Here, we also point out the difference between a request and a session/flow. As discussed in the previous section, a request is a demand from a pair of nodes $n_i$, $n_j$, for a number of attempts to be carried out using a resource of the EGS. Meanwhile, a session \emph{models} a request as a sequence of calls which arrive in batches and which may be interspersed with idle and calibration periods. Higher layer algorithms or applications determine the session type, while the service model specifies EGS and node behaviors for a session in progress.
\if{false}
\begin{table}[]
    \centering
    \begin{tabular}{c|l}
    Variable Name & Definition\\
    \hline
        $K$ & number of nodes connected to EGS \\
        $c_k$ & number of communication qubits at node $n_k$\\
        $C$ & number of resources at EGS\\
        $T$ & number of possible session types\\
        $F$ & number of possible flows, $|\mathcal{F}|$
    \end{tabular}
    \caption{Variable definitions used in the description of the EGS setup.}
    \label{tab:varible_defs}
\end{table}
\fi
 We study three distinct service models of the EGS:
 \begin{itemize}
     \item[-] \textit{Single EPR Pair Generation with Strict Resource Reservation}: a session consists of entanglement generation attempts interleaved with calibration periods.
     Once a session is admitted at the EGS, attempts are carried out until one is successful, or until the last attempt is complete. Both events result in session termination. A flow holds onto its EGS resource during calibration periods, even though it is not actively utilized. 
     \item[-] \textit{Multiple EPR Pair Generation with Strict Resource Reservation}: all properties of the previous service model apply, with the exception that a session terminates only when all attempts are carried out. A session can thus produce multiple EPR pairs.
     \item[-] \textit{Multiple EPR Pair Generation with Resource Relinquishment}: a session consists of entanglement generation attempts, or more generally ``active'' periods which engage an EGS resource, interleaved with idle periods at the beginning of which the EGS resource is given up, and at the end of which the flow attempts to re-obtain a resource. Failure to obtain a resource (either at the beginning of a session or after an idle period) triggers a jump-over retrial: the session either transitions to the next idle period, or if one does not exist, terminates. As in the previous service model, successful entanglement generation by itself does not cause session termination.
 \end{itemize}
% \gv{
% Discussion: say we have QKD in mind primarily, esp. for all-photonic system...may make it easier for readers to appreciate the practical applicability of model.} \sg{I've actually had some thoughts on this: 1. I think this model may actually by nice as a lightweight (in terms of implementation cost and calculation of the schedule) scheduling model for real quantum networks. In a near term quantum network this model would be advantageous because the implementation cost of such a scheduler is low. In a scaled network, this scheduler is advantageous because the cost of calculating the schedule is independent of the number of flows. This scheduler results in an analytic QoS marker (the blocking probability), which is exceptional in the domain of quantum network scheduling and may be very useful for obtaining a way to practically delivery requested rates of entanglement packets to flows, I plan to make this point in the numerical analysis as well. Also, because we obtain analytic results for this method, it can be used as a baseline which any other method of scheduling can be compared to in terms of real service metrics derivable from that which we have shown -- the blocking probability. 2. here we have focused a lot on considerations that are especially relevant to Heralded Entanglement and our initial results are in this direction so it might be nice to not primarily link to QKD but rather extend the discussion with QKD and the CI model. }

Irrespective of the service model, the state space of the system can be represented using a vector
\begin{align}
    \vb{x} = [\vb{x}^{f_1},\dots,\vb{x}^{f_F}],
    \label{eq:cox_state}
\end{align}
where $F$ is the number of possible flows, and each $\vb{x}^f$ is a vector describing the number of jobs in the queues corresponding to flow $f\in\mathcal{F}$. Let us now examine the structure of each vector $\vb{x}^f$; modeling each component of a flow using a Coxian distribution (see discussion in Section~\ref{sec:background} for a justification), we define the following variables:
\begin{itemize}
    \item[-] $A^f_{i,j}/C^f_{i,j}/I^f_{i,j}$: $i$th phase of attempt/calibration/idle period $j$, for flow $f$;
    \item[-] $N_A^f/N_C^f/N_I^f$: number of phases per attempt/calibration/idle period, of flow $f$;
    \item[-] $M_A^f/M_C^f/M_I^f$: total number of attempts/calibration periods/idle periods, for flow $f$.
\end{itemize}
We further define the following useful variables and notation: 
\begin{itemize}
    \item[-] $L^f\equiv N_A^f\times M_A^f+N_C^f\times M_C^f+N_I^f\times M_I^f$ as the total number of phases in session type $f$;
    \item[-] $L\equiv 
 \sum\limits_{f\in\mathcal{F}}L^f$ -- the dimension of vector $\vb{x}$;
 \item[-] $x^{f,A/C/I}_{i,j}$ is the notation used when referring to the number of jobs currently in the $i$th phase of the $j$th attempt/calibration/idle period of a session belonging to flow $f$;
 \item[-] $\vb{e}^{f,A/C/I}_{i,j}$ are vectors of dimension $L$, with all entries zero except the one corresponding to the $i$th phase of attempt/calibration/idle period $j$ of flow $f$, which is one;
 \item[-] $\vb{e}^f_i$, $i\in\{1,\dots,L\}$, is a vector of dimension $L$ with all entries zero except the one corresponding to the $i$th component of $\vb{x}^f$, which is equal to one;
 \item[-]
 $x_i$, $i\in\{1,\dots,L\}$, is used when referring to the $i$th element of the vector $\vb{x}$ when there is no need to identify a flow or period within it;
 \item[-]
 $x^f_i$, $i\in\{1,\dots,L^f\}$, is used when referring to the $i$th element of $\vb{x}^f$ when there is no need to identify a specific period or phase within a flow $f$.
 \item[-] For a given state $\vb{x}$, the number of active sessions of type $f$ is denoted by
 \begin{align}
     q^f(\vb{x}^f) \equiv \sum\limits_{i=1}^{L^f}x_i^f.
     \label{eq:qfs}
 \end{align}
% Note that for a flow $f$ with member nodes $n_i$ and $n_j$, $q^f(\vb{x})\leq c^f_{\min}$, where $c^f_{\min}=\min(c_i,c_j)$, since the most resource-constrained node determines the maximum number of contemporaneous sessions.
\end{itemize}

Finally, we assume that quantum nodes have a limited number of communication qubits -- $c_k$ for node $n_k$, and that 
%For a given flow $f=(n_i,n_j)$, define $c^f_{\min}=\min(c_i,c_j)$.
the EGS has a total of $C$ resources. %Table~\ref{tab:varible_defs} summarizes the variables which are used in the analysis.

%% file: strict_reservation_analysis_acm.tex
\subsection{Single EPR Pair Generation with Strict Resource Reservation Service Model}
\label{sec:single_epr_strict}
Recall that in the strict resource reservation service model, sessions, once admitted, are processed in their entirety, terminating prematurely only if an EPR pair attempt is successful. This means that $(i)$ by definition, there is no notion of an idle period in this service model, and $(ii)$ a session that is being actively serviced by the EGS does not relinquish its BSA even during  a calibration period. Figure~\ref{fig:strict_res_periods} provides the general form of a session within this service model. In Figure~\ref{fig:strict_res_phases} the attempt and calibration periods are depicted in their decomposed form; \eg, exponential phases $A_{1,1},\dots,A_{N_A,1}$ comprise the Cox-distributed period {\chancery A}$_1$.

%It finally remains to restrict the state space of the Markov chain. 
Given these specifications, we can define the state space of the associated continuous-time Markov chain. Namely, the EGS must heed the resource (BSA) capacity, and each network node must heed its own communication qubit limit.
The admissible state space is thus given by
\begin{align}
    \mathcal{S} = \left\{\vb{x}\in \mathbb{N}^L: \sum\limits_{f\in\mathcal{F}}q^f(\vb{x}^f) \leq C,\quad \sum\limits_{f\in\mathcal{F}:n_k\in f}q^f(\vb{x}^f)\leq c_k,\forall k\in \{1,\dots,K\}\right\},
    \label{eq:restr_state_space}
\end{align}
where we introduce the notation $n_k\in f$ to mean that node $n_k$ partakes in flow $f$. The set $\mathcal{S}$ is coordinate convex, i.e., if $\vb{x}\in\mathcal{S}$, then  $\vb{y}\in\mathcal{S}$ for all $\vb{y}$ such that $\vb{0}\leq \vb{y}\leq \vb{x}$ component-wise.

\begin{figure}
    \centering
    \includegraphics[width=\textwidth]{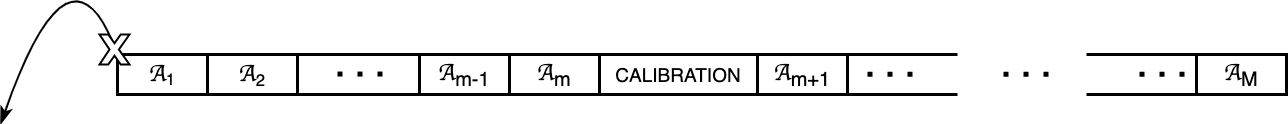}
    \caption{Strict resource reservation service model. A session consists of multiple EPR pair generation attempts, denoted by {\chancery A}$_i$, $i=1,\dots,M$. A calibration period is carried out after every $m$ attempts. In the ``multiple EPR pair generation'' variant of this service model, an admitted session does not relinquish resources for its entire duration, while in the ``single EPR pair generation'' variant, the session ends after a successful attempt.}
    \vspace{-5mm}
    \label{fig:strict_res_periods}
\end{figure}
\begin{figure}
    \centering
    \begin{minipage}[c]{0.55\textwidth}
    \includegraphics[width=\textwidth,trim={1cm 4cm 1.2cm 3cm},clip]{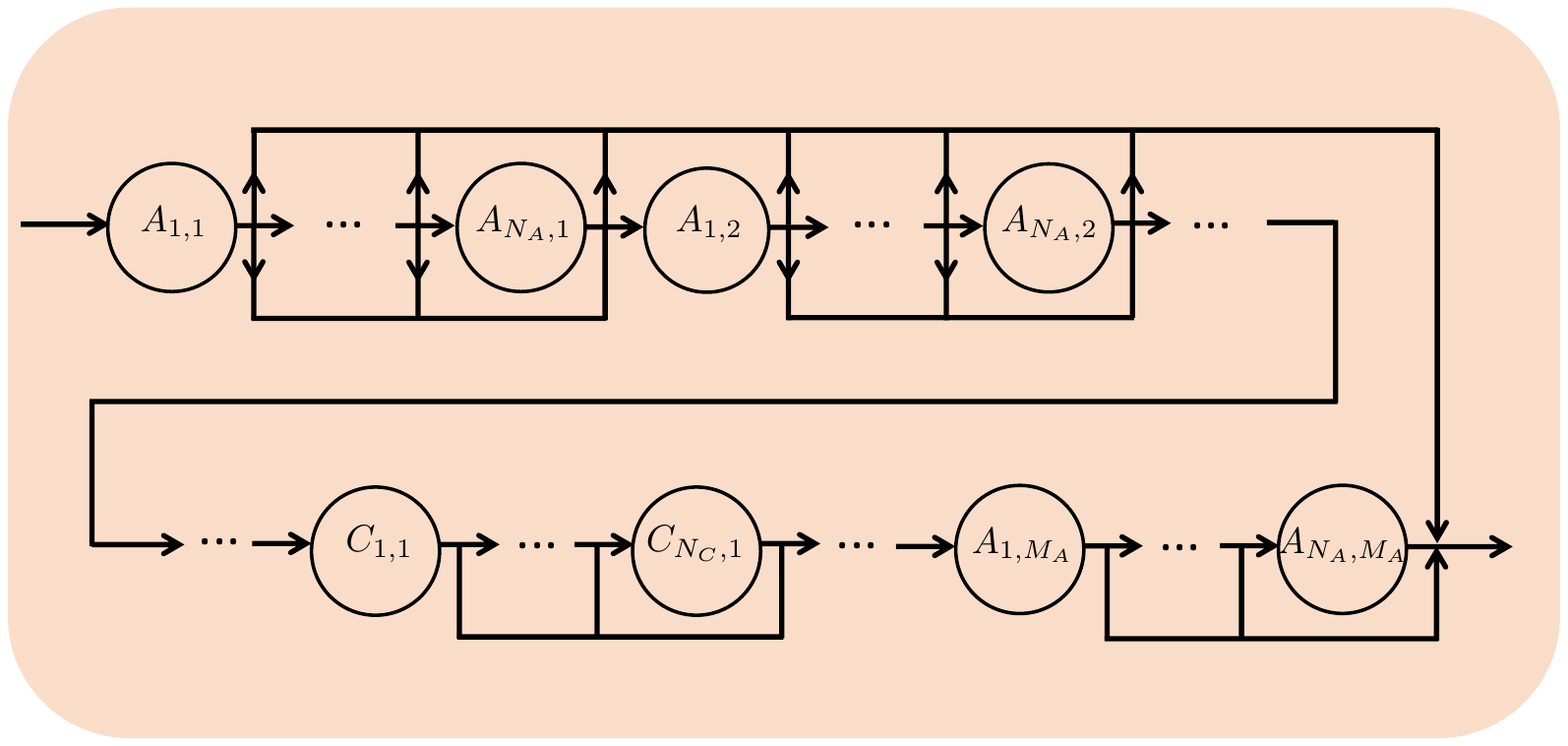}
    \end{minipage}\quad
    \begin{minipage}[c]{0.4\textwidth}
    \caption{A session from the single EPR pair generation strict resource reservation service model, shown at the level of periods in Figure~\ref{fig:strict_res_periods}, decomposed into exponentially-distributed phases so as to result in Coxian-distributed attempt and calibration periods.}
    \label{fig:strict_res_phases}
    \end{minipage}
\end{figure}

For the queueing network in Figure~\ref{fig:strict_res_phases}, under the strict resource reservation model we specify the following properties:
\begin{itemize}
    \item[-] All external arrival rates (\ie, those originating from outside of the network) $\nu_i^f(\vb{x})$ are zero, except for those of queues $A_{1,1}^f$, $f\in\mathcal{F}$. We denote these rates with $\nu_1^f(\vb{x})$, $f\in\mathcal{F}$, $\vb{x}\in\mathcal{S}$, so that the transition from $\vb{x}$ to $\vb{x}+\vb{e}^{f,A}_{1,1}$ occurs with rate 
    \begin{align}
        \nu_1^f(\vb{x}) = \begin{cases}
            \nu_1^f, &\text{ if } \vb{x}+\vb{e}^{f,A}_{1,1} \in\mathcal{S},\\
            0, & \text{ else.}
        \end{cases}
    \end{align}
    \item[-] Transition $\vb{x}\to\vb{x}-\vb{e}^f_i+\vb{e}^f_{i+1}$ occurs with probability $p^f_{i,i+1}$, for $1\leq i<L^f$.
    \item[-] A special case of the above is that $p^f_{i,i+1}=1$, $\forall f\in\mathcal{F}$, if $x^f_i$ corresponds to the last phase of a calibration period.
    \item[-] Transition 
    $\vb{x}\to\vb{x}-\vb{e}^f_i+\vb{e}^f_{j}$
    occurs with probability $p^f_{i,j}$ if $j$ is such that $x^f_j$ corresponds to the initial phase of the call/calibration period that follows the call/calibration period corresponding to $x^f_i$.
    \item[-] Transition $\vb{x}\to\vb{x}-\vb{e}^f_i$ occurs with probability $p^f_i$ if $i$ is such that $x^f_i$ corresponds to an attempt phase. This transition represents the event that entanglement generation succeeds after the $i$th phase of flow $f$ -- in this service model, this causes the session to end. We note that leaving the session from a calibration phase is not possible.
\end{itemize}
\if{false}
Let us define the following two sets:
\begin{itemize}
\item[-] $\mathcal{C}^f$ is the set that contains all call (attempt) phases of flow $f$, i.e., $\mathcal{C}^f=\{A^f_{i,j}, i\in \{1,\dots,N_A^f\}, j\in\{1,\dots,M^f_A\}\}$;
    \item[-] $\mathcal{I}_{\mathcal{C}^f}$ is the set that contains indices corresponding to call phases of flow $f$, i.e., $\mathcal{I}_{\mathcal{C}^f} = \{i: x^f_i \text{ corresponds to a call phase of } f\}$.
    \end{itemize}
    \fi

    Finally, we define $\mu^f_{i}$ as flow $f$'s job processing rate at queue $i$ of flow $f$, and $\lambda^f_i(\vb{x})$ as the total arrival rate into queue $i$ of flow $f$ while in state $\vb{x}$. When $i=1$, i.e., it corresponds to the first phase of a session, the total arrival rate is $\lambda^f_1(\vb{x}) = \nu_1^f(\vb{x})$. For all other queues, the arrival rate is given by
    \begin{align}
    \lambda_i^f(\vb{x}) = 
    %\nu^f_1(\vb{x})\tilde{p}^f_2\dots \tilde{p}^f_i \equiv 
    \begin{cases}
    \nu^f_1\tilde{p}^f_i \equiv
        \lambda_i^f, \text{ if } \vb{x}+\vb{e}_i^{f}\in\mathcal{S},\\
        0, \text{ else.}
    \end{cases}
    \label{eq:lambda_i_cox}
    \end{align}
    Above, $\tilde{p}^f_i$, $2\leq i\leq L^f$, denotes the probability of reaching the $i$th phase starting from the first phase of a session belonging to flow $f$. Appendix~\ref{app:prob_reaching_queue} provides a derivation of these probabilities.

%\subsubsection{Stationary Distribution and Blocking Probability}
We next begin the analysis of the above-described system, beginning with the derivation of the stationary distribution.
\begin{theorem}The stationary distribution $\pi(\vb{x})$ of the system with single EPR pair generation while in strict resource reservation mode is given by
    \begin{align}
        \pi(\vb{x}) = \left\{\left(\sum\limits_{\vb{y}\in\mathcal{S}}\prod\limits_{f\in\mathcal{F}}\prod\limits_{i=1}^{L^f}\frac{(\rho^f_i)^{y^f_i}}{y^f_i!}\right)^{-1}\right\}\prod\limits_{f\in\mathcal{F}}\prod\limits_{i=1}^{L^f}\frac{(\rho^f_i)^{x^f_i}}{x^f_i!},
        \label{eq:cox_stationary_full}
    \end{align}
    where $\rho_i^f$ is the traffic intensity of the $i$th queue of a session corresponding to flow $f$.
\end{theorem}
\begin{proof}
To obtain the stationary probability $\pi(\vb{x})$ of the system with state space $\mathcal{S}$ as defined in (\ref{eq:restr_state_space}), we apply the local balance approach. Henceforth, we adopt the convention that $\pi(\vb{x})=0$ for any $\vb{x}\notin \mathcal{S}$. Consider any state $\vb{x}\in\mathcal{S}$; we have that
\begin{itemize}
    \item[-] The rate of leaving $\vb{x}$ due to outside arrival coming into the system is given by
    \begin{align}
        %A = \pi(\vb{x})\sum\limits_{f\in\mathcal{F}: q^f(\vb{x})<c^f_{\min}}\nu^f_1,
        A = \pi(\vb{x})\sum\limits_{f\in\mathcal{F}}\nu^f_1(\vb{x});
    \end{align}
    \item[-] The rate of entering $\vb{x}$ due to job departure to outside of the queueing network is given by
    \begin{align}
        %A^{\prime} = \sum\limits_{f\in\mathcal{F}}\sum\limits_{i\in \mathcal{I}_{\mathcal{C}^f}}\pi(\vb{x}+\vb{e}^f_i)(x^f_i+1)p^f_i\mu^f_i,
        A^{\prime} = \sum\limits_{f\in\mathcal{F}}\sum\limits_{i=1}^{L^f}\pi(\vb{x}+\vb{e}^f_i)(x^f_i+1)p^f_i\mu^f_i,
    \end{align}
    where we take advantage of our convention that $\pi(\vb{x}+\vb{e}^f_i)=0$ if $\vb{x}+\vb{e}^f_i \notin \mathcal{S}$, as well as use the fact that $p_i^f=0$ if $i$ corresponds to a calibration phase of a session belonging to flow $f$.
    %where the second sum is indexed over elements of the set $\mathcal{I}_{\mathcal{C}^f}$ because departures to the outside can only occur from attempt phases;
    \item[-] The rate of leaving $\vb{x}$ due to departure from queue $i$ of flow $f$ is given by
    \begin{align}
        B^f_i = \pi(\vb{x})x^f_i\mu^f_i;
    \end{align}
    \item[-] The rate of entering $\vb{x}$ due to an arrival at queue $i$ of flow $f$ is given by
    \begin{align}
        B^{\prime f}_i = \pi(\vb{x}-\vb{e}^f_i)\nu^f_i(\vb{x}-\vb{e}^f_i)+\sum\limits_{j\neq i}\pi(\vb{x}+\vb{e}^f_j-\vb{e}^f_i)(x^f_j+1)p^f_{j,i}\mu^f_j.
    \end{align}
\end{itemize}
To obtain the stationary distribution of the system, we solve the following equations and then show that this solution is in fact the stationary distribution of the system, 
%if $\vb{x}+\vb{e}^f_i \in\mathcal{S}$, then
\begin{align}
    \pi(\vb{x}+\vb{e}^f_i) = \frac{\lambda_i^f(\vb{x})}{\mu_i^f(x_i^f+1)}\pi(\vb{x}) =
    \begin{cases}
        \frac{\rho^f_i}{x_i^f+1}\pi(\vb{x}), &\text{if } \vb{x}+\vb{e}^f_i \in\mathcal{S},\\
        0, &\text{else}
    \end{cases}\label{eq:pi_guess1}
\end{align}
where ${\rho_i^f\equiv \lambda^f_{i}/\mu^f_{i}}$.
Substituting this expression into $A=A^{\prime}$ and simplifying, we obtain
\begin{align}
    %\sum\limits_{f\in\mathcal{F}}\nu^f_1(\vb{x}) &= \sum\limits_{f\in\mathcal{F}}\sum\limits_{i\in \mathcal{I}_{\mathcal{C}^f}}\lambda^f_i(\vb{x})p^f_i.
    \sum\limits_{f\in\mathcal{F}}\nu^f_1(\vb{x}) &= \sum\limits_{f\in\mathcal{F}}\sum\limits_{i=1}^{L^f}\lambda^f_i(\vb{x})p^f_i.
    \label{eq:12}
\end{align}
The expression above is the traffic conservation equation: the aggregate arrival rate into the queueing network while in state $\vb{x}$ equals the overall departure rate from it.
Continuing the application of local balance, we require that for all $i$ and $f$, $B^f_i = B^{\prime f}_i$. Note that if $\vb{x}=\vb{0}$, then both $B_i^f$ and $B_i^{\prime f}$ are zero, since we cannot leave this state due to a departure from any queue, and we cannot enter this state due to an arrival at any queue, respectively, since there are no active sessions while in state $\vb{0}$. For any other $\vb{x}\in \mathcal{S}$, we require
\begin{align}
\pi(\vb{x})x^f_i\mu^f_i = \pi(\vb{x}-\vb{e}^f_i)\nu^f_i(\vb{x}-\vb{e}^f_i)+\sum\limits_{j\neq i}\pi(\vb{x}+\vb{e}^f_j-\vb{e}^f_i)(x^f_j+1)p^f_{j,i}\mu^f_j.
\label{eq:Bi}
\end{align}

%Before continuing the analysis, we first check the equation above for all possible scenarios while omitting $f$ for cleanliness. First, suppose that $\vb{x}-\vb{e}_i\notin \mathcal{S}$. This means that $\pi(\vb{x}-\vb{e}_i)=0$. By coordinate convexity of the admissible state space, this also means that $\pi(\vb{x})=0$ and $\pi(\vb{x}-\vb{e}_i+\vb{e}_j) =0$, $\forall j$. The equation thus still holds in this scenario. Next, suppose that $\vb{x}-\vb{e}_i\in \mathcal{S}$, but $\vb{x}\notin \mathcal{S}$. The latter means that $\pi(\vb{x})=0$, and also that $\nu_i(\vb{x}-\vb{e}_i)=0$. Further, since a job going from queue $j$ to queue $i$ while in state $\vb{x}-\vb{e}_i+\vb{e}_j$ would cause a transition to state $\vb{x}$, we have that $p_{j,i}(\vb{x}-\vb{e}_i+\vb{e}_j)=0$, too. Thus, the equality holds in this case as well. 

Using (\ref{eq:pi_guess1}) with (\ref{eq:Bi}) and simplifying, we obtain
\begin{align}
    %\pi(\vb{x}-\vb{e}^f_i)\frac{\lambda_i^f(\vb{x}-\vb{e}_i^f)}{\mu_i^fx_i^f}x^f_i\mu^f_i &= \pi(\vb{x}-\vb{e}^f_i)\nu^f_i(\vb{x}-\vb{e}^f_i)+\sum\limits_{j\neq i}\pi(\vb{x}-\vb{e}^f_i)\frac{\lambda_j^f(\vb{x}-\vb{e}^f_i)}{\mu_j^f(x_j^f+1)}(x^f_j+1)p^f_{j,i}\mu^f_j,\\
    \pi(\vb{x}-\vb{e}^f_i)\lambda_i^f(\vb{x}-\vb{e}_i^f) &= \pi(\vb{x}-\vb{e}^f_i)\nu^f_i(\vb{x}-\vb{e}^f_i)+\sum\limits_{j\neq i}\pi(\vb{x}-\vb{e}^f_i)\lambda_j^f(\vb{x}-\vb{e}^f_i)p^f_{j,i}.
\end{align}
If $\vb{x}-\vb{e}^f_i \notin \mathcal{S}$, then $\pi(\vb{x}-\vb{e}^f_i)=0$ and the equation above holds; else we obtain
\begin{align}
    \lambda_i^f(\vb{x}-\vb{e}_i^f) &= \nu^f_i(\vb{x}-\vb{e}^f_i)+\sum\limits_{j\neq i}\lambda_j^f(\vb{x}-\vb{e}^f_i)p^f_{j,i}.
\end{align}
%\begin{align}
%    \pi(\vb{x})x^f_i\mu^f_i &= \pi(\vb{x})\frac{x^f_i}{\rho^f_i}\nu^f_i+\sum\limits_{j\neq i}\pi(\vb{x})\frac{x^f_i}{\rho^f_i}\frac{\rho^f_j}{(x^f_j+1)}(x^f_j+1)p^f_{j,i}\mu^f_j,\\
    %\lambda^f_i(\vb{x}) &= \nu^f_i(\vb{x})+\sum\limits_{j\neq i}\lambda^f_j(\vb{x})p^f_{j,i}.
%\end{align}
We have thus recovered the traffic equations for each individual queue in the network: the total arrival rate into queue $i$ of flow $f$ equals the sum of the exogenous arrival rates and the arrival rates from other queues. Thus our solution satisfying (\ref{eq:pi_guess1}) represents the stationary distribution.

Next, from (\ref{eq:pi_guess1}), it follows that
\begin{align}
    \pi(\vb{x}) = D\prod\limits_{f\in\mathcal{F}}\prod\limits_{i=1}^{L^f}\frac{(\rho^f_i)^{x^f_i}}{x^f_i!},
    \label{eq:pi_guess2}
\end{align}
where $D=\pi(\vb{0})$. Using the fact that $\sum\limits_{\vb{x}\in\mathcal{S}}\pi(\vb{x})=1$, we can solve for $D$, which results in the stationary distribution (\ref{eq:cox_stationary_full}),
\if{false}
\begin{align}
    \pi(\vb{x}) = \left\{\left(\sum\limits_{\vb{y}\in\mathcal{S}}\prod\limits_{f\in\mathcal{F}}\prod\limits_{i=1}^{L^f}\frac{(\rho^f_i)^{y^f_i}}{y^f_i!}\right)^{-1}\right\}\prod\limits_{f\in\mathcal{F}}\prod\limits_{i=1}^{L^f}\frac{(\rho^f_i)^{x^f_i}}{x^f_i!},
\end{align}
\fi
where the term inside the curly brackets is $\pi(\vb{0})$. Note from (\ref{eq:cox_stationary_full}) that if a subset of flows $\mathcal{F}^{\prime}\subseteq \mathcal{F}$ never wish to generate requests -- i.e., $\nu_1^f(\vb{x})=0, \forall \vb{x}\in\mathcal{S}, f\in\mathcal{F}^{\prime}$ -- then these flows' traffic intensities are equal to $0$. If we take the convention that $0^0\equiv 1$, then (\ref{eq:cox_stationary_full}) supports this setting by considering only flows from $\mathcal{F}\setminus \mathcal{F}^{\prime}$ when $\vb{x}\in\mathcal{S}$, and
causing $\pi(\vb{x})$ to equal zero whenever $x_i^f\neq 0$ for $f\in\mathcal{F}^{\prime}$.
\end{proof}

Having determined the stationary distribution, we are now ready to derive the blocking probability for a request from a given flow. 
\begin{theorem}
\label{thm:blocking_flow}
    For the system with single EPR pair generation operating in strict reservation mode, the probability that an arriving request belonging to flow $i\in\{1,\dots,F\}$ is blocked is given by 
    \begin{align*}
    \overline{\pi}_i(C) = 
    \left(\sum\limits_{\vb{q}\in\mathcal{Q}^{\prime}(i)}\prod\limits_{j=1}^F\frac{\left(\rho^{f_j}\right)^{q_j}}{q_j!}\right)^{-1}\sum\limits_{\vb{q}\in\mathcal{Q}(C)\bigcap \mathcal{Q}^{\prime}(i)}\prod\limits_{j=1}^F\frac{\left(\rho^{f_j}\right)^{q_j}}{q_j!}, 
    \end{align*}
    where $\vb{q} = [q_1,\dots,q_F]$ represents the number of active sessions $q_i$ from each flow $f_i$,
    \begin{align*}
    \quad\mathcal{Q}(h) &\coloneqq \left\{\vb{q}=[q_1,\dots,q_F] \in \mathbb{Z}_+^F: \sum\limits_{i=1}^Fq_i=h,~ \sum\limits_{i: n_k\in f_i}q_i\leq c_k,\forall k\in\{1,\dots,K\}\right\},\text{ and} \\ %\label{eq:set_Q}\\
    \mathcal{Q}^{\prime}(i) &\coloneqq \left\{\vb{q}=[q_1,\dots,q_F] \in \mathbb{Z}_+^F: \sum\limits_{j: n_k\in f_j}q_j\leq c_k,\forall n_k\notin f_i,~ \sum\limits_{j:n_l\in f_j}q_j < c_l, n_l\in f_i\right\}, %\label{eq:set_Qprime}
    \end{align*}
    where $\mathbb{Z}_+$ indicates the set of non-negative integers.
\end{theorem}
\begin{proof}
To begin, let $\mathrm{P}(\vb{q})$ denote the probability that EGS resources are occupied according to $\vb{q}$.
Consider the following two events: $\Omega_1(h)$ is the event that $h$ EGS resources are occupied, and $\Omega_2(i)$ is the event that flow $f_i$ has available communication qubits.  
By the PASTA (Poisson Arrivals See Time Averages) property, we can write the probability that an arriving request of flow $f_i$ sees $h$ occupied resources (conditioned on the flow having enough qubits to generate a request), i.e., $\mathrm{P}(\Omega_1(h)|\Omega_2(i))$, as
\begin{align}
    \overline{\pi}_i(h) &= %\frac{\mathrm{P}(\Omega_1(h)\cap \Omega_2(i))}{\mathrm{P}(\Omega_2(i))}= \frac{1}{\sum\limits_{\vb{q}\in\mathcal{Q}^\prime(i)}\sum\limits_{\vb{x}:\vb{q}(\vb{x})=\vb{q}}\pi(\vb{x})}\sum\limits_{\vb{q}\in \mathcal{Q}(h)\cap \mathcal{Q}^{\prime}(i)}~\sum\limits_{\vb{x}:\vb{q}(\vb{x})=\vb{q}}\pi(\vb{x}),
    \frac{\mathrm{P}(\Omega_1(h)\cap \Omega_2(i))}{\mathrm{P}(\Omega_2(i))}= \left(\sum\limits_{\vb{q}\in\mathcal{Q}^\prime(i)}\mathrm{P}(\vb{q})\right)^{-1}\sum\limits_{\vb{q}\in \mathcal{Q}(h)\cap \mathcal{Q}^{\prime}(i)}\mathrm{P}(\vb{q}).\label{eq:prob_h}
    %,
    %~\text{where}
    %\label{eq:prob_h}\\
    %\mathcal{Q}(h) &\coloneqq \left\{\vb{q}=[q_1,\dots,q_F] \in \mathbb{Z}_+^F: \sum\limits_{i=1}^Fq_i=h,~ \sum\limits_{i: n_k\in f_i}q_i\leq c_k,\forall k\in\{1,\dots,K\}\right\},\text{ and}\label{eq:set_Q}\\
    %\mathcal{Q}^{\prime}(i) &\coloneqq \left\{\vb{q}=[q_1,\dots,q_F] \in \mathbb{Z}_+^F: \sum\limits_{j: n_k\in f_j}q_j\leq c_k,\forall n_k\notin f_i,~ \sum\limits_{j:n_l\in f_j}q_j < c_l, n_l\in f_i\right\}.\label{eq:set_Qprime}
\end{align}
%\gv{In blocking probability: probability that a request gets blocked conditioned on the event that it can be generated (enough communication qubits available for associated flow): $P(A|B) = P(A\bigcap B)/P(B)$, where B = event that communication qubits available for request in question, A = event that all resources at EGS are occupied.}
The set $\mathcal{Q}(h)$ contains all possible combinations of active sessions such that the total number of active sessions is exactly $h$, and communication qubit constraints are not violated. The set $\mathcal{Q}^{\prime}(i)$ contains all combinations of active sessions such that communication qubit constraints are not violated, with the additional constraint that nodes belonging to flow $f_i$ have at least one unoccupied communication qubit each. 
%The intersection of $\mathcal{Q}(h)$ with the set $\mathcal{Q}^{\prime}(i)$ ensures that in the resulting set there exists at least one flow with available communication qubits. 
Here we use the fact that flow $f_i$ would not be able to initiate a session if any node associated with $f_i$ does not have an unoccupied communication qubit.
\if{false}
\begin{align}
    \overline{\pi}(h) = \sum\limits_{q_1=H_1(h)}^{G_1(h)}\quad\sum\limits_{q_2=H_2(h,q_1)}^{G_2(h,q_1)}\cdots\sum\limits_{q_F=H_F(h,q_1,\dots,q_{F-1})}^{G_F(h,q_1,\dots,q_{F-1})}~\sum\limits_{\vb{x}:\vb{q}(\vb{x})=\vb{q}}\pi(\vb{x}),
    \label{eq:prob_h}
\end{align}
where $\vb{q}$ is a vector containing $q^f(\vb{x})$ values for $f\in\mathcal{F}$, and $\vb{q}=[q_1,\dots,q_F]$. 
Above, each upper sum limit
\begin{align}
    G_k(h,q_1,\dots,q_{k-1}) \coloneqq \min\left(h-\sum\limits_{i=1}^{k-1}   q_i,~\min\limits_{n\in f_k}\left(c_{n}-\sum\limits_{i=1}^{k-1}q_i\mathds{1}\left\{n\in f_i\right\}\right)\right)
\end{align}
denotes the maximum number of sessions that can be active for flow $f_k$ given that flows $f_1,\dots,f_{k-1}$ have $q_1,\dots,q_{k-1}$ active sessions and that the total number of sessions should add up to $h$; and $\mathds{1}$ is the indicator function. Each lower sum limit of (\ref{eq:prob_h})
\begin{align}
    H_k(h,q_1,\dots,q_{k-1}) \coloneqq \max\left(0,h-\sum\limits_{i=1}^{k-1}q_i-\sum\limits_{j=k+1}^F \min\limits_{n\in f_j} \left(c_n-\sum\limits_{l=1}^{k-1}q_l\mathds{1}\{n\in f_l\}\right)\right)
\end{align}
denotes the minimum number of sessions that flow $f_k$ must contribute in order for it to be possible that the total number of sessions adds up to $h$.
Throughout, we adopt the convention that $\sum\limits_{i=a}^b x_i=0$ when $b<a$.
\fi

%The last sum of (\ref{eq:prob_h}) can be expanded as follows:
The probability that flows occupy EGS resources according to $\vb{q}$ is
\begin{align}
\mathrm{P}(\vb{q}) = 
    \sum\limits_{\vb{x}:\vb{q}(\vb{x})=\vb{q}}\pi(\vb{x})=
    %\sum\limits_{j=1}^F
    \sum\limits_{\vb{x}^{f_1}:q^{f_1}(\vb{x}^{f_1})=q_1}\cdots
    \sum\limits_{\vb{x}^{f_F}:q^{f_F}(\vb{x}^{f_F})=q_F}D\prod\limits_{f\in\mathcal{F}}\prod\limits_{i=1}^{L^f}\frac{(\rho^f_i)^{x^f_i}}{x^f_i!},
    \label{eq:inner_sum}
\end{align}
where $\vb{q}(\vb{x})$ is a vector containing $q^f(\vb{x}^f)$ values for $f\in\mathcal{F}$. 
%\gv{Call the thing above $\pi(\vb{q})$, then write derivations above in terms of this -- should be much simpler.}
Recursive application of the multinomial theorem on (\ref{eq:inner_sum}) results in
\begin{align}
    %\sum\limits_{\vb{x}:\vb{q}(\vb{x})=\vb{q}}\pi(\vb{x}) = 
    P(\vb{q})=
    D\prod\limits_{j=1}^F\frac{1}{q_j!}\left(\sum\limits_{i=1}^{L^{f_j}}\rho_i^{f_j}\right)^{q_j}
    \equiv D\prod\limits_{j=1}^F\frac{\left(\rho^{f_j}\right)^{q_j}}{q_j!},
    \label{eq:inner_sum_eval}
\end{align}
where $\rho^{f_j}\coloneqq \sum\limits_{i=1}^{L^{f_j}}\rho_i^{f_j}$.
%is the overall traffic intensity of flow $f_j$.
The constant $D$ can be simplified similarly, after the additional step of rewriting the sum over $\vb{y}\in\mathcal{S}$ in terms of the total number of occupied EGS resources. Although $D$ cancels out within the blocking probability, we present it in simplified form ($\tilde{D}$) for completeness:
\begin{align}
    \tilde{D}^{-1} = \sum\limits_{\vb{y}\in\mathcal{S}}\prod\limits_{f\in\mathcal{F}}\prod\limits_{i=1}^{L^f}\frac{(\rho^f_i)^{y^f_i}}{y^f_i!} = \sum\limits_{h=0}^{C}\sum\limits_{\vb{q}\in\mathcal{Q}(h)}\prod\limits_{j=1}^F\frac{(\rho^{f_j})^{q_j}}{q_j!}.
\end{align}
\if{false}
\begin{align}
    \tilde{D}^{-1} = \sum\limits_{\vb{y}\in\mathcal{S}}\prod\limits_{f\in\mathcal{F}}\prod\limits_{i=1}^{L^f}\frac{(\rho^f_i)^{y^f_i}}{y^f_i!} = \sum\limits_{h=0}^{C}\sum\limits_{q_1=0}^{G_1(h)}\cdots \sum\limits_{q_F}^{G_F(h,f_1,\dots,f_F)}\prod\limits_{j=1}^F\frac{(\rho^{f_j})^{q_j}}{q_j!} \sg{\mathds{1}\left\{\sum_{i=1}^{F}  q_i = h \right\}}.
\end{align}
\fi

Using (\ref{eq:inner_sum_eval}) and (\ref{eq:prob_h}), we can obtain the blocking probabilities for the system -- that is, the probability that an arriving request belonging to flow $f_i$ sees $C$ resources occupied is given by
\begin{align}
    \overline{\pi}_i(C) = 
    %\frac{\tilde{D}}{\sum\limits_{\vb{q}\in\mathcal{Q}^{\prime}(i)}\sum\limits_{\vb{x}:\vb{q}(\vb{x})=\vb{q}}\pi(\vb{x})}\sum\limits_{\vb{q}\in\mathcal{Q}(C)\bigcap \mathcal{Q}^{\prime}(i)}\prod\limits_{j=1}^F\frac{\left(\rho^{f_j}\right)^{q_j}}{q_j!},
    \left(\sum\limits_{\vb{q}\in\mathcal{Q}^{\prime}(i)}\prod\limits_{j=1}^F\frac{\left(\rho^{f_j}\right)^{q_j}}{q_j!}\right)^{-1}\sum\limits_{\vb{q}\in\mathcal{Q}(C)\bigcap \mathcal{Q}^{\prime}(i)}\prod\limits_{j=1}^F\frac{\left(\rho^{f_j}\right)^{q_j}}{q_j!},
    \label{eq:prob_C}
\end{align}
\if{false}
\begin{align}
    \overline{\pi}(C) = \tilde{D}\sum\limits_{q_1=0}^{G(C)}\quad\sum\limits_{q_2=0}^{G_2(C,q_1)}\cdots\sum\limits_{q_F=0}^{G_F(C,q_1,\dots,q_{F-1})}\prod\limits_{j=1}^F\frac{\left(\rho^{f_j}\right)^{q_j}}{q_j!}\sg{\mathds{1}\left\{\sum_{i=1}^{F}  q_i = C \right\}},
    \label{eq:prob_C}
\end{align}
\fi
which is a quantity that can be computed numerically.
\end{proof}

\begin{remark}
Let us take a closer look at the overall traffic intensity of a flow, $\rho^f$, for a given $f\in \mathcal{F}$:
\begin{align}
    \rho^f = \sum\limits_{i=1}^{L^f}\rho_i^f = \sum\limits_{i=1}^{L^f}\frac{\lambda_i^f}{\mu_i^f} = \sum\limits_{i=1}^{L^f}\frac{\nu_1^f\tilde{p}_i^f}{\mu_i^f} = \nu_1^f\sum\limits_{i=1}^{L^f}\frac{\tilde{p}_i^f}{\mu_i^f}.
\end{align}
Noting that the sum represents the mean duration of a type $f$ session, we see that $\rho^f$ is the overall traffic intensity of flow $f$, i.e., it is the product of mean arrival rate and mean session duration.
\end{remark}

\begin{remark}
For our model the average blocking probability for a given flow $f$ is expressed as a conditional probability, as entanglement requests are triggered only when all the associated nodes have sufficient communication qubits. In \cite{bonald2006erlang}, conditioning is not necessary since flows always trigger requests according to an unrestricted Poisson process. Further, in our model, blocking probabilities can be different for different flows whereas all the flows have the same blocking probabilities in \cite{bonald2006erlang}.
\end{remark}

Using Theorem~\ref{thm:blocking_flow} we can derive an expression for the average blocking probability for an incoming request by taking an expectation over flow-type of the incoming request. For the following, let $\overline{\pi}_{f_i}(C) \equiv \overline{\pi}_{i}(C)$, where flow $f_i\in\mathcal{F}$ corresponds to the $i$th flow label in $\{1,\dots,F\}$.
\begin{corollary}
The average blocking probability of an incoming entanglement request, denoted by $\overline{\pi}(C)$, is given as
\begin{equation}
\overline{\pi}(C)=\sum_{f\in\mathcal{F}}\left\{\frac{\mathrm{P}(\mathcal{Q}'(f))\nu_1^f}{\sum_{g\in\mathcal{F}}\mathrm{P}(\mathcal{Q}'(g))\nu_1^g}\right\}\overline{\pi}_f(C).
\label{eq:pi_bar_avg}
\end{equation}
\end{corollary}
\begin{proof}
For a flow $f$, entanglement requests are generated according to a Poisson process with rate $\nu_1^f$ only when all the associated nodes have communication qubits which happens with probability $\mathrm{P}(\mathcal{Q}'(f))$. If an entanglement request is triggered then the probability that it is of type~$f$ is proportional to $\mathrm{P}(\mathcal{Q}'(f))\nu_1^f$. From Theorem~\ref{thm:blocking_flow} we have the expression for the blocking probability for a type~$f$ request as $\overline{\pi}_f(C)$. Therefore the average blocking probability is given by (\ref{eq:pi_bar_avg}).
\if{false}
\begin{equation}
\overline{\pi}(C)=\sum_{f\in\mathcal{F}}\left\{\frac{\mathrm{P}(Q'(f))\nu_1^f}{\sum_{g\in\mathcal{F}}\mathrm{P}(Q'(g))\nu_1^g}\right\}\overline{\pi}_f(C).
\end{equation}
\fi
\end{proof}

%Next we show insensitivity of the blocking probability to service time distributions.
Finally, it remains to prove the insensitivity of each flow's blocking probability to the  traffic characteristics of the system beyond the flow-level traffic intensities.
%\subsubsection{Proof of Insensitivity} 
\begin{theorem}\label{thm:insensitivity}
    The blocking probabilities $\overline{\pi}_i(C)$, $i\in\{1,\dots,F\}$ for the system with single EPR pair generation and strict resource reservation depend only on the mean traffic intensities at the flow level, $\rho^{f_i}$, and are not sensitive to the underlying distributions of attempt and calibration period duration.
\end{theorem}
To prove the insensitivity theorem we show that the stationary distribution of the Markov chain that describes the system behaviour remains the same when attempt and calibration durations have Coxian or exponential distributions, as long as the means of the distributions match. This implies that the result also holds for a general distribution. A detailed proof of this theorem is presented in Appendix~\ref{app:insensitivity}.

\subsection{Multiple EPR Pair Generation with Strict Resource Reservation Service Model}
\label{sec:mult_epr_gen_strict}
\begin{figure}[t]
    \centering
    \begin{minipage}[c]{0.6\textwidth}
    \includegraphics[width=\textwidth,trim={1cm 4cm 1.2cm 5.5cm},clip]{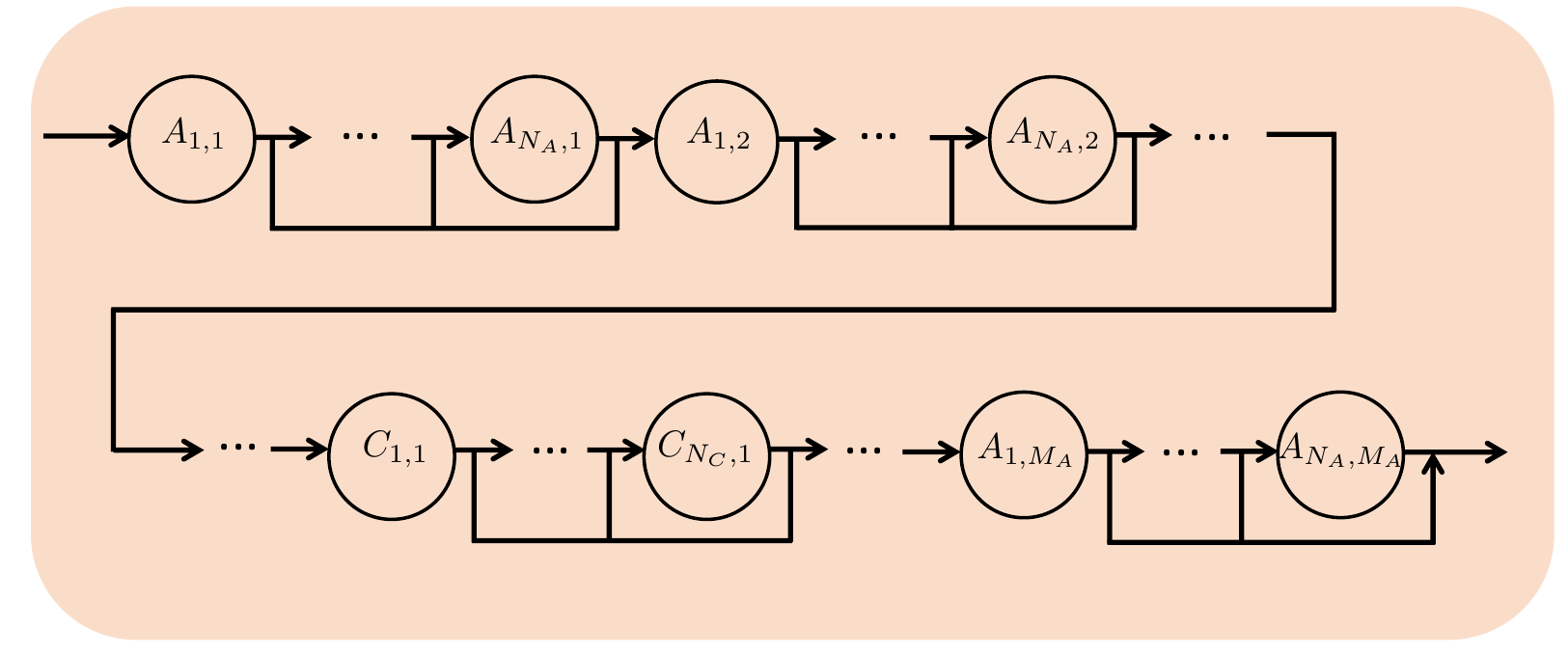}
    \end{minipage}\hfill
    \begin{minipage}[c]{0.37\textwidth}
    \caption{A session from the multiple EPR pair generation strict resource reservation service model, shown at the level of periods in Figure~\ref{fig:strict_res_periods}, decomposed into exponentially-distributed phases so as to result in Cox-distributed attempt and calibration periods.}
    \label{fig:strict_res_mult_epr}
    \end{minipage}
\end{figure}
When the generation of multiple EPR pairs is permitted in strict resource reservation mode, each session, if admitted for service, traverses all periods (albeit, not necessarily all phases). Consequently, as shown in Figure~\ref{fig:strict_res_mult_epr}, transitions to outside of the queueing network are only permitted from the final attempt period of the session. The overall system is thus very similar to that of Section~\ref{sec:single_epr_strict}, and all previous assumptions hold, with the exception that $p_i^f=0$, $\forall f\in\mathcal{F}$, whenever $i$ belongs to a phase other than that of the last period. This modification merely affects the ``overall traffic intensity'' $\rho$, but does not change the form of the stationary distribution or the blocking probability.

\subsection{Multiple EPR Pair Generation with Resource Relinquishment}
\label{subsec:jump_over}
\begin{figure}
    \centering
    \includegraphics[width=\textwidth]{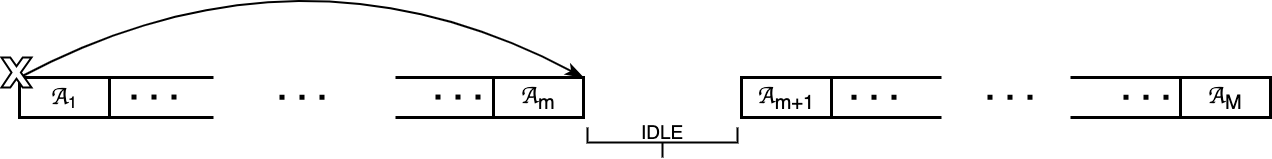}
    \caption{Service model with resource relinquishment  and jump-over blocking. A session consists of multiple ``active'' (periods which make use of an EGS resource module), denoted by {\chancery A}$_i$, $i=1,\dots,M$, interspersed with idle periods. In jump-over blocking, a blocked session goes to the beginning of the next idle period, if there is one, or terminates in case no idle periods remain in the session.}
    \vspace{-5mm}
    \label{fig:jump_over_periods}
\end{figure}
\begin{figure}
    \centering
    \begin{minipage}[c]{0.64\textwidth}
    \includegraphics[width=\textwidth,trim={0 8.5cm 0 0},clip]{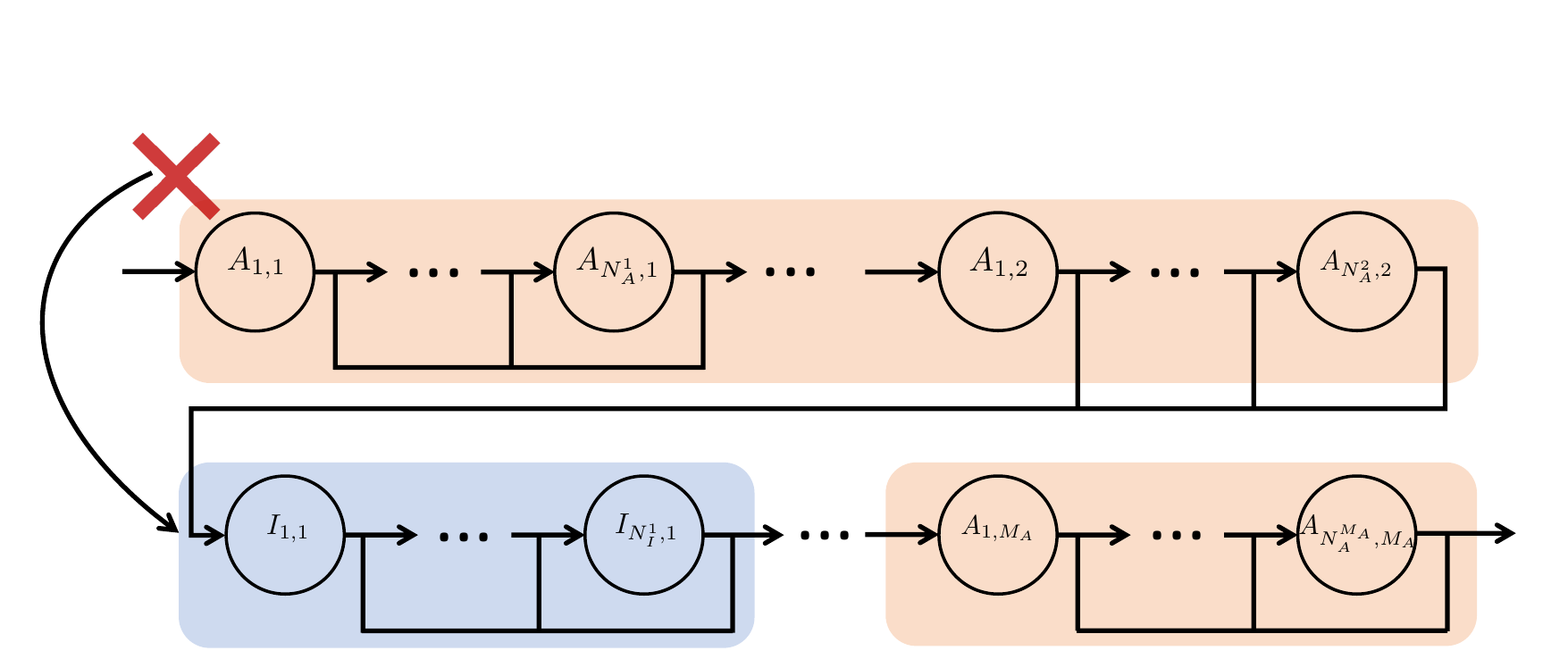}
    \end{minipage}
    \begin{minipage}[c]{0.35\textwidth}
    \caption{A session with jump-over blocking, shown at the level of periods in Figure~\ref{fig:jump_over_periods}, decomposed into exponentially-distributed phases so as to result in Cox-distributed active and idle periods. Transitions to outside the queueing network are permitted only from the last period of the session.}
    \label{fig:jump_over_phases}
    \end{minipage}
\end{figure}
We next study a type of scenario in which flows relinquish resources, such as during calibration, thereby inducing what we refer to as ``idle'' periods. There are several other situations where a flow is amenable to giving up its EGS resource module: for instance, depending on the application/protocol, nodes may wish to perform processing in-between entanglement generation attempts. Further, flows may wish to relinquish the resources during such processing periods, but not during calibration periods, or vice-versa. In this service mode, a retrial model is necessary, since a flow that has relinquished its EGS resource must re-obtain it to continue service after an idle period. We opt for the jump-over blocking mechanism, wherein a flow, when blocked, transitions the session to the beginning of the next idle period, or ends the session if no idle periods remain. Figure~\ref{fig:jump_over_periods} depicts the general form of a session within this service mode: note that periods labeled {\chancery A}$_i$ are not necessarily entanglement generation attempts -- we now refer to them as ``active'' periods, which unlike idle periods \emph{do} engage EGS resources. 

Figure~\ref{fig:jump_over_phases} depicts a session at phase-level detail. For the analysis that follows, we assume that transitions to outside of the queueing network are not permitted from any phases that do not belong to the last period of the session.
The state space of this system is given by the set
\begin{align}
    \mathcal{S}^{\prime\prime} = \left\{\vb{x} \in \mathbb{N}^L: \sum\limits_{f\in\mathcal{F}}\tilde{q}^f(\vb{x}) \leq C,\quad \sum\limits_{f\in\mathcal{F}:n_k\in f}q^f(\vb{x}^f) \leq c_k, \forall k\in \{1,\dots,K\}\right\},
    \label{eq:state_space_jump_over}
\end{align}
where for a given state $\vb{x}$, $\tilde{q}^f(\vb{x})$ is the number of active sessions of flow $f$:
\begin{align}
    \tilde{q}^f(\vb{x}) \equiv \sum\limits_{j=1}^{M_A^f}\sum\limits_{i=1}^{N_A^{f,j}}x_{i,j}^{f,A}.
\end{align}
Here, $M_A^f$ is now defined as the number of active periods for a flow $f$ session, and $N_A^{f,j}$ is the number of phases in the $j$th active period of the session. Since the subscript $A$ now denotes any kind of active period and not only an attempt, we introduce a dependence on the specific period for the number of phases -- this enables us to model generally-distributed period durations. Note that in this service mode, communication qubits are reserved for the entire duration of a session, including during idle periods. This is reflected in the usage of $q^f(\vb{x}^f)$ in the second sum of (\ref{eq:state_space_jump_over}).

The blocking probability for this service mode is of the same form as for the two previously discussed service modes, with the main differences again manifesting through the traffic intensities $\rho_i^f$. Additional consideration is needed within the analysis to account for the fact that blocking may occur not only at the beginning of a session, but also during -- following idle periods. A detailed analysis is presented in Appendix~\ref{app:jump_over}.

%% file: figs/SampleParameters.tex
\begin{table}[t]
    \centering
    \begin{tabular}{c|c}
        \textbf{Description} &  \textbf{Value} \\
        \hline
        Link lengths & 10 km  \\
        One-way communication time (\text{RTT}/2) & 50.03 $\mu$s \\
        % Buffer for jitter in heralding flag arrival ($0.1 \times \text{RTT}/2$) &  5.003   $\mu$s \\
        % Duration of BSM & 10 $\mu$s \\
        Duration of calibration period (CR check), $T_{\text{calib}}$ & 1 ms \\
        Probability of single attempt success, $p_{\text{gen}}$ & 1e-5 (a.u.)\\
        Duration of a single attempt, $T_{\text{attempt}}$ & 115.072 $\mu$s\\ 
        Attempt batch size & 100 \\
        $\#$ batches (strict allocation) or re-trials (jump-over) & 10 \\
        Total calibrations (idle periods) in strict (jump-over) mode & 9 \\
        % Resource allocation duration & 125 ms \\
    \end{tabular}
    \caption{Physical parameters used in simulations correspond to an EGS supporting batched single click HE generation for NV colour centers in diamond as quantum nodes \cite{SingleClickDiamond}.}
    \label{tab:sampleParams}
    \vspace{-5mm}
\end{table}

%% file: insensitivity_appendix.tex
\begin{proof}
To show insensitivity to the distributions of periods we will prove that the stationary distribution for the total number of ongoing sessions of flows remains the same when attempts and calibration periods have either Coxian or exponential distributions with the same mean. 

We first derive the expressions for the stationary distribution for the case where each attempt and calibration period are exponentially distributed. If $M_A^f$ and $M_C^f$ are the number of attempt and calibration periods, respectively, for flow $f$ (and there are no idle periods in this system), then $M^f\equiv M_A^f+M_C^f$ is the total number of periods in each session of type $f$. The state of the system can then be described using the vector
\begin{align}
    \vb{Z} = \left[\vb{Z}^{f_1},\dots,\vb{Z}^{f_F}\right] = \left[Z_i^{f}, 1\leq i\leq M^f, f\in\mathcal{F}\right],
\end{align}
where $Z_i^{f}$ indicates the total number of ongoing sessions of type $f$ in period $i$. Let $L$ be the state dimension, i.e., $L\equiv \sum\limits_{f\in\mathcal{F}}M^f$; then the admissible state space for this system is given by
\begin{align}
    \mathcal{S}^{\prime} = \left\{\vb{Z}\in\mathbb{N}^L:\sum\limits_{f\in\mathcal{F}}\sum\limits_{i=1}^{M^f}Z_i^f\leq C, \sum\limits_{f\in\mathcal{F}:n_k\in f}\sum\limits_{i=1}^{M^f}Z_i^f\leq c_k,\forall k\in\{1,\dots,K\}\right\}.
\end{align}
Let $1/\theta^f_j$ and $1/\sigma^f_j$ be the average duration of  attempt and calibration periods $j$, respectively, for $f\in\mathcal{F}$. Further, let $\omega_j^f$ be the arrival rate into the $j$th period of a session belonging to flow $f$.
By the results of the previous subsection, we know that the stationary distribution for this system is
\begin{align}
    \pi^{\prime}(\vb{Z}) = \left(\sum\limits_{\vb{Y}\in\mathcal{S}^{\prime}}\prod\limits_{f\in\mathcal{F}}\prod\limits_{i=1}^{M^f}\frac{(\eta^f_i)^{Y^f_i}}{Y^f_i!}\right)^{-1}\prod\limits_{f\in\mathcal{F}}\prod\limits_{j=1}^{M^f}\frac{(\eta^f_j)^{Z^f_j}}{Z^f_j!},
    \label{eq:stationary_exp}
\end{align}
where period $j$'s traffic intensity for flow $f$ is $\eta_j^f=\omega_j^f/\theta^f_j$ if $j$ corresponds to an attempt period, and $\eta_j^f=\omega_j^f/\sigma^f_j$ if it corresponds to a calibration period. 
Analogous to the Coxian case, $\omega_j^f=\nu_j^f$ if $j=1$, else it is
\begin{align}
    %\lambda_i^f = \nu_1^f \hat{p}^f_{i},
    \omega_j^f = \nu_1^f \zeta^f_{1,2}\cdots \zeta^f_{j-1,j},
    \label{eq:lambda_i_exp}
\end{align}
%where $\hat{p}^f_{i}$ denotes 
where $\zeta^f_{l-1,l}$ is the probability of transitioning from period $l-1$ to period $l$ of a session belonging to flow $f$. Let $\zeta_j^f$ be the probability of leaving the queueing network after the $j$th queue due to a successful BSM at the EGS for creating an entanglement, and note that in this service mode $\zeta_j^f=0$ if $j$ corresponds to a calibration period.
%I.e., we scale the external arrival rate by
%the probability of reaching the $i$th period of a type $f$ session beginning from its first period.

To prove insensitivity, we assume that the average duration of a period in the exponential scenario is equal to the average duration of the corresponding period in the Coxian scenario. In other words, we have that for the $j$th period, depending on whether it is an attempt or calibration period, respectively,
\begin{align}
    \frac{1}{\theta^f_j} = \sum\limits_{i=1}^{N_A^f}\frac{r_{i,j}^f}{\mu_{i,j}^f}, \quad\text{or}\quad \frac{1}{\sigma^f_j} = \sum\limits_{i=1}^{N_C^f}\frac{r_{i,j}^f}{\mu_{i,j}^f},
    \label{eq:mu_exp_cox}
\end{align}
where $r_{i,j}^f$ denotes the probability of reaching the $i$th phase of the $j$th period starting from its initial phase, and $\mu_{i,j}^f$ denotes the average duration of the $i$th phase of period $j$ within a flow-$f$ session. 
%Similarly, we have for calibration periods that
%\begin{align}
%    \frac{1}{\sigma^f_j} = \sum\limits_{i=1}^{N_C^f}\frac{r_{i,j}^f}{\sigma_{i,j}^f},
%    \label{eq:sigma_exp_cox}
%\end{align}
%where $r_{i,j}^f$ denotes the probability of reaching the $i$th phase of the calibration period from its first phase, and $\sigma_{i,j}^f$ denotes the average duration of the $i$th phase for flow $f$.

We further assume that for each attempt period of the exponential scenario, the entanglement success probability equals that of the corresponding attempt period in the Coxian scenario. That is, $\zeta_j^f=P_j^f$ in this case, where $P_j^f$ is the probability of leaving the queueing network during the $j$th period of a flow-$f$ session in the Coxian scenario, as computed in (\ref{eq:prob_leave_period_cox}). This assumption is physically motivated for scenarios where the mean duration of an attempt is significantly shorter than the timescale of parameter drift affecting a quantum node. For a detailed justification, see the discussion on success probability, Section \ref{subsubsec:pgen}.

In the Coxian scenario, we can rewrite the state representation in (\ref{eq:cox_state}) as follows:
\begin{align}
    \vb{x} = [\vb{x}_{1}^{f_1},\dots,\vb{x}_{M^{f_1}}^{f_1},\dots,\vb{x}_{1}^{f_F},\dots,\vb{x}_{M^{f_F}}^{f_F}],
\end{align}
where $M^{f_i}$ is the number of periods in sessions of type $f_i$.
For any $\vb{x}\in\mathcal{S}$, let $q^f_j(\vb{x}^f_j) = \sum_i x_{i,j}^f$, i.e., this is the number of sessions in the $j$th period of sessions belonging to flow $f$. Then for a given $\vb{Z}\in\mathcal{S}^{\prime}$,
\begin{align}
    \sum\limits_{\vb{x}:q^{f}_m(\vb{x}^f_m)=Z_m^{f},\forall m, f}\pi(\vb{x}) &= D\sum\limits_{\vb{x}^{f_1}_{1}:q_1^{f_1}(\vb{x}^{f_1}_{1})=Z_1^{f_1}}\cdots
\sum\limits_{\vb{x}^{f_F}_{M^{f_F}}:q_{M^{f_F}}^{f_F}(\vb{x}^{f_F}_{M^{f_F}})=Z^{f_F}_{M^{f_F}}}\prod\limits_{f\in\mathcal{F}}\prod\limits_{j=1}^{M^f}\prod\limits_{i=1}^{N^f_j}\frac{\left(\rho_{i,j}^f\right)^{x_{i,j}^f}}{x_{i,j}^f!},
\label{eq:29}
\end{align}
where $D$ is the normalizing constant of the Coxian distribution (see (\ref{eq:cox_stationary_full})), $N^f_j$ is the number of phases in the $j$th period of a session belonging to flow $f$, and $\rho^f_{i,j}=\lambda_{i,j}^f/\mu_{i,j}^f$ is the traffic intensity in the $i$th phase of this period. Multiple applications of the multinomial theorem yield
\begin{align}
    \sum\limits_{\vb{x}:q^{f}_m(\vb{x}^f_m)=Z_m^{f},\forall m, f}\pi(\vb{x}) &=D\prod\limits_{f\in\mathcal{F}}\prod\limits_{j=1}^{M^f}\frac{1}{Z^f_j!}\left(\sum\limits_{i=1}^{N^f_j}\rho_{i,j}^f\right)^{Z_j^f}.
    \label{eq:30}
\end{align}
When $j$ corresponds to an attempt period, $N_j^f = N_A^f$, and we have that
\begin{align}
    \sum\limits_{i=1}^{N^f_j}\rho_{i,j}^f = \sum\limits_{i=1}^{N^f_A}\rho_{i,j}^f = \sum\limits_{i=1}^{N^f_A}\frac{\lambda_{i,j}^f}{\mu_{i,j}^f} &= 
\nu_{1}^f\sum\limits_{i=1}^{N^f_A}\frac{\tilde{p}^f_{i,j}}{\mu_{i,j}^f}
\label{eq:35}\\
&=
%= \frac{\lambda_{j}^f}{\theta_j^f} =
\nu_{1}^f\sum\limits_{i=1}^{N^f_A}\frac{
\prod\limits_{l=1}^{i-1}p^{j,f}_{l,l+1}\prod\limits_{k=1}^{j-1}\left(1-P^f_k\right)}{\mu_{i,j}^f}\label{eq:36}\\
&=\nu_{1}^f\prod\limits_{k=1}^{j-1}\left(1-P^f_k\right)\sum\limits_{i=1}^{N^f_A}\frac{r^f_{i,j}}{\mu_{i,j}^f}\label{eq:37}\\
    &= \frac{\nu_{1}^f}{\theta_j^f}\prod\limits_{k=1}^{j-1}\left(1-\zeta^f_k\right)\label{eq:38}\\
    &= \frac{\nu_{1}^f}{\theta_j^f}\prod\limits_{k=1}^{j-1}\zeta^f_{k,k+1}
    =\eta^f_j,\label{eq:39}
\end{align}
where in the last equality of (\ref{eq:35}) we use (\ref{eq:lambda_i_cox}), in (\ref{eq:36}) we use (\ref{eq:prob_sesh_reaches_ith_cox}); for (\ref{eq:37}) we use the fact that 
%the first product in the numerator 
$\prod\limits_{l=1}^{i-1}p^{j,f}_{l,l+1}$
simply represents the probability of going from the first phase of the period to the $i$th; (\ref{eq:38}) follows from our two assumptions on the entanglement generation success probability and mean period duration; and the last equality of (\ref{eq:39}) follows from (\ref{eq:lambda_i_exp}) and the definition of traffic intensity.
We can use a similar argument to show that when $j$ corresponds to a calibration period, the sum within (\ref{eq:30}) equals $\eta^f_j=\omega^f_j/\sigma^f_j$.

We next address the normalization constant $D$ within (\ref{eq:30}):
\begin{align}
    D = \sum\limits_{\vb{y}\in\mathcal{S}}\prod\limits_{f\in\mathcal{F}}\prod\limits_{i=1}^{L^f}\frac{(\rho^f_i)^{y^f_i}}{y^f_i!} &= \sum\limits_{\vb{Z}\in\mathcal{S}^{\prime}}\sum\limits_{\vb{x}:q^{f^{\prime}}_m(\vb{x}^{f^{\prime}}_m)=Z_m^{f^{\prime}},\forall m, f^{\prime}}~\prod\limits_{f\in\mathcal{F}}\prod\limits_{j=1}^{M^f}\prod\limits_{i=1}^{N_j^f}\frac{(\rho^f_{i,j})^{y^f_{i,j}}}{y^f_{i,j}!},
\end{align}
where we have partitioned the state space $\mathcal{S}$ based on the number of jobs in each period of a session.
Applications of the multinomial theorem as in (\ref{eq:29}), followed by utilization of (\ref{eq:mu_exp_cox}), results in $D$'s equivalence to the normalization constant of the exponential-scenario stationary distribution, see (\ref{eq:stationary_exp}). This result, combined with (\ref{eq:30}), means that
\begin{align}
    \sum\limits_{\vb{x}:q^{f}_m(\vb{x}^f_m)=Z_m^{f},\forall m, f}\pi(\vb{x}) = \pi^{\prime}(\vb{Z}).
\end{align}
The above equation shows that the stationary distribution for the total number of ongoing sessions of flows remains the same when attempts and calibration periods have either Coxian or exponential distributions with the same mean. From this, we conclude that the insensitivity property holds.
\end{proof}

%% file: jump_over_appendix.tex
In this appendix, we derive the stationary distribution $\pi(\vb{x})$, $\vb{x}\in \mathcal{S}^{\prime\prime}$ for the service mode with resource relinquishment and jump-over blocking as the retrial mechanism. In this service mode, a session that had initially been admitted for service by the EGS may get blocked later on, depending on the state of the system $\vb{x}$ at the moment the session leaves an idle period. Thus, to define the traffic characteristics within a session, we require transition probabilities that are functions of the state. Let $p^f_{i,j}(\vb{x})$ be the probability of transitioning from the $i$th phase of a type $f$ session to its $j$th phase. Since in this service mode, a session can only end during its last period, the corresponding model is most similar to the ``multiple EPR pair generation with strict resource reservation'' scenario. Thus, all traffic  characteristics of Section~\ref{sec:mult_epr_gen_strict} apply (i.e., $p^f_{i,j}(\vb{x})= p^f_{i,j}$, $\forall f, i, j, \vb{x}$), with the following exceptions:
\begin{itemize}
    \item[-] If $i$ is the starting phase of an active period (excluding the first period of a session), $j$ is any phase of the preceding idle period, and $\vb{x}+\vb{e}_i\notin \mathcal{S}^{\prime\prime}$, then $p^f_{j,i}(\vb{x})=0$;
    \item[-] If moreover $k$ is the starting phase of the next idle period, then $p^f_{j,k}(\vb{x})=p^f_{j,i}$.
\end{itemize}
These amendments describe the jump-over blocking dynamics.

The external arrival rates into the system are zero for all phases, except for the first of every session; these are given by $\nu_1^f(\vb{x}) = \nu_1^f$ if $\vb{x}+\vb{e}_1^f \in \mathcal{S}^{\prime\prime}$, and zero otherwise.  Letting $\gamma_i^f$ be the probability that a session belonging to flow $f$ reaches its $i$th phase starting from its first phase, we have that the total arrival rate into phase $i$ for flow $f$ while in state $\vb{x}$ is given by
\begin{align}
    \lambda_i^f(\vb{x}) = \begin{cases}
        \nu_1^f\gamma_i^f \equiv \lambda_i^f, & \text{if }\vb{x}+\vb{e}_i^f \in \mathcal{S}^{\prime\prime},\\
        0, & \text{else}.
    \end{cases}
    \label{eq:lambda_jump_over}
\end{align}
Let us examine why $\lambda_i^f$ has no dependence on the state $\vb{x} \in \mathcal{S}^{\prime\prime}$. For the following, suppose $\vb{x}+\vb{e}_i^f \in \mathcal{S}^{\prime\prime}$. First, consider $i$ to be any phase of an active period, and $j$ to be the first phase of the same period. Then $\gamma_i^f = p_{j,j+1}^f\dots p_{i-1,i}^f$. Next, let $i$ be the first phase of an idle period. The arrival rate into this phase is $\nu_1^f$, regardless of whether the transition is happening from the preceding active period, or from the previous idle period. The latter would happen if the system was at capacity (the EGS did not have enough resource modules) at the time of the transition, thereby causing the next active period to be skipped (along with any other active periods that immediately follow it). The new routing rules introduced above ensure that the arrival rate into the idle period is equal to that of the period(s) being jumped over. The arrival rate into any other phase of an idle period is then computed similar to that of an active period's non-initial phase.

To obtain the stationary distribution of this system, we apply the local balance approach as in Section~\ref{sec:single_epr_strict}. The rate of leaving state $\vb{x}\in \mathcal{S}^{\prime\prime}$ due to an outside arrival is
\begin{align}
    A = \pi(\vb{x})\sum\limits_{f\in\mathcal{F}}\nu_1^f(\vb{x}), 
\end{align}
and the rate of entering $\vb{x}$ due to job departure to the outside is
\begin{align}
    A^{\prime} = \sum\limits_{f\in\mathcal{F}}\sum\limits_{i=1}^{L^f}\pi(\vb{x}+\vb{e}^f_i)(x^f_i+1)\mu_i^f p_i^f(\vb{x}+\vb{e}^f_i).
    \label{eq:41}
\end{align}
In (\ref{eq:41}), we can replace $p_i^f(\vb{x}+\vb{e}^f_i)$ with $p^f_i$ since if $\vb{x}+\vb{e}^f_i\notin \mathcal{S}^{\prime\prime}$, then by our convention $\pi(\vb{x}+\vb{e}^f_i)=0$. Recall that if $\vb{x}+\vb{e}^f_i\in \mathcal{S}^{\prime\prime}$ but $i$ corresponds to any phase that does not belong to the last period of a session, then $p^f_i=0$. Then, $A=A^{\prime}$ yields a  traffic conservation equation of the same form as (\ref{eq:12}), if we let
$\pi(\vb{x})$ take the same form as (\ref{eq:pi_guess1}) (note that the definition of $\lambda_i^f(\vb{x})$ is now given by (\ref{eq:lambda_jump_over}) throughout). Local balance for each queue $i$ of a session belonging to flow $f$ yields the expected traffic equations
\begin{align}
    \lambda_i^f(\vb{x}) = \nu_i^f(\vb{x}) + \sum\limits_{j\neq i}\lambda_j^f(\vb{x})p_{j,i}^f(\vb{x}+\vb{e}_j^f-\vb{e}_i^f).
\end{align}

In order to obtain the blocking probability, we must consider not only the external arrival process, but also the internal jump-over blocking mechanism. For the former, we can apply the PASTA property so long as arrivals are a Poisson process. For the latter, we utilize the ``Departures See Time Averages'' corollary from \cite{bonald2006erlang}, which was proven for a the slightly different model presented in that manuscript -- namely, there each active period is followed by an idle period. The corollary nevertheless applies to our modified system, and the proof is identical; we include it below for completeness.

\begin{corollary}
For the system introduced in Section~\ref{subsec:jump_over} and  Appendix~\ref{app:jump_over}, denote a free queue as one that can admit a new job from any already-admissible state; i.e., for flow $f$, $i$ is free if $\vb{x}+\vb{e}_i^f \in \mathcal{S}^{\prime\prime}$ for any $\vb{x}\in \mathcal{S}^{\prime\prime}$.
Jobs leaving any free queue (either moving to another
queue or leaving the network) see the network in steady state immediately after their departure.
\end{corollary}
\begin{proof}
Define $\pi^f_i(\vb{x})$ as the probability that jobs leaving free queue $i$ of a session belonging to flow $f$ see the network in state $\vb{x}$ immediately after their departure.  This probability is given by
\begin{align}
    \pi^f_i(\vb{x}) = \frac{\pi(\vb{x}+\vb{e}_i^f)\mu_i^f(x_i^f+1)}{\sum\limits_{\vb{y}\in\mathcal{S}^{\prime\prime}}\pi(\vb{y}+\vb{e}_i^f)\mu_i^f(y_i^f+1)},
    \label{eq:46}
\end{align}
where the numerator represents the transition rate out of state $\vb{x} + \vb{e}_i^f$ and the denominator represents the transition rate out of all possible admissible states. We remark here that since $i$ is a free queue, $\vb{x}+\vb{e}_i^f\in\mathcal{S}^{\prime\prime}$, $\forall \vb{x}\in\mathcal{S}^{\prime\prime}$ by definition. 

In the non-trivial case that $\vb{x}\in\mathcal{S}^{\prime\prime}$, we know from the stationary distribution analysis that
\begin{align}
    \pi(\vb{x}+\vb{e}_i^f) = \frac{\rho_i^f}{x_i^f+1}\pi(\vb{x}).
\end{align}
Using this result with (\ref{eq:46}) and simplifying yields
\begin{align}
    \pi^f_i(\vb{x}) = \frac{\pi(\vb{x})}{\sum\limits_{\vb{y}\in\mathcal{S}^{\prime\prime}}\pi(\vb{y})}=\pi(\vb{x}),
\end{align}
where the last equality follows from the fact that the denominator evaluates to one since the sum is over all admissible states.
\end{proof}
 
With this corollary, we conclude that for an active period that immediately follows an idle period,
the blocking probability is the stationary probability that upon departure from said idle period, all EGS resource modules are engaged. As mentioned earlier, the same applies to the first period of a session (recall that by our convention, these are active periods). Finally, consider any active period that immediately follows another active period in the session: its blocking probability is equal to that of the first active period in the batch. In other words, if $j$ is the active period under consideration, and $i<j$ is the closest idle period that precedes it (such that there are no other idle periods between $i$ and $j$), then $i+1$ is the first active period in the batch that contains $j$, and $j$ has the same blocking probability as $i+1$. If there are no idle periods preceding $j$, then $j$ has the same blocking probability as the first active period in the session. We thus conclude that the blocking probability for an arbitrary active period has the same form as (\ref{eq:prob_C}), with the only difference being the definition of the traffic intensities $\rho_i^f$.

% \begin{theorem}
% \label{thm:blocking_JumpOver_flow}
%     For the system with multiple EPR pair generation operating in the mode with resource relinquishment, the probability that an arriving request belonging to flow $i\in\{1,\dots,F\}$ is blocked is given by 
%     \begin{align*}
%     \overline{\pi}_i(C) = 
%     \left(\sum\limits_{\vb{q}\in\mathcal{Q}^{\prime}(i)}\prod\limits_{j=1}^F\frac{\left(\rho^{f_j}\right)^{q_j}}{q_j!}\right)^{-1}\sum\limits_{\vb{q}\in\mathcal{Q}(C)\bigcap \mathcal{Q}^{\prime}(i)}\prod\limits_{j=1}^F\frac{\left(\rho^{f_j}\right)^{q_j}}{q_j!}, 
%     \end{align*}
%     where $\vb{q} = [q_1,\dots,q_F]$ represents the number of active sessions $q_i$ from each flow $f_i$,
%     \begin{align*}
%     \quad\mathcal{Q}(h) &\coloneqq \left\{\vb{q}=[q_1,\dots,q_F] \in \mathbb{Z}_+^F: \sum\limits_{i=1}^Fq_i=h,~ \sum\limits_{i: n_k\in f_i}q_i\leq c_k,\forall k\in\{1,\dots,K\}\right\},\text{ and}\label{eq:set_Q}\\
%     \mathcal{Q}^{\prime}(i) &\coloneqq \left\{\vb{q}=[q_1,\dots,q_F] \in \mathbb{Z}_+^F: \sum\limits_{j: n_k\in f_j}q_j\leq c_k,\forall n_k\notin f_i,~ \sum\limits_{j:n_l\in f_j}q_j < c_l, n_l\in f_i\right\}.\label{eq:set_Qprime}
%     \end{align*}
% \end{theorem}

%% file: AppendixCorrelatedInformation.tex
\section{Extended Information on Physical EGS Operation Protocols}
\label{appendix:OpAppendix}
\subsection{CI Generation Protocol}
The goal of a node pair $(n_i, n_j)$ running a CI generation protocol is to generate a data record of a qubit string at each of nodes $n_i$ and $n_j$ such that the data records of $n_i$ and $n_j$ have entanglement-like correlations. These data records can serve as raw key, and the CI protocol is a form of QKD \cite{bb84,e91} that uses a BSA as a measurement device at a central midpoint, known as MDI-QKD \cite{originalMDIQKD1, originalMDIQKD2, MDIQKDDeployed}. The advantage of MDI-QKD over traditional QKD protocols is that it is more secure, as it removes the possibility of detector side-channel attacks. Detector side-channel attacks have been shown to be easy to implement and challenging to defend against \cite{sideChannels1,sideChannels2,sideChannels3}, hence this is an important advantage.  To support running the CI generation protocol, nodes $n_i$ and $n_j$ require a photon source. This source can be a multi-photon emitter attenuated to the single-photon level, such as an attenuated laser pulse, or a true single photon source like a quantum communication qubit. The CI generation protocol's compatibility with multi-photon sources is advantageous. In the early stages of quantum network development, multi-photon sources are more cost-effective and easier to build/purchase and maintain compared to single photon sources \cite{MDIQKDDeployed}.

In broad terms, the protocol consists of four main stages. First, nodes exchange calibration information to ensure that the photons emitted by their sources will be spectrally indistinguishable. Second, each node encodes a qubit state into an emission of the photon source. Third, the photons are sent to a BSA, at which a BSM is attempted between the two received photons, which if successful, projects the photons into a maximally entangled state. Fourth, if the BSM succeeds, the records of the qubits prepared by ($n_i$, $n_j$) share entanglement-like correlations and a success flag is communicated to the nodes. In case of success the nodes store the qubit record. If the measurement is not successful, a failure flag is passively communicated to the nodes as the absence of sending a success flag; after the RTT time between each node and the EGS passes, plus some buffer time accounting for measurement and jitter, nodes discard qubit records from failed attempts. It is not necessary to wait for the outcome flag before commencing subsequent attempts. However, if the nodes use communication qubits to generate the single photons, the communication qubits must be measured in between subsequent attempts in order to obtain the data record. If a simple photon source is used this measurement step is not necessary, as the data record corresponds to the preparation record of the photon. Sequential execution of the second and third stages constitute a single CI generation attempt and may be repeated in batches. For any attempt, the fourth stage must eventually occur -- the outcome flag must eventually be received by the nodes. Otherwise, a stop of the protocol may be triggered. 

A communication sequence between the node pair ($n_i, n_j$) allocated use of an EGS resource to perform CI generation is illustrated in Figure \ref{fig:CorrelatedInfoCommSeq}.
\input{figs/CorrelatedInfoCommSeq}

\subsection{Comparison of the HE and CI protocols}

\input{figs/HeraldedCommDiagram}

The main difference between these two protocols is that in the HE protocol a successful BSM at the EGS means that nodes become entangled whereas in the CI protocol a successful BSM at the EGS means that the data records corresponding to the photons sent to the EGS will display entanglement like correlations. Other important differences are the repetition rate and the probability that single attempts succeed. See Figures \ref{fig:CorrelatedInfoCommSeq} and \ref{fig:HeraldedCommSeq} for visualization of how the communication sequences differ between the protocols. The frequency at which CI generation can be attempted is inherently faster than HE generation because subsequent attempts do not require waiting for the outcome flag to be received (step four in the HE and CI communication sequences). The rate at which CI generation can be attempted with simple photon sources is limited by the emission rate of the photon source, the channel capacity of the fiber over which photons are sent,  or the rate at which the BSA detectors become responsive again following detection of a photon (recovery from dead-time), whichever of these factors is most restrictive. If quantum communication qubits are used to enact the CI generation protocol, the other potential rate limiting factor is how quickly the qubits can be measured following photon emission. To summarize, subsequent attempts of the CI generation protocol can begin before the completion of previous attempts, and the rate limiting factors in an implementation may allow for significantly high attempt repetition rates (MHz or GHz). In contrast, the HE generation protocol requires nodes to receive the heralding flag of any previous attempt before commencing any subsequent attempt, i.e., attempts must be non-overlapping. In a local quantum network with optical fiber links between nodes and the EGS ranging from 5 to 50 km, the RTTs vary from approximately 25 $\mu$s to 250 $\mu$s. Such RTTs dominate the time duration of a single HE generation attempt, limiting the maximum repetition rate of attempts in the protocol to the level of kHz.

A real world implementation of CI generation with simple photon sources \cite{MDIQKDDeployed} displayed a high total probability of a single attempt succeeding. In this implementation the main limitations against an attempt succeeding were loss of a photon travelling in fiber or a detection pattern incompatible with projection into a maximally entangled state. In contrast, in real world implementations of single slick bipartite heralded entanglement generation the probability of an attempt succeeding is limited by the probability of single photon emission from the communication qubit, and further suppressed by the probability of losing a photon travelling in fiber and the success probability of a BSM. 

%% file: figs/CorrelatedInfoCommSeq.tex
\begin{figure}
\center
\begin{tikzpicture}[>=latex]
\coordinate (A) at (1,7);
\coordinate (B) at (1,0);
\coordinate (C) at (6,7);
\coordinate (D) at (6,0);
\coordinate (E) at (11,7);
\coordinate (F) at (11,0);
\draw[thick] (A)--(B) (C)--(D) (E)--(F);
\draw (A) node[above]{\Large $n_i$, Initiator};
\draw (C) node[above]{\Large EGS};
\draw (E) node[above]{\Large $n_j$, Secondary Node};

\coordinate (G) at ($(A)!.02!(B)$);
\draw (G) node[left]{\textit{submit request $\phantom{x}$}};

\coordinate (H) at ($(C)!.11!(D)$);
\draw (H) node[right]{\begin{tabular}{l}
\verb$Block Request$\\
\textit{no resource available}
\end{tabular}};
\draw[-Straight Barb] (G) -- (H) node[midway,sloped,above]{req($i,j, P_{ij}$)};

\coordinate (I) at ($(A)!.16!(B)$);
\draw (I) node[left]{\textit{submit request $\phantom{x}$ }};
\draw[|->|] (G) -- (I) node[midway,sloped,below]{$\Delta \tau_1$};

\coordinate (J) at ($(C)!.25!(D)$);
\coordinate (J') at ($(C)!.26!(D)$);
\draw (J) node[right]{$\phantom{x} $\verb$Accept Request$};
\draw[-Straight Barb] (I) -- (J) node[midway,sloped,above]{req($i,j, P_{ij}$)};

\coordinate (K) at ($(A)!0.35!(B)$);
\draw[-Straight Barb] (J') -- (K) node[midway,sloped,below]{ack($i, j, t_{\text{start}}, P_{ij}$)};

\coordinate (L) at ($(E)!0.35!(F)$);
\draw[-Straight Barb] (J') -- (L) node[midway,sloped,below]{ack($i, j, t_{\text{start}}, P_{ij}$)};

\coordinate (M) at ($(A)!0.38!(B)$);
\draw (M) node[left]{$t_{\text{start}}$};
\draw[-|] (K) -- (M) node[midway,sloped,above]{};

\coordinate (N) at ($(E)!0.38!(F)$);
\draw (N) node[right]{$t_{\text{start}}$};
\draw[-|] (L) -- (N) node[midway,sloped,above]{};

\coordinate (O) at ($(C)!0.47!(D)$);
\coordinate (O') at ($(C)!0.48!(D)$);
\draw[-Straight Barb, teal, line width=0.4mm] (N) -- (O) node[midway,sloped,above]{};
\draw[-Straight Barb, teal, line width=0.4mm] (M) -- (O) node[midway,sloped,above]{};

\coordinate (P) at ($(E)!0.57!(F)$);
\coordinate (Q) at ($(A)!0.57!(B)$);
\draw[-Straight Barb, blue, dotted, line width=0.3mm] (O') -- (P) node[midway,sloped,above]{\verb$Success   $};
\draw[-Straight Barb, blue, dotted, line width=0.3mm] (O') -- (Q) node[midway,sloped,above]{\verb$   Success$};

\coordinate (P') at ($(E)!0.58!(F)$);
\coordinate (Q') at ($(A)!0.58!(B)$);
\coordinate (R) at ($(C)!0.67!(D)$);
\draw[-Straight Barb, teal, line width=0.4mm] (P') -- (R) node[midway,sloped,above]{};
\draw[-Straight Barb, teal, line width=0.4mm] (Q') -- (R) node[midway,sloped,above]{};

\coordinate (R') at ($(C)!0.68!(D)$);
\coordinate (S) at ($(E)!0.77!(F)$);
\coordinate (T) at ($(A)!0.77!(B)$);
\draw[-Straight Barb, red, dashed, line width=0.3mm] (R') -- (S) node[midway,sloped,above]{\verb$Failure   $};
\draw[-Straight Barb, red, dashed, line width=0.3mm] (R') -- (T) node[midway,sloped,above]{\verb$   Failure$};

\coordinate (S') at ($(E)!0.78!(F)$);
\coordinate (T') at ($(A)!0.78!(B)$);
\coordinate (U) at ($(C)!0.87!(D)$);
\draw[-Straight Barb, teal, line width=0.4mm] (S') -- (U) node[midway,sloped,above]{};
\draw[-Straight Barb, teal, line width=0.4mm] (T') -- (U) node[midway,sloped,above]{};

\coordinate (U') at ($(C)!0.88!(D)$);
\coordinate (V) at ($(E)!0.97!(F)$);
\coordinate (W) at ($(A)!0.97!(B)$);
\draw[-Straight Barb, blue, dotted, line width=0.3mm] (U') -- (V) node[midway,sloped,above]{\verb$Success   $};
\draw[-Straight Barb, blue, dotted, line width=0.3mm] (U') -- (W) node[midway,sloped,above]{\verb$   Success$};

\coordinate (EL1) at ($(A)!0.48!(B)$);
\coordinate (ER1) at ($(E)!0.48!(F)$);
\coordinate (EM1) at ($(C)!0.57!(D)$);
\draw[-Straight Barb, teal, line width=0.4mm] (EL1) -- (EM1) node[midway,sloped,above]{};
\draw[-Straight Barb, teal, line width=0.4mm] (ER1) -- (EM1) node[midway,sloped,above]{};

\coordinate (EL2) at ($(A)!0.68!(B)$);
\coordinate (ER2) at ($(E)!0.68!(F)$);
\coordinate (EM2) at ($(C)!0.77!(D)$);
\draw[-Straight Barb, teal, line width=0.4mm] (EL2) -- (EM2) node[midway,sloped,above]{};
\draw[-Straight Barb, teal, line width=0.4mm] (ER2) -- (EM2) node[midway,sloped,above]{};

\coordinate (EM1') at ($(C)!0.58!(D)$);
\coordinate (EL1') at ($(A)!0.67!(B)$);
\coordinate (ER1') at ($(E)!0.67!(F)$);
\draw[-Straight Barb, blue, dotted, line width=0.3mm] (EM1') -- (EL1') node[midway,sloped,above]{\verb$   Success$};
\draw[-Straight Barb, blue, dotted, line width=0.3mm] (EM1') -- (ER1') node[midway,sloped,above]{\verb$Success   $};

\coordinate (EM2') at ($(C)!0.78!(D)$);
\coordinate (EL2') at ($(A)!0.87!(B)$);
\coordinate (ER2') at ($(E)!0.87!(F)$);
\draw[-Straight Barb, blue, dotted, line width=0.3mm] (EM2') -- (EL2') node[midway,sloped,above]{\verb$   Success$};
\draw[-Straight Barb, blue, dotted, line width=0.3mm] (EM2') -- (ER2') node[midway,sloped,above]{\verb$Success   $};

\end{tikzpicture}
\caption{Communication sequence for the CI generation protocol. The field $P_{ij}$ in a request is a package of request parameters, such as the number of attempts requested. $\Delta \tau_1$ is an exponentially distributed inter-arrival time between subsequent requests. Black arrows indicate classical messages. Teal arrows indicate pulses of light at the single photon level sent over optical fiber. Blue, dotted (red, dashed) arrows indicate a success (failure) flag communicated over optical fiber. Commencement of subsequent attempts does not depend on receiving the outcome flag of the previous attempt, thus the communication time between the EGS and nodes does not limit the attempt rate. The attempt rate is limited by the minimum of the rate at which the photon sources at the nodes can output pulses, the channel capacity of the fiber over which photons are sent, the rate at which the detectors can respond to photons (i.e., accounting for detector dead-time), and possibly the rate at which the photon sources may be measured following photon emission.}
\label{fig:CorrelatedInfoCommSeq}
\end{figure}
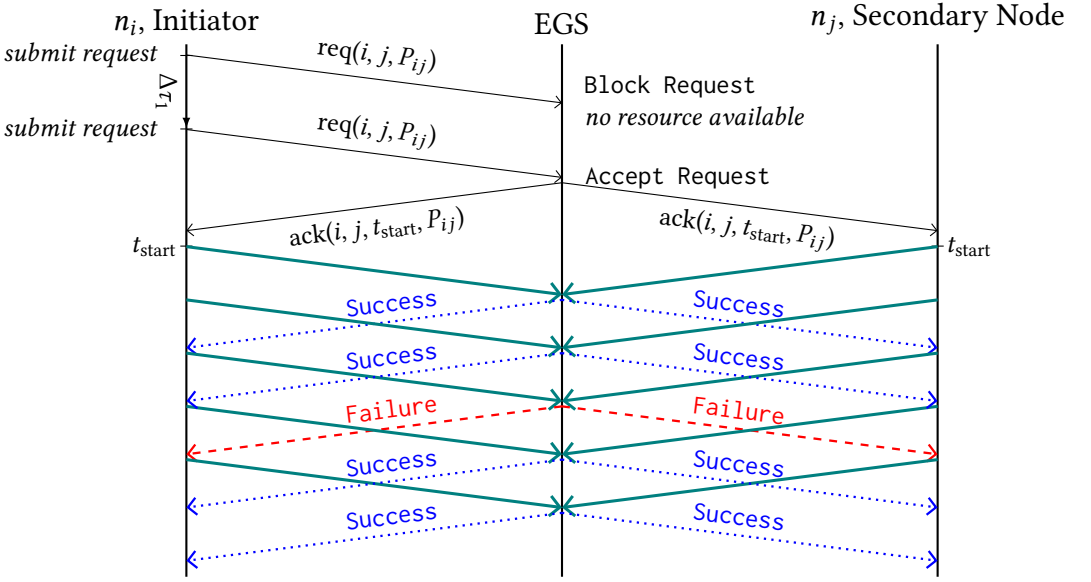

%% file: figs/HeraldedCommDiagram.tex
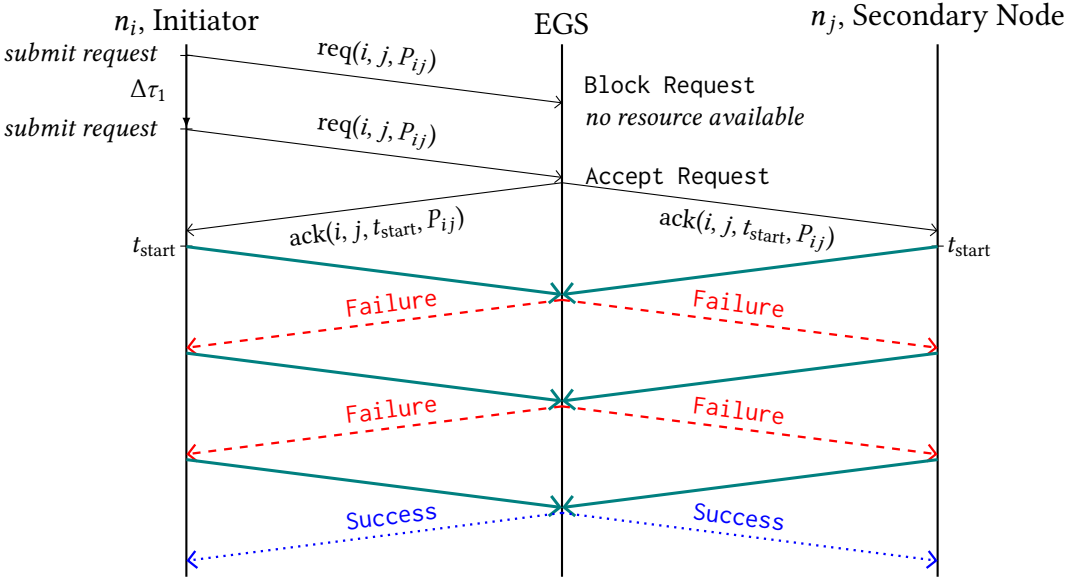
\begin{figure}
\center
\begin{tikzpicture}[>=latex]
\coordinate (A) at (1,7);
\coordinate (B) at (1,0);
\coordinate (C) at (6,7);
\coordinate (D) at (6,0);
\coordinate (E) at (11,7);
\coordinate (F) at (11,0);
\draw[thick] (A)--(B) (C)--(D) (E)--(F);
\draw (A) node[above]{\Large $n_i$, Initiator};
\draw (C) node[above]{\Large EGS};
\draw (E) node[above]{\Large $n_j$, Secondary Node};

\coordinate (G) at ($(A)!.02!(B)$);
\draw (G) node[left]{\textit{submit request $\phantom{x}$}};

\coordinate (H) at ($(C)!.11!(D)$);
\draw (H) node[right]{\begin{tabular}{l}
\verb$Block Request$\\
\textit{no resource available}
\end{tabular}};
\draw[-Straight Barb] (G) -- (H) node[midway,sloped,above]{req($i,j, P_{ij}$)};

\coordinate (I) at ($(A)!.16!(B)$);
\draw (I) node[left]{\textit{submit request $\phantom{x}$ }};
\node (deltaTau1) at (0.5,6.4) {$\Delta \tau_1$};
\draw[|->|] (G) -- (I) node[midway,sloped,below]{};

\coordinate (J) at ($(C)!.25!(D)$);
\coordinate (J') at ($(C)!.26!(D)$);
\draw (J) node[right]{$\phantom{x} $\verb$Accept Request$};
\draw[-Straight Barb] (I) -- (J) node[midway,sloped,above]{req($i,j, P_{ij}$)};

\coordinate (K) at ($(A)!0.35!(B)$);
\draw[-Straight Barb] (J') -- (K) node[midway,sloped,below]{ack($i, j, t_{\text{start}}, P_{ij}$)};

\coordinate (calL1) at ($(A)!0.38!(B)$);

\coordinate (calR1) at ($(E)!0.35!(F)$);

\coordinate (L) at ($(E)!0.35!(F)$);
\draw[-Straight Barb] (J') -- (L) node[midway,sloped,below]{ack($i, j, t_{\text{start}}, P_{ij}$)};

\coordinate (M) at ($(A)!0.38!(B)$);
\draw (M) node[left]{$t_{\text{start}}$};
\draw[-|] (K) -- (M) node[midway,sloped,above]{};

\coordinate (N) at ($(E)!0.38!(F)$);
\draw (N) node[right]{$t_{\text{start}}$};
\draw[-|] (L) -- (N) node[midway,sloped,above]{};

\coordinate (O) at ($(C)!0.47!(D)$);
\coordinate (O') at ($(C)!0.48!(D)$);
\draw[-Straight Barb, teal, line width=0.4mm] (N) -- (O) node[midway,sloped,above]{};
\draw[-Straight Barb, teal, line width=0.4mm] (M) -- (O) node[midway,sloped,above]{};

\coordinate (P) at ($(E)!0.57!(F)$);
\coordinate (Q) at ($(A)!0.57!(B)$);
\draw[-Straight Barb, red, dashed, line width=0.3mm] (O') -- (P) node[midway,sloped,above]{\verb$Failure   $};
\draw[-Straight Barb, red, dashed, line width=0.3mm] (O') -- (Q) node[midway,sloped,above]{\verb$   Failure$};

\coordinate (P') at ($(E)!0.58!(F)$);
\coordinate (Q') at ($(A)!0.58!(B)$);
\coordinate (R) at ($(C)!0.67!(D)$);
\draw[-Straight Barb, teal, line width=0.4mm] (P') -- (R) node[midway,sloped,above]{};
\draw[-Straight Barb, teal, line width=0.4mm] (Q') -- (R) node[midway,sloped,above]{};

\coordinate (R') at ($(C)!0.68!(D)$);
\coordinate (S) at ($(E)!0.77!(F)$);
\coordinate (T) at ($(A)!0.77!(B)$);
\draw[-Straight Barb, red, dashed, line width=0.3mm] (R') -- (S) node[midway,sloped,above]{\verb$Failure   $};
\draw[-Straight Barb, red, dashed, line width=0.3mm] (R') -- (T) node[midway,sloped,above]{\verb$   Failure$};

\coordinate (S') at ($(E)!0.78!(F)$);
\coordinate (T') at ($(A)!0.78!(B)$);
\coordinate (U) at ($(C)!0.87!(D)$);
\draw[-Straight Barb, teal, line width=0.4mm] (S') -- (U) node[midway,sloped,above]{};
\draw[-Straight Barb, teal, line width=0.4mm] (T') -- (U) node[midway,sloped,above]{};

\coordinate (U') at ($(C)!0.88!(D)$);
\coordinate (V) at ($(E)!0.97!(F)$);
\coordinate (W) at ($(A)!0.97!(B)$);
\draw[-Straight Barb, blue, dotted, line width=0.3mm] (U') -- (V) node[midway,sloped,above]{\verb$Success   $};
\draw[-Straight Barb, blue, dotted, line width=0.3mm] (U') -- (W) node[midway,sloped,above]{\verb$   Success$};

\end{tikzpicture}
\caption{Communication sequence for HE protocol. The field $P_{ij}$ in a request is a package of request parameters, such as the number of attempts requested. $\Delta \tau_1$ is the exponentially distributed inter-arrival time between subsequent requests. Black arrows indicate classical messages. Teal arrows indicate single photons sent over optical fiber. Blue, dotted (red, dashed) arrows indicate a success (failure) flag communicated over optical fiber. Commencement of subsequent entanglement generation attempts requires the heralding signal of the previous attempt to be received, thus the attempt rate is limited by the RTT between the EGS and the nodes. }
\label{fig:HeraldedCommSeq}
\end{figure}

%% file: AppendixSimulationImplementation.tex
\section{Simulation implementation}
\label{app:SimulationImplementation}

We call the models \textit{discrete}, \textit{exponential}, and \textit{Cox}, 
which references that in the discrete model every attempt/calibration/idle period is assigned an identical fixed duration, $T_{\text{attempt}}$/$T_{\text{calib}}$/$T_{\text{idle}}$, whereas in the exponential and Cox models every attempt/calibration/idle period is assigned a duration sampled from the appropriate distribution. These simulations allow a method of comparing the predictions of Theorems {\ref{thm:blocking_flow}} and \ref{thm:insensitivity} to the blocking probabilities we would expect to observe in a real implementation where the distribution governing the duration of events coincides with the simulated model. The discrete simulation mode is implemented as a discrete time simulation where the time step is set to the duration of a single attempt, which is the shortest timescale in the system. 
The Poisson process by which flows submit requests is implemented by scaling the exponentially distributed inter-arrival times between requests to a number of time steps and rounding up to the nearest time step. 
The exponential and Cox simulation modes are implemented as discrete event simulations.
% scaling - dividing by the time scale

The expected number of requests placed over the duration of each simulation depends on the request arrival rate from each flow, $\nu^{f}_1$. Each simulation corresponded to 1150.73 seconds of simulated real-time. The data in Figure \ref{fig:validationStrict} is representative for a discussion of the expected number of requests placed over that time. For the minimum (maximum) request arrival rate simulated, the expected number of requests from each flow is 100 (1135). Since there are 28 flows in total, the total expected number of requests placed during a simulation is then 2800 (31780) for the minimum (maximum) request rate simulated. Note that we in general observe very close agreement between the blocking probabilities resulting from discrete/exponential/Cox simulation types as well as with our numeric evaluation of (\ref{eq:prob_C}). The magnitude of disagreement between simulation types may be referred to as the error between simulation types. To investigate whether these errors change systematically when simulation parameters are altered we have executed simulations with various time-step durations pertaining to the discrete simulation as well as with various numbers of attempts per batch of entanglement generation attempts. Under changes to these parameters, the errors between simulation types remain very small ($<1 \%$) and do not change in any systematic way. 

In the implementation of the multiple EPR pair generation with resource relinquishment (jump-over) service mode, the mean duration of the idle periods, $T_{\text{idle}}$, is set to be equivalent to the mean duration of the calibration periods of the single/multiple EPR pair generation with strict resource reservation service modes, i.e., $T_{\text{idle}} = T_{\text{calib}}$. The motivation for this implementation choice is to enable direct comparison between the different service modes, as in Figure \ref{fig:comparePerfMetrics}.

In the non-homogeneous traffic scenario considered, for every flow, if one or both of the users involved is connected by a longer link, then the attempt time for the flow is lengthened to accommodate the greater RTT between the more distant node and the EGS. Thus flows in $S_2$ require longer resource reservation times to achieve the same number of attempts. As mentioned above, in the discrete simulations of homogeneous traffic, the time step is set to the duration of a single attempt. In the implementation of our non-homogeneous traffic model, we set the time step in discrete simulations to the duration of an attempt in $S_1$, which is the shortest time scale in the system. We then fix the mean attempt time in $S_2$ to double the mean attempt time of flows in $S_1$. This time is longer than the total of the RTT required for flows in $S_2$ plus the buffer included to capture measurement times for flows in $S_1$. Hence we have modelled adding some additional buffer time to each attempt for flows in $S_2$. As a consequence, in discrete simulations the duration of each attempt for flows in $S_2$ is exactly two-time steps. 
We assume all flows submit requests at the same rate, thus isolating the source of non-homogeneity to the different attempt times in sets $S_1$ and $S_2$, which are largely due to the different link-lengths.

\begin{table}
    \centering
    \begin{tabular}{c|c}
         \textbf{Parameter type} & \textbf{Value} \\ \hline
         Number of phases & 4 \\
         Mean duration of phase 1 & $  0.41\overline{6} \cdot T_{\text{attempt}}$ \\
         Mean duration of phase 2 & $ 0.521 \cdot T_{\text{attempt}}$ \\
         Mean duration of phase 3 & $ 0.651 \cdot T_{\text{attempt}} $ \\
         Mean duration of phase 4 & $ 0.814 \cdot T_{\text{attempt}}$ \\
         Transition probabilities & ( 0.6, 0.48,0.384, 0.307)\\
    \end{tabular}
    \caption{Cox distribution parameters for individual entanglement generation attempts in the homogeneous traffic scenario.}
    \label{tab:CoxAttemptParams}
    \vspace{-5mm}
\end{table}

% \subsection{Non-Homogeneous Traffic}
% For every flow $f$, if one or both of the users involved is connected by a longer link ($f \in S2$), then the attempt time for the flow is lengthened to accommodate the greater RTT between the more distant node and the EGS. Thus flows in S2 require longer resource reservation times to achieve the same number of attempts. We suppose all flows submit requests at the same rate, thus isolating the source of non-homogeneity to the different link-lengths. 

%% file: AppendixValidation.tex
\section{Extended Data: Homogenous Traffic in the Single Entanglement Generation with Strict Resource Reservation Service Mode}
\label{appendix:extendedData}

\begin{figure}[H]
    \centering
    \begin{subfigure}[t]{0.32\textwidth}
        \centering
        \includegraphics[height=1.8in]{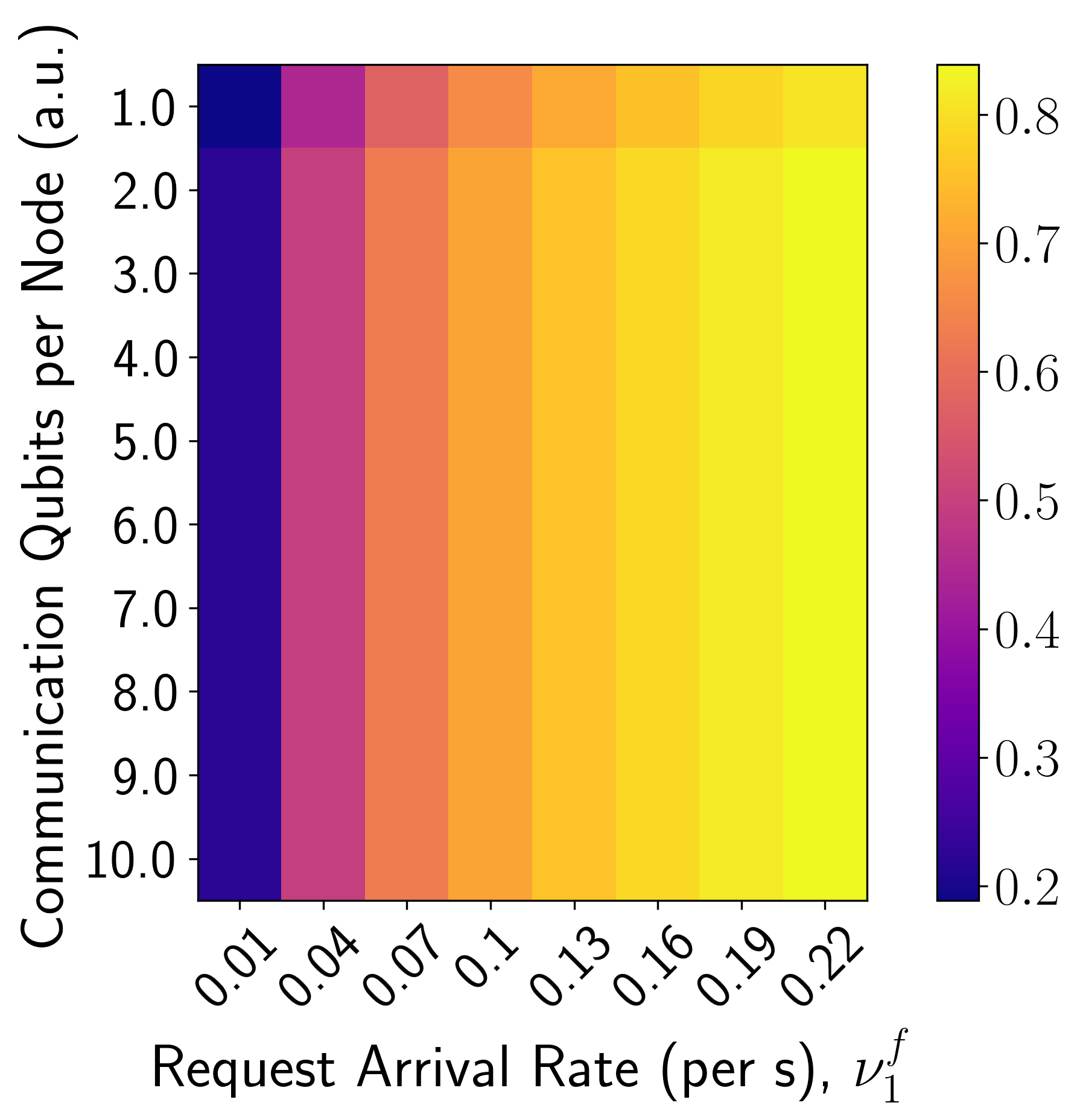}
        \caption{one EGS Resource ($C=1$)}
        \end{subfigure}
        ~
    \begin{subfigure}[t]{0.32\textwidth}
        \centering
        \includegraphics[height=1.8in]{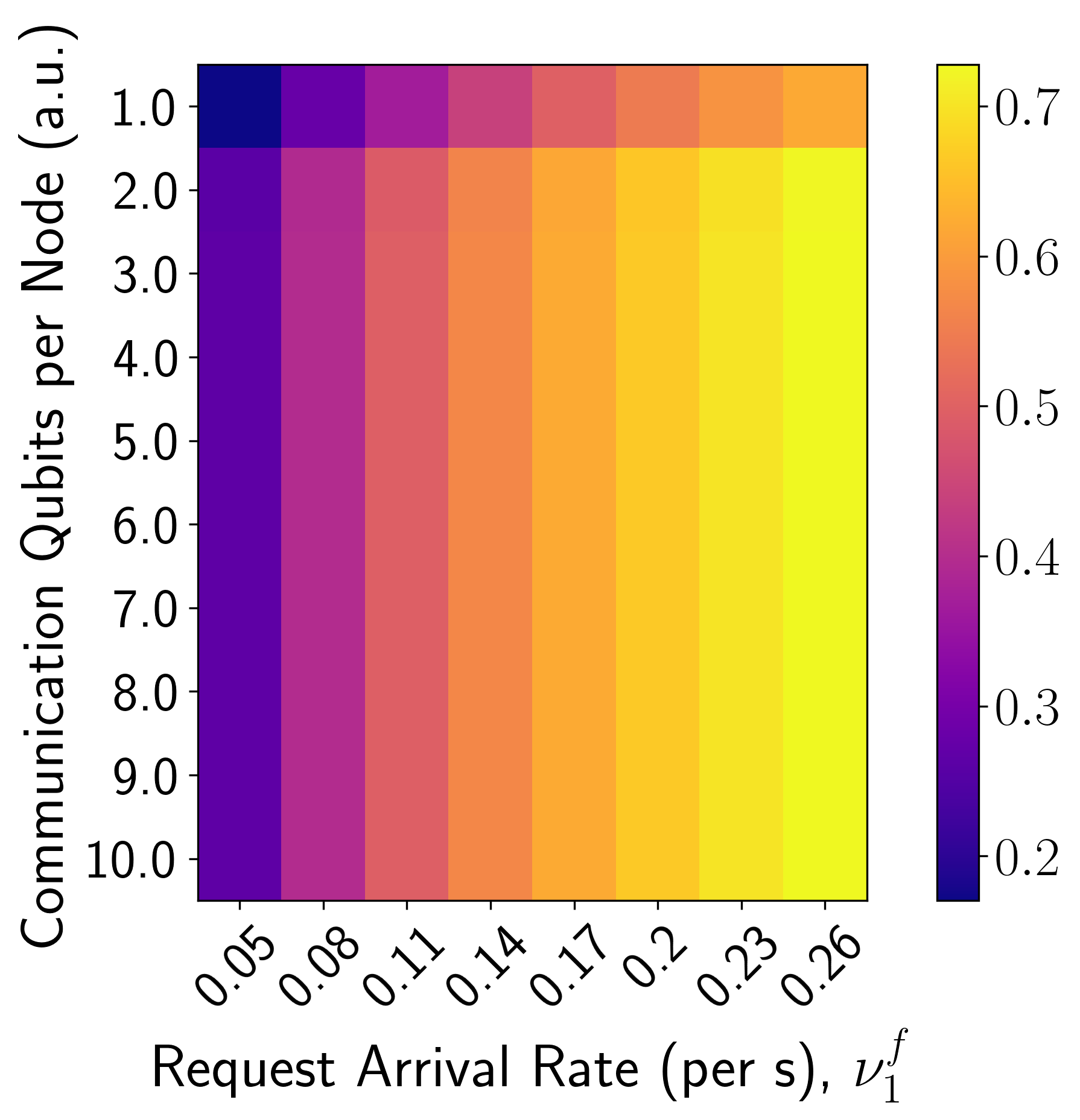}
        \caption{two EGS Resources ($C=2$)}
        \end{subfigure}
        ~
       \begin{subfigure}[t]{0.32\textwidth}
        \centering
        \includegraphics[height=1.8in]{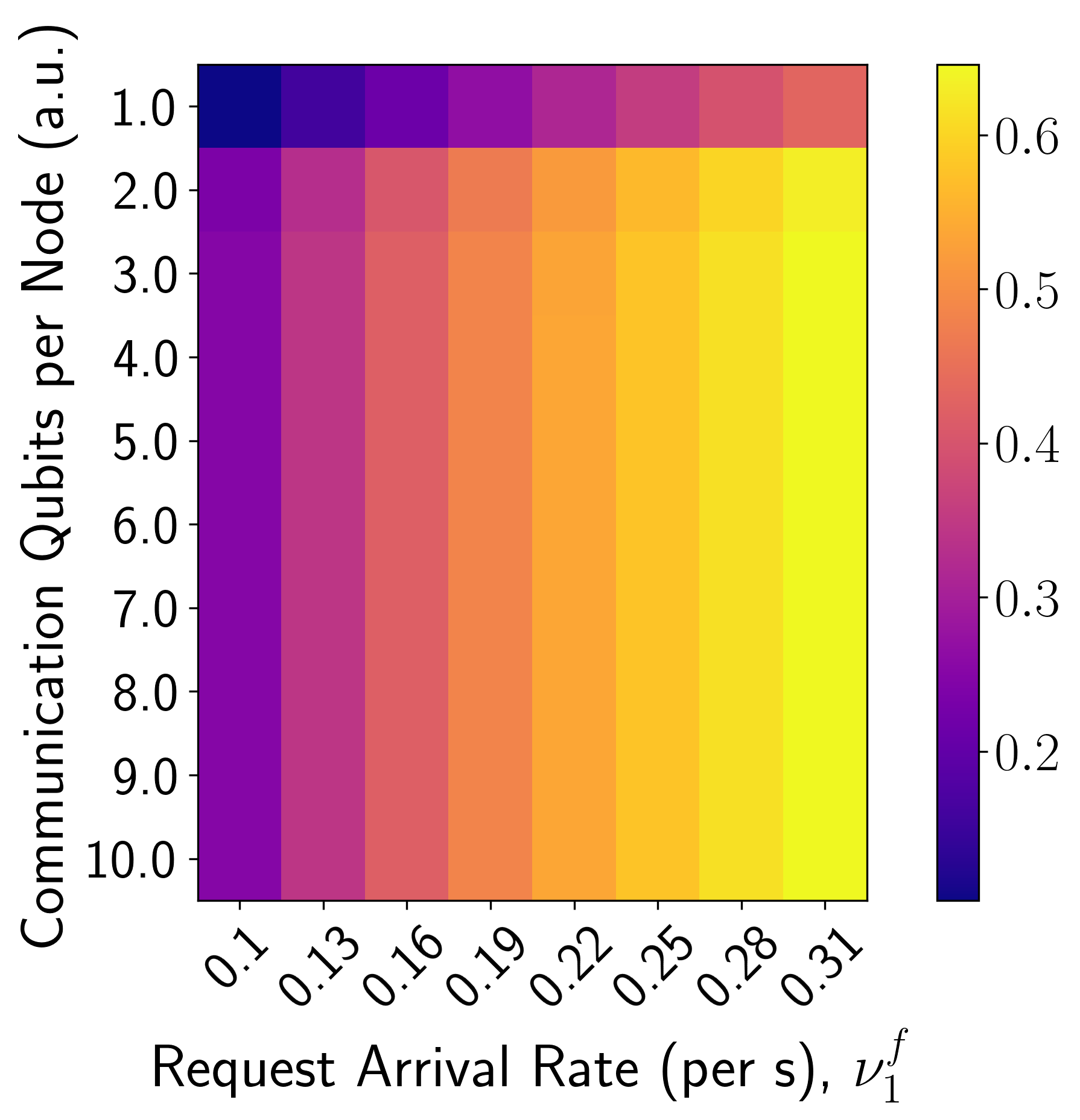}
        \caption{three EGS Resources ($C=3$)}
        \end{subfigure}
        ~
    \caption{Heatmaps of the average blocking probability per flow when the number of communication qubits per node and the request arrival rates are varied. Data results from numeric evaluation of (\ref{eq:prob_C}) for an EGS with 20 nodes and serving $\binom{20}{2}=190$ flows, one for each possible node pairing. Session traffic is homogeneous. }
    \label{fig:extraCommQData}
\end{figure}